\documentclass[12pt,twoside]{report}
\usepackage[retainorgcmds]{IEEEtrantools}
\usepackage{amsthm}
\usepackage{amsmath}
\usepackage{mathrsfs}
\usepackage{amssymb} 
\usepackage{bbm} 
\usepackage{multirow} 
\usepackage{units}
\usepackage{stmaryrd}   
\usepackage{verbatim}
\usepackage{enumerate} 
\usepackage[table]{xcolor}
\usepackage{slashbox} 
\usepackage{ wasysym }
\usepackage{tikz}
\usetikzlibrary{calc}
\usepackage{tabularx}

\newcommand{\hideN}[1]{\null}     

\usepackage{etoolbox}
\makeatletter
\patchcmd{\chapter}{\if@openright\cleardoublepage\else\clearpage\fi}{}{}{}
\makeatother
\usepackage{tocloft}

\usepackage{lipsum}

\makeatletter
\patchcmd{\@IEEEeqnarray}{\relax}{\relax\intertext@}{}{}
\makeatother

\usepackage{xcolor}
\usepackage{hyperref} 
\hypersetup{
    colorlinks,
    linkcolor={blue!80!black},
    citecolor={green!50!black},
    urlcolor={blue!80!black}
}
\usepackage{blindtext}
\usepackage{eso-pic}
\usepackage{ifthen}

\usepackage{graphicx}

\renewcommand

\usepackage[top=2.5cm,bottom=2.5cm,headheight=1cm]{geometry}
\DeclareMathAlphabet\mathbfcal{OMS}{cmsy}{b}{n}
\usepackage{tocloft}
\newlength{\mylen}

\renewcommand{\cftfigpresnum}{\figurename\enspace}
\renewcommand{\cftfigaftersnum}{:}
\usepackage{graphicx}
\graphicspath{{images/}}
\usepackage{caption} 
\usepackage{subcaption}
\usepackage[font=small,labelfont=bf]{caption} 

\usepackage[titletoc]{appendix}

\usepackage{fancyhdr}        
\newcommand\blfootnote[1]{%
  \begingroup
  \renewcommand\thefootnote{}\footnote{#1}%
  \addtocounter{footnote}{-1}%
  \endgroup
}

\title{A Study on Arbitrarily Varying Channels with Causal Side Information at the Encoder 
\blfootnote
{
This work was supported by the Israel Science Foundation (grant No. 1285/16).
}
}
\author{Uzi Pereg and Yossef Steinberg  }

\usepackage[backend=bibtex,sorting=nyt,maxbibnames=10]{biblatex}
\ifdefined\bibstar\else\newcommand{\bibstar}[1]{}\fi
\renewbibmacro{in:}{} 										

\definecolor{light-gray}{gray}{0.8}

\newcounter{parentnumber}

\definecolor{dark-gray}{gray}{0.3}
\usepackage{accents}
\newlength{\dhatheight}

\newcommand{\bieee}{\begin{IEEEeqnarray}{rCl}}
\newcommand{\eieee}{\end{IEEEeqnarray}}
\newcommand{\prob}[1]{\Pr\left(#1\right)}
\newcommand{\given}{\mid}
\newcommand{\cprob}[2]{\Pr\left(#1\given #2\right)}
\newcommand{\E}{\mathbb{E}}
\newcommand{\eps}{\varepsilon}

\newcommand{\ie}{\emph{i.e.} }
\newcommand{\eg}{\emph{e.g.} }
\newcommand{\etal}{\emph{et al.} }
\newcommand{\cf}{\emph{cf.} }

\newcommand{\tm}{\widetilde{m}}	
\newcommand{\tM}{\widetilde{M}}																				

\newcommand{\tp}{\widetilde{p}}

\newcommand{\tsn}{\widetilde{s}^{\, n}}

\newcommand{\tf}{\widetilde{f}}
\newcommand{\tg}{\widetilde{g}}
\newcommand{\tfnu}{\tf^{\nu}}
\newcommand{\gnu}{\tg}
\newcommand{\bR}{\tR}
\newcommand{\oS}{\overline{S}}

\newcommand{\oq}{\overline{q}}

\newcommand{\tR}{\widetilde{R}}

\newcommand{\hm}{\hat{m}}
\newcommand{\hq}{\widehat{q}}

\newcommand{\hgamma}{\widehat{\gamma}}
\newcommand{\hP}{\hat{P}}
\newcommand{\hM}{\hat{M}}

\newcommand{\Aset}{\mathcal{A}}
\newcommand{\Bset}{\mathcal{B}}

\newcommand{\Dset}{\mathcal{D}}

\newcommand{\Uset}{\mathcal{U}}
\newcommand{\Vset}{\mathcal{V}}
\newcommand{\Qset}{\mathcal{Q}}

\newcommand{\Sset}{\mathcal{S}}
\newcommand{\Wset}{\mathcal{W}}
\newcommand{\Xset}{\mathcal{X}}
\newcommand{\Yset}{\mathcal{Y}}
\newcommand{\Zset}{\mathcal{Z}}
\newcommand{\Eset}{\mathcal{E}}
\newcommand{\markovC}[1]{%
\begin{tikzpicture}[#1]%
\draw (0,0.3ex) -- (1ex,0.3ex);%
\draw (0.5ex,0.3ex) circle (0.2ex);
\draw[white] (0.2ex,0) -- (0.5ex,0);%
\end{tikzpicture}%
}
\newcommand{\Cbar}{\markovC{scale=2}}
\newcommand{\Bern}{\text{Bernoulli}}
\newcommand{\withprob}{\text{w.p. }}
\newcommand{\interior}[1]{\text{int}\hspace{-0.01cm}\big( #1 \big)}

\theoremstyle{remark}	\newtheorem{theorem}{Theorem}
\theoremstyle{remark}	\newtheorem{lemma}[theorem]{Lemma}
\theoremstyle{remark}	\newtheorem{coro}[theorem]{Corollary}
\theoremstyle{remark} \newtheorem{definition}{Definition}
\theoremstyle{remark} 
\theoremstyle{remark} \newtheorem{example}{Example}

\newcommand{\channel}{W_{Y|X,S}}
\newcommand{\nchannel}{W_{Y^n|X^n,S^n}} 
\newcommand{\rp}{\Wset^q} 
\newcommand{\rpig}{\rp_0} 
\newcommand{\compound}{\Wset^\Qset} 
\newcommand{\avc}{\Wset}																		
\newcommand{\avcig}{\avc_0}																	
\newcommand{\opC}{\mathbb{C}}																
\newcommand{\inC}{\mathsf{C}}															 	
\newcommand{\inR}{\mathsf{R}}
\newcommand{\Uchannel}{V^{\encs}_{Y|U,S}}														
\newcommand{\xichannel}{\Uchannel}									
\newcommand{\Uavc}{\Vset_0^{\,\xi}}													

\newcommand{\pSpace}{\mathcal{P}}														

\newcommand{\dF}{\mathsf{F}}																
\newcommand{\dE}{\mathsf{E}}																
	
\newcommand{\dM}{\mathsf{M}}															 	
\newcommand{\dK}{k}

\newcommand{\enc}{f}																				
\newcommand{\dec}{g}																			 	
\newcommand{\code}{\mathscr{C}}															
\newcommand{\gcode}{\mathscr{C}^{\,\Gamma}}									
\newcommand{\tcode}{\widetilde{\code}}											

\newcommand{\cerr}{P_{e|s^n}^{(n)}}													
\newcommand{\err}{P_e^{(n)}}															
\newcommand{\cost}{\phi}																		
\newcommand{\plimit}{\Omega}																			
\newcommand{\tset}{\Aset^{\delta}}													
\newcommand{\qn}{q^n}
\newcommand{\tQ}{\hat{\Qset}_n}														

\newcommand{\enci}{f_i}																			
\newcommand{\encn}{f^n}																			
\newcommand{\encs}{\xi}																			

\newcommand{\Cavc}{\opC(\avc)}
\newcommand{\Cavcig}{\opC(\avcig)}
\newcommand{\ICrp}{\inC(\rp)}
\newcommand{\ICrpig}{\inC(\rpig)}

\newcommand{\emp}{\hP}																		  

\newcommand{\rstarC}{																			  
\, \hspace{-0.3cm} \text{ $$ \mbox{  
\hspace{-0.1cm} 
\small $\star$   
} $$ }
\hspace{-0.25cm}}

\newcommand{\rCav}{\opC^{\rstarC}\hspace{-0.1cm}(\avc)}
\newcommand{\rCavig}{\opC^{\rstarC}\hspace{-0.1cm}(\avcig)}

\newcommand{\sCondQ}{\mathscr{T}^{\Qset}}
\newcommand{\sCond}{\mathscr{T}}

\newcommand{\bc}{W_{Y_1,Y_2|X,S}}
\newcommand{\sbc}{W_{Y_1|X,S}}
\newcommand{\wbc}{W_{Y_2|X,S}}
\newcommand{\nBC}{W_{Y_1^n,Y_2^n|X^n,S^n}} 
\newcommand{\tBset}{\Bset}
\newcommand{\Brp}{\tBset^q} 
\newcommand{\BrpS}{\tBset^{q^*}} 
\newcommand{\Bcompound}{\tBset^\Qset} 
\newcommand{\BcompoundP}{\tBset^{\pSpace(\Sset)}} 
\newcommand{\avbc}{\tBset}																		
\newcommand{\Bavcig}{\avbc_0}																	
\newcommand{\BopC}{\mathbb{C}}																
\newcommand{\BinC}{\mathsf{R}_{out}}									
\newcommand{\BinR}{\mathsf{R}_{in}}									

\newcommand{\BrCcompound}{\BopC^{\rstarC}\hspace{-0.1cm}(\Bcompound)}
\newcommand{\BrCcompoundP}{\BopC^{\rstarC}\hspace{-0.1cm}(\BcompoundP)}

\newcommand{\BrCav}{\BopC^{\rstarC}\hspace{-0.1cm}(\avbc)}
\newcommand{\BrCavig}{\BopC^{\rstarC}\hspace{-0.1cm}(\Bavcig)}
\newcommand{\BrICav}{\BinC^{\rstarC}}

\newcommand{\BrIRavig}{\mathsf{R}^{\rstarC}_{0,in}}

\newcommand{\BCrp}{\BopC(\Brp)}
\newcommand{\BCcompound}{\BopC(\Bcompound)}
\newcommand{\BCavc}{\BopC(\avbc)}
\newcommand{\BCavcig}{\BopC(\Bavcig)}
\newcommand{\BICrp}{\mathsf{C}(\Brp)}
\newcommand{\BICrpS}{\inC(\tBset^{q^*})}
\newcommand{\BICcompound}{\BinC(\Bcompound)}
\newcommand{\BIRcompound}{\BinR(\Bcompound)}

\newcommand{\BIRavc}{\BinR^{\rstarC}}

\newcommand{\LambdaOig}{\widetilde{\Lambda}_0} 			
\newcommand{\LambdaO}{\widetilde{\Lambda}}						
\newcommand{\apLSpaceS}{\overline{\pSpace}_\Lambda(\Sset)}			
\newcommand{\pLSpaceS}{\overline{\pSpace}_\Lambda(\Sset)}			
\newcommand{\pLSpaceSn}{\pSpace^n_\Lambda(\Sset^n)}		
\newcommand{\rpLSpaceU}{\pSpace_{\plimit,\Lambda,\encs\;}^{\rstarC}\hspace{-0.1cm}(\Uset)}		  
\newcommand{\pLSpaceU}{\pSpace_{\plimit,
\Lambda,\encs}(\Uset)}		    

\newcommand{\opCli}{\mathbb{C}_{\plimit,\Lambda}}																
\newcommand{\inCli}{\mathsf{C}_{\plimit,\Lambda\,}}															 	

\newcommand{\LcompoundP}{\Wset^{\apLSpaceS}} 

\newcommand{\LCcompound}{\opCli(\compound)}

\newcommand{\LCcompoundP}{\opCli(\LcompoundP)}
\newcommand{\LrCcompoundP}{\opCli^{\rstarC}\hspace{-0.1cm}(\LcompoundP)}
\newcommand{\LCavc}{\opCli\,(\avc)}
\newcommand{\LCavcig}{\opCli(\avcig)}

\newcommand{\LICavc}{\inR_{low,\,\plimit,
\Lambda}(\avc)}
\newcommand{\LICavcig}{\inCli(\avcig)}

\newcommand{\LrCcompound}{\opC_{\plimit,\Lambda \,}^{\rstarC}\hspace{-0.1cm}(\compound)}

\newcommand{\LrCav}{\opCli^{\rstarC}\hspace{-0.1cm}\,(\avc)}
\newcommand{\LrCavig}{\opCli^{\rstarC}\hspace{-0.1cm}(\avcig)}
\newcommand{\LrICav}{\inR^{\rstarC}_{low,\plimit,\Lambda \,}\hspace{-0.1cm}(\avc)}
\newcommand{\LrIRav}{\inR^{\rstarC}_{up,\plimit,\Lambda \,}\hspace{-0.1cm}(\avc)}
\newcommand{\LrICavig}{\inCli^{\rstarC}\hspace{-0.1cm}(\avcig)}

\newcommand{\opCia}{\overline{\mathbb{C}}_{\plimit,\Lambda}}															
\newcommand{\opClia}{\overline{\mathbb{C}}_{\plimit,\Lambda}}																

\newcommand{\ALCcompound}{\opCia(\compound)}

\newcommand{\ALCavc}{\opClia\,(\avc)}

\newcommand{\ALrCcompound}{\opCia^{\rstarC}\hspace{-0.1cm}(\compound)}

\newcommand{\ALrCav}{\opClia^{\rstarC}\hspace{-0.1cm}\,(\avc)}

\begin{document}

\maketitle


{
\hypersetup{
    linkcolor={blue!50!black}
}
\pagestyle{plain}
\tableofcontents
}
\newpage
\pagestyle{fancy}
\begin{abstract} 
 In this work, we study two models of arbitrarily varying channels, when causal side information is available at the encoder in a causal manner.
%
First, we study the arbitrarily varying channel (AVC) with input and state constraints, when the encoder has state information in a causal manner.  
Lower and upper bounds on the random code capacity are developed. 
A lower bound on the deterministic code capacity is established in the case of a message-averaged input constraint.  
%
In the setting where a state constraint is imposed on the jammer, while the user is under no constraints,
 the random code bounds coincide, and the random code capacity is determined. 
Furthermore, for this scenario, 
a generalized non-symmetrizability condition is stated, under which 
 the deterministic code capacity coincides with the random code capacity.

A second model considered in our work is the arbitrarily varying degraded broadcast channel with causal side information at the encoder (without constraints).
We 
 establish inner and outer bounds on both the random code capacity region and the deterministic code capacity region. 
 The capacity region is then determined for a class of channels satisfying a condition on the mutual informations between the strategy variables and the channel outputs.
 As an example, we show that the condition holds for the arbitrarily varying binary symmetric broadcast channel, and we find the corresponding capacity region.

\end{abstract}
%


\newpage
\chapter*{\vspace{-2cm}
Introduction}
\addcontentsline{toc}{chapter}{Introduction}
\pagestyle{fancy}

In practice, the statistics of a communication system are not necessarily known in exact, and they may even change over time.
The arbitrarily varying channel (AVC) is an appropriate model to describe such a situation, as introduced by Blackwell \etal \cite{BBT:60p}.
 Among the motivations for this field of research is the adversarial communication model, 
 where  a \emph{jammer} selects a sequence of channel states in an attempt to disrupt communication. 

Considering the AVC without SI, 
Blackwell \etal  determined  the random code channel capacity  \cite{BBT:60p}, \ie the capacity achieved by stochastic-encoder stochastic-decoder coding schemes with common randomness. It was also demonstrated in  \cite{BBT:60p}  that the random code capacity is not necessarily  achievable using deterministic codes. 
  A well-known result by Ahlswede \cite{Ahlswede:78p} is the dichotomy property presented by the AVC in the absence of state information. Namely, without SI, the deterministic code capacity either equals the random code capacity or else, it is zero.

Subsequently, Ericson \cite{Ericson:85p} and Csisz{\'{a}}r and Narayan \cite{CsiszarNarayan:88p}
 have established a simple single-letter condition, namely non-symmetrizability, which is both necessary and sufficient for the capacity to be positive in the case of an AVC without state information.  
The derivation of sufficiency, in \cite{CsiszarNarayan:88p}, is independent of Ahlswede's work and is based on a subtle decoding rule, analyzed through the method of types. 

Csisz{\'{a}}r and Narayan also determined the random code capacity \cite{CsiszarNarayan:88p1} and the deterministic code capacity \cite{CsiszarNarayan:88p} of the AVC, when input and state constraints are imposed on the user and the jammer, respectively.  In \cite{CsiszarNarayan:88p}, they show that dichotomy in the notion of \cite{Ahlswede:78p} does not hold when state constraints are imposed on the jammer. That is, the deterministic code capacity can be lower than the random code capacity, and yet non-zero.

Vast research has been conducted on other AVC models as well. Recently, the arbitrarily varying wiretap channel
has been extensively studied,  as \eg in 
\cite{MBL:09c,BocheShaefer:13p,ACD:13b,BocheShaeferPoor:14c,NotzelWieseBoche:16p,GoldfeldCuffPermuter:16a}.
The multiple user scenario was first studied by Jahn \cite{Jahn:81p}, who presented an inner bound on the capacity region of the arbitrarily varying broadcast channel.  
 More recent results on the arbitrarily  varying broadcast channel are derived \eg in  \cite{WinshtokSteinberg:06c,HofBross:06p}.

Additional models of interest involve SI available at the encoder.
%
%
In \cite{Ahlswede:86p}, Ahlswede addressed the AVC with non-causal SI available at the encoder, also referred to as the arbitrarily varying Gel'fand-Pinsker model \cite{GelfandPinsker:80p}.  The analysis relies on a technique that Ahlswede developed, which is referred to as Ahlswede's Robustification Technique \cite{Ahlswede:80p,Ahlswede:86p}.
This technique was then used in  \cite{WinshtokSteinberg:06c}, to establish
 the capacity region of the  arbitrarily varying degraded broadcast channel with non-causal SI at the encoder.
 The AVC with causal SI is addressed in the book by Csisz{\'a}r and K{\"o}rner \cite{CsiszarKorner:82b}, while their approach is independent of Ahlswede's work.
 A straightforward application of Ahlswede's Robustification Technique 
 fails to comply with the causality requirement. 

 In this work, we study two models,
analyzed using a modified version of Ahlswede's
Robustification and Elimination Techniques \cite{Ahlswede:78p,Ahlswede:79p,Ahlswede:80p,
Ahlswede:86p}. 
In particular, we adjust Ahlswede's Robustification Technique, previously used for the case of  non-causal SI, such that it would be applicable in the case of causal SI. 

The first model considered in this work is the AVC with input and state constraints when causal SI is available at the encoder. 
We find lower and upper bounds on the random code capacity. Furthermore we find a lower bound on the deterministic code capacity, for an input constraint that is averaged over the messages.
For the case where a state constraint is imposed on the jammer, while the user is under no constraints,
the random code bounds coincide, and the random code capacity is determined.
 In this scenario, 
a generalized non-symmetrizability condition is stated, under which 
 the deterministic code capacity coincides with the random code capacity.  

The second model considered in this work is the arbitrarily varying degraded broadcast channel with causal SI at the encoder (without constraints). 
Inner and outer bounds on the random code capacity region and the deterministic code capacity region are established. 
Specifically,  Jahn's inner bound \cite{Jahn:81p} and the dichotomy property are extended to the case where causal SI is available. 
 Furthermore, we find an outer bound, and  conditions on the broadcast channel under which the inner and outer bounds coincide and the capacity region is determined.  As an example, we show that the condition holds for the arbitrarily varying binary symmetric broadcast channel, and we find the corresponding capacity region.

The manuscript is divided into two main parts. 
 In Chapter~\ref{chap:avcC}, we treat the AVC with causal SI in the presence of input and state constraints.
In Chapter~\ref{chap:AVDBC}, we treat the arbitrarily varying degraded broadcast channel with causal SI
 (without constraints). 

\vspace{-1cm}
\chapter{Causal Side Information and Constraints}
\label{chap:avcC}
In this chapter, we address the arbitrarily varying channel 
 with causal side information available at the encoder, under input and state constraints. 

\section{Definitions and Previous Results}
\label{sec:notation}
\subsection{Notation}
We use the following notation conventions throughout. 
Calligraphic letters $\Xset,\Sset,\Yset,...$ are used for finite sets.
Lowercase letters $x,s,y,\ldots$  stand for constants and values of random variables, and uppercase letters $X,S,Y,\ldots$ stand for random variables.  
 The distribution of a random variable $X$ is specified by a probability mass function (pmf) 
	$P_X(x)=p(x)$ over a finite set $\Xset$. Let $\pSpace(\Xset)$ denote the set of all pmfs over $\Xset$.
		
 We use $x^j=(x_1,x_{2},\ldots,x_j)$ to denote  a constant sequence, with $j\geq 1$. 
For a pair of integers $i$ and $j$, $1\leq i\leq j$, we define the discrete interval $[i:j]=\{i,i+1,\ldots,j \}$.  
 A random sequence $X^n$ and its distribution $P_{X^n}(x^n)=p^n(x^n)$ are defined accordingly.

	\subsection{Channel Description}
	\label{subsec:channels}
 A state-dependent discrete memoryless channel (DMC) 
$(\Xset\times\Sset,\channel,\Yset)$ consists of  finite input, state and output alphabets $\Xset$, $\Sset$, $\Yset$, respectively, 
 and a collection of conditional pmfs $p(y|x,s)$ over $\Yset$. The channel is memoryless without feedback, and therefore   $p(y^n|x^n,s^n)= \prod_{i=1}^n \channel(y_i|x_i,s_i)$. 
The AVC is a DMC $\channel$  with a state sequence of unknown distribution,  not necessarily independent nor stationary. That is, $S^n\sim \qn(s^n)$ with an unknown joint pmf $\qn(s^n)$ over $\Sset^n$. In particular, $\qn(s^n)$ can give mass $1$ to some state sequence $s^n$.
For state-dependent channels with causal SI, the channel input at time $i\in[1:n]$ may depend on the sequence of past and present states $s^i$. 
The AVC with causal SI is denoted by $\avc=\{\channel\}$.

The compound channel is used as a tool in the analysis. 
 Different models of compound channels are described in the literature. Here,  
 the compound channel is  a DMC with a discrete memoryless state, where the state distribution $q(s)$ is not known in exact, but rather belongs to a family of distributions $\Qset$, with $\Qset\subseteq \pSpace(\Sset)$. That is, the state sequence $S^n$ is independent and identically distributed (i.i.d.) according to $ q(s)$, for some pmf $q\in\Qset$. 
We note that this differs from the classical definition of the compound channel, as in \cite{CsiszarKorner:82b}, where the state is fixed throughout the transmission.
The compound channel with causal SI is denoted by $\compound$.

\subsection{Coding}
We introduce some preliminary definitions, starting with the definitions of a deterministic code and a random code for the AVC $\avc$ with  causal SI. Note that in general,
 the term `a code', unless mentioned otherwise, refers to a deterministic code.

\begin{definition}[Code] 
\label{def:capacity}
A $(2^{nR},n)$ code for the AVC $\avc$ with causal SI consists of the following;   
a message set $[1:2^{nR}]$, 
where it is assumed throughout that $2^{nR}$ is an integer,
a set of $n$ encoding functions 
$\enci:  [1:2^{nR}]\times \Sset^i \rightarrow \Xset$, for  
$ i\in [1:n]$,  
 and a decoding function
$
\dec: \Yset^n\rightarrow [1:2^{nR}]  
$. 

At time $i\in [1:n]$, given a message $m\in [1:2^{nR}]$ and a sequence $s^i$,
 the encoder transmits $x_i=\enci(m,s^i)$. The codeword is then given by 
\bieee
x^n= \encn(m,s^n) \triangleq \left(
\enc_1(m,s_1),\enc_2(m,s^2),\ldots,\enc_n(m,s^n)   \right) \;.
\eieee
The decoder  receives the channel output $y^n$, and finds an estimate of the message $\hm=g(y^n)$.  
We denote the code by $\code=\left(\encn(\cdot,\cdot),\dec(\cdot) \right)$.
\end{definition}

%

We proceed now to coding schemes 
when using stochastic-encoder stochastic-decoder pairs with common randomness.
The codes formed by these pairs are referred to as random codes, a.k.a. correlated codes \cite{Ahlswede:86p}. 

\begin{definition}[Random code]
\label{def:corrC} 
A $(2^{nR},n)$ random code for the AVC $\avc$ consists of a collection of 
$(2^{nR},n)$ codes $\{\code_{\gamma}=(\encn_\gamma,\dec_\gamma)\}_{\gamma\in\Gamma}$, along with a probability distribution $\mu(\gamma)$ over the code collection $\Gamma$. 
We denote such a code by $\gcode=(\mu,\Gamma,\{\code_{\gamma}\}_{\gamma\in\Gamma})$.

\end{definition}

Next, we write the definition of Shannon strategy coding with causal SI  \cite{Shannon:58p}. 
 Though, we  use a  different formulation, as \eg in \cite{ElGamalKim:11b}
 (see \cite[Remark 7.8]{ElGamalKim:11b}). 
\begin{definition}[Shannon strategy code]  \cite{Shannon:58p}
\label{def:StratCode}
A $(2^{nR},n)$ Shannon strategy code for the AVC $\avc$ with causal SI is a $(2^{nR},n)$  
 code 
with an encoder that is composed of  
an encoding strategy sequence
$
u^n: [1:2^{nR}] \rightarrow \Uset^n 
$, 
 an encoding function 
 $\encs:\Uset\times\Sset\rightarrow \Xset$,  and 
a decoding function
$
\dec: \Yset^n\rightarrow [1:2^{nR}] 
$. The codeword is then given by 
\bieee
\label{eq:StratEnc}
x^n=\encs^n(u^n(m),s^n)\triangleq \big[\, \encs(u_i(m),s_i) \,\big]_{i=1}^n \;.
\eieee
We denote such a code by $\code=\left(u^n(\cdot),\encs(\cdot,\cdot),\dec(\cdot) \right)$.
\end{definition}
The definitions above apply to the compound channel $\compound$ as well.
\subsection{Input and  State Constraints} 
Next, we consider input and state constraints.
Let $\cost:\Xset\rightarrow [0,\infty)$ and $l:\Sset\rightarrow [0,\infty)$ be some given bounded functions, and define
	\bieee
	\cost^n(x^n)&=&\frac{1}{n} \sum_{i=1}^n \cost(x_i) \;,
	\label{eq:LInConstraintStrict} \\
	l^n(s^n)&=&\frac{1}{n} \sum_{i=1}^n l(s_i) \;.
	\eieee
Let $\plimit>0$ and $\Lambda>0$. Below, we specify input constraint $\plimit$ and state constraint $\Lambda$, corresponding to  the functions
$\cost^n(x^n)$ and $l^n(s^n)$, respectively, for the AVC and the compound channel with causal SI.

		We may assume without loss of generality that $0\leq \plimit\leq \cost_{max}$ and
 $0\leq\Lambda\leq l_{max}$, where $\cost_{max}=\max_{x\in\Xset} \cost(x)$ and $l_{max}=\max_{s\in\Sset} l(s)$. It is also assumed that for some $x_0\in\Xset$ and $s_0\in\Sset$, $\cost(x_0)=l(s_0)=0$.		
	%
	%
	\subsubsection{State Constraints}
	State constraints are imposed on the compound channel $\compound$ and the AVC $\avc$ with causal SI, as specified below.
	Given some $\Lambda>0$, define 
	a set of constrained single-letter state distributions, 
\begin{align}
\pLSpaceS&\triangleq 
\{ q(s)\in\pSpace(\Sset) \,:\; \E_q\, l(S)\leq\Lambda \} \;,
\label{eq:aplspaceS}
\intertext{and 
a set of constrained $n$-fold state distributions, }
\pLSpaceSn
&\triangleq\{ q^n(s^n) \in\pSpace^n(\Sset^n) \,:\; q^n(s^n)=0 \;\text{ if $l^n(s^n)>\Lambda$}\, \} \;.
\end{align}
The set $\pLSpaceS$ represents a state constraint \emph{on average}, whereas the set $\pLSpaceSn$ represents a state constraint held \emph{almost surely}.

We say that a compound channel $\compound$ with causal SI is under a state constraint $\Lambda$, if the set $\Qset$ of state distributions is limited to 
\begin{align}
&\Qset \subseteq \pLSpaceS \;.
\intertext{ 	
%
As for the AVC $\avc$ with causal SI, 
 it is now assumed that $l^n(S^n)\leq\Lambda$ \withprob $1$, \ie
}
\label{eq:StateCn}
&q^n(s^n)\in\pLSpaceSn \;.
\end{align}
	
\subsubsection{Input Constraints}
 Consider the AVC $\avc$ with causal SI, under an input constraint as specified below. 
 Attention should  be drawn to the fact that, when SI is available and the channel input depends on the state sequence $S^n$,  the input cost depends on the jammer's strategy $q^n(s^n)$ as well.


We consider two types of input constraints. 
We say that the AVC $\avc$ with causal SI is under \emph{per message input constraint} $\plimit$, if
\begin{subequations}
\label{eq:inputC}
\bieee
&&
 \sum_{s^n\in\Sset^n}q^n(s^n)\cost^n (\enc^n(m,s^n))\leq \plimit \;, \IEEEnonumber\\ 
&&\text{for all $m\in [1:2^{nR}]$ and $q^n(s^n)\in\pLSpaceSn$
} \;.\quad
\label{eq:inputCstrict}
\eieee
As for the second type,  
we say that the AVC $\avc$ with causal SI is under
 \emph{average input constraint} $\plimit$, if
\bieee
&& \frac{1}{2^{nR}}\sum_{m=1}^{2^{nR}}
 \sum_{s^n\in\Sset^n}q^n(s^n)\cost^n (\enc^n(m,s^n))\leq \plimit \;, \IEEEnonumber\\ 
&&\text{for all $q^n(s^n)\in\pLSpaceSn$
} \;.\quad
\label{eq:inputCaverage}
\eieee
\end{subequations}
Input constraint on the compound channel $\compound$ with causal SI is defined in a similar manner, 
where  (\ref{eq:inputC}) is taken with respect to i.i.d. state distributions $q^n(s^n)=\prod_{i=1}^n q(s_i)$, with $q\in\Qset$.

\subsection{Capacity under Constraints}
We move to the definition of  an achievable rate and the capacity of the AVC $\avc$ with causal SI, under input and state constraints.
Deterministic codes and random codes over the AVC $\avc$ with causal SI are defined as in 
Definition~\ref{def:capacity} and Definition~\ref{def:corrC},
 respectively, with the additional constraint 
 (\ref{eq:inputCstrict}) or (\ref{eq:inputCaverage}) 
 on the codebook.

 Define the conditional probability of error of a code $\code$ given a state sequence $s^n\in\Sset^n$ by  
\begin{subequations}
\begin{align}
\label{eq:cerr}
&\cerr(\code)\triangleq 
\frac{1}{2^{ nR }}\sum_{m=1}^{2^ {nR}}
\sum_{y^n:\dec(y^n)\neq m} \nchannel(y^n|\encn(m,s^n),s^n) \;,
\end{align}
where $\nchannel(y^n|x^n,s^n)=\prod_{i=1}^n \channel(y_i|x_i,s_i)$. 
Now, define the average probability of error of $\code$ for some distribution $\qn(s^n)\in\pSpace^n(\Sset^n)$, 
\bieee
\err(\qn,\code)\triangleq 
\sum_{s^n\in\Sset^n} \qn(s^n)\cdot\cerr(\code) \;.
\eieee
\end{subequations}

\begin{definition}[
Achievable rate and capacity under constraints]
\label{def:Lcapacity}
A code $\code=(\encn,\dec)$ is a  called a
$(2^{nR},n,\eps)$ code for the AVC $\avc$, under per message  input constraint $\plimit$ and  state constraint $\Lambda$, 
 when (\ref{eq:inputCstrict})  is satisfied 
 and 
\begin{align}
\label{eq:Lerr}
& \err(q^n,\code) 
\leq \eps \;,\quad
\text{for all $q^n\in\pLSpaceSn$} \;.
\end{align}

  We say that a rate $R$ is achievable under per message
	input constraint $\plimit$ and state constraint $\Lambda$,
	if for every $\eps>0$ and sufficiently large $n$, there exists a  $(2^{nR},n,\eps)$ code for the AVC
	$\avc$ under per message
	input constraint $\plimit$ and state constraint $\Lambda$. The operational capacity is defined as the supremum of all achievable rates, 
	and it is denoted by $\LCavc$. 
 We use the term `capacity' referring to this operational
meaning, and in some places we call it the deterministic code capacity in order to emphasize that achievability is measured with respect to  deterministic codes.

Analogously to the deterministic case,  a $(2^{nR},n,\eps)$ random code $\gcode=$ $(\mu,\Gamma,$ $\{\code_{\gamma}\}_{\gamma\in\Gamma})$
 for the AVC $\avc$, under per message
 input constraint $\plimit$ and state constraint $\Lambda$, satisfies the requirements
\begin{subequations}
\label{eq:LrcodeReq}
\begin{align}
&
\sum_{\gamma\in\Gamma}\mu(\gamma) \left[
 \sum_{s^n\in\Sset^n} q^n(s^n)  \cost^n(\enc_\gamma^n(m,s^n))
\right]
  \leq \plimit 
\,,\; 
\text{for all $m\in [1:2^{nR}]$ 
 \,,\;
 $q^n\in\pLSpaceSn$}\;,  \label{eq:codeInputCr}
\intertext{and} 
&\err(q^n,\gcode)\triangleq \sum_{\gamma\in\Gamma} \mu(\gamma) \left[ \sum_{s\in\Sset} q^n(s^n)\cdot \cerr(\code_\gamma) 
\right]
\leq \eps \;,\quad\text{for all $q^n\in\pLSpaceSn$} \;.
\label{eq:Lrerr}				
\end{align}
\end{subequations}
The capacity achieved by random codes is then denoted by $\LrCav$, and it 
 is referred to as the \emph{random code capacity}.
\end{definition}


The definitions above are naturally extended to the compound channel under per message
 input constraint $\plimit$ and state constraint $\Lambda$, by relaxing the requirements (\ref{eq:inputCstrict}), 
(\ref{eq:Lerr}) and (\ref{eq:LrcodeReq}) to i.i.d. state distributions $q\in\Qset$. 
The respective deterministic code capacity and random code capacity,  
 $\LCcompound$ and $\LrCcompound$, are defined accordingly.
%
Furthermore, similar definitions apply to the average input constraint, taking an average over the messages, as in (\ref{eq:inputCaverage}). 
 Hence, the deterministic code capacities $\ALCavc$, $\ALCcompound$ and the random code capacities $\ALrCav$, $\ALrCcompound$ are defined accordingly. 

\subsection{In the Absence of Side Information } 
\label{sec:LNOsi}
In this subsection, we briefly review known results for the case where the state is not known to the encoder or the decoder, \ie SI is not available. For the sake of brevity, we skip the compound channel. Then, consider an AVC without SI, which we denote by $\avcig$.

\subsubsection{Without Constraints}
We begin with the case where there are no constraints, \ie
$\plimit=\cost_{max}$ and 
$\Lambda=l_{max}$.  
Then, the subscript `$\plimit,\Lambda$' in the capacity notation is not necessary, and thus omitted.

We cite the random code capacity theorem of the AVC without SI, free of constraints, which was first introduced  by Blackwell \etal \cite{BBT:60p}. Let 
\bieee
\label{eq:minimax}
\inC^{\rstarC}\hspace{-0.1cm}(\avcig)\triangleq \max_{p(x)}\min_{q(s)} I_q(X;Y)=\min_{q(s)}\max_{p(x)} I_q(X;Y) \;.
\eieee
\begin{theorem} \cite{BBT:60p} 
\label{theo:avcC0R}
The random code capacity of an AVC $\avcig$ without SI, free of constraints,  is  given by 
\bieee
\rCavig = \inC^{\rstarC}\hspace{-0.1cm}(\avcig) \;.
\eieee
\end{theorem}
We note that the expression in (\ref{eq:minimax}) has a game-theoretic minimax interpretation \cite{BBT:60p,BMM:85p,HST:89p,McEliece:83b}. 
%
Now, a well-known result by Ahlswede \cite{Ahlswede:78p} says that 
the deterministic code capacity $\Cavcig$ is characterized by the following dichotomy.  
\begin{theorem}[Ahlswede's Dichotomy] \cite{Ahlswede:78p} 
\label{theo:avcC0}
The capacity of an AVC $\avcig$ without SI, free of  constraints, either coincides with the random code capacity or else, it is zero.
That is, 
$
\Cavcig = \rCavig  
$ or else, 
$
 \Cavcig=0 
$. 
\end{theorem}
A necessary and sufficient condition for a positive capacity was 
 established by Ericson \cite{Ericson:85p} and Csisz{\'{a}}r and Narayan   \cite{CsiszarNarayan:88p}, in terms of the following definition.
\begin{definition}
\label{def:symmetrizable}
 A state-dependent DMC $\channel$ is said to be \emph{symmetrizable} if for some conditional distribution $J(s|x)$,
\begin{align}
\label{eq:symmetrizable}
\sum_{s\in\Sset} \channel(y|x_1,s)J(s|x_2)=\sum_{s\in\Sset} \channel&(y|x_2,s)J(s|x_1) \,,\; \nonumber\\
&\forall\, x_1,x_2\in\Xset \,,\; y\in\Yset \;.
\end{align}
Equivalently, the channel $\widetilde{W}(y|x_1,x_2)$ $=$ $
\sum_{s\in\Sset} \channel(y|x_1,s)J(s|x_2)$ is symmetric, \ie $\widetilde{W}(y|x_1,x_2)=\widetilde{W}(y|x_2,x_1)$, for all $x_1,x_2\in\Xset$ and $y\in\Yset$. We say that such a  $J:\Xset\rightarrow\Sset$ symmetrizes $\channel$. 
  We say that the AVC $\avcig$ 
 is symmetrizable if the corresponding state-dependent DMC $\channel$ is symmetrizable.  
\end{definition}
 %
\begin{theorem} \cite{Ericson:85p,CsiszarNarayan:88p}
\label{theo:symm0}
An AVC $\avcig$ without SI, free of constraints, has a positive capacity $\Cavcig>0$ if and only if it is \emph{not} symmetrizable. 
\end{theorem}

\subsubsection{Under Constraints}
%
Csisz{\'{a}}r and Narayan addressed the AVC $\avcig$ without SI under constraints in
 \cite{CsiszarNarayan:88p1} and \cite{CsiszarNarayan:88p}. The focus here is on the case of per message input constraint, although their results apply to the average case as well.
Let
\bieee
\LrICavig &\triangleq& 
 \min_{q(s)\in\pLSpaceS}
\; \max_{p(x)\,:\; \E\, \cost(X)\leq\plimit} I_q(X;Y) \;,
\label{eq:Lminimax}
\eieee   
where $\pLSpaceS$ is  defined in (\ref{eq:aplspaceS}).

\begin{theorem} \cite{CsiszarNarayan:88p1}
\label{theo:LavcC0R}
The random code capacity of an AVC $\avcig$ without SI, under per message
 input constraint $\plimit$ and state constraint $\Lambda$,  is  given by 
\bieee
\LrCavig = \LrICavig \;.
\eieee
\end{theorem}
 %

As for the deterministic code capacity, dichotomy in the 
classical notion of \cite{Ahlswede:78p}
 no longer holds when a state constraint $\Lambda<l_{max}$ is imposed on the jammer  \cite{CsiszarNarayan:88p}.
 That is, the capacity of the AVC $\avcig$ can be strictly lower than the random code capacity, and yet non-zero.

For every $p\in\pSpace(\Xset)$ with $\E \cost(X)\leq\plimit$, let
\bieee
\LambdaOig(p)=\min\, \sum_{x\in\Xset}\sum_{s\in\Sset} p(x)J(s|x)l(s) \;,
\eieee
where the minimization is over all conditional distributions $J(s|x)$ that symmetrize $\channel$ (see Definition~\ref{def:symmetrizable}). We use the convention that a minimum value over an empty set is $+\infty$. 
Assume that $\max\limits_{p(x)\,:\; \E \cost(X)\leq\plimit} \LambdaOig(p)\neq \Lambda$. 

Then, define $\LICavcig$ as follows, 
\begin{subequations}
\label{eq:LICavcig}
\begin{align}
\LICavcig&\triangleq 0 \,,\;\text{if $\max_{p(x)\,:\; \E \cost(X)\leq\plimit} \LambdaOig(p)< \Lambda$}\;,
\intertext{and}
\LICavcig&\triangleq \max_{p(x)\,:\; \E\,\cost(X)\leq\plimit \,,\; \LambdaOig(p)\geq \Lambda} \,\min_{q(s)\,:\;\E_q\, l(S)\leq\Lambda} I_q(X;Y)>0 \,,\; 
\nonumber\\
&\text{if $\max_{p(x)\,:\; \E \cost(X)\leq\plimit} \LambdaOig(p)> \Lambda$}\;.
\end{align}
\end{subequations}
\begin{theorem}\cite{CsiszarNarayan:88p1}
\label{theo:LavcC0stateC}
The capacity of an AVC $\avcig$ without SI, under per message
 input constraint $\plimit$ and state constraint $\Lambda$, is given by
\bieee
\LCavcig=\LICavcig \,,\;\text{if $\max_{p(x)\,:\; \E \cost(X)\leq\plimit} \LambdaOig(p)\neq \Lambda$}\;.
\eieee
In particular, if $\avcig$ is non-symmetrizable,
$
\LCavcig=\LrICavig 
$. 
\end{theorem}

\subsection{In The Presence of Side Information}
In this subsection, we briefly review known results for the case where the state is known to the encoder, and no constraints are imposed.
The compound channel and the AVC with non-causal SI, free of constraints, were addressed by Ahlswede in \cite{Ahlswede:86p}.

The AVC with causal SI, free of constraints, was addressed in the problem set of the book by Csisz{\'a}r and K{\"o}rner
\cite[Problem 12.18, part (b)]{CsiszarKorner:82b}. The corresponding results are stated below.
Let
\bieee 
\label{eq:cvCIcoro}
\inC^{\rstarC}\hspace{-0.1cm}(\avc) \triangleq
\min_{q\in\pSpace(\Sset)} \max_{p(u),\encs(u,s)}  I_q(U;Y)   \;,
\eieee 
subject to $X=\encs(U,S)$, where $U$ is an auxiliary random variable, independent of $S$, and the maximization is over the pmf $p(u)$ and the set of all functions $\encs:\Uset\times\Sset\rightarrow\Xset$. 
\begin{theorem} \cite{CsiszarKorner:82b}
\label{theo:avcCr}
The random code capacity of the AVC $\avc$ with causal SI available at the encoder, free of constraints, is given by 
\bieee
\label{eq:avcC} 
\rCav&=& \inC^{\rstarC}\hspace{-0.1cm}(\avc)  \;.
\eieee
\end{theorem}

\begin{theorem}\cite{CsiszarKorner:82b}
\label{theo:corrTOdetC}
The capacity of an AVC $\avc$ with causal SI at the encoder, free of constraints,  either coincides with the random code capacity or else, it is zero.
That is, 
$
 \Cavc = \rCav 
$ 
or else, 
$
\Cavc=0 
$. 
\end{theorem}

This completes our review of previous work, where SI and constraints were considered in separate. Next, we give our results, concerning the combined setting, where SI is available and constraints are imposed.

\section{Results}
\subsection{The Compound Channel with Causal SI}   
\label{sec:Lcompound}
We present a lower bound on the capacity of the compound channel with causal SI, under per message 
 input  constraint $\plimit$, taking the set of state distributions to be $\Qset=\pLSpaceS$.
 For a given mapping $\encs:\Uset\times\Sset\rightarrow\Xset$,
let
\begin{align}
& \rpLSpaceU\triangleq
\bigg\{
p\in\pSpace(\Uset) \,:\; \E_q\, \cost(\encs(U,S))
\leq \plimit \,,\;\text{for all $q\in\apLSpaceS$}
\bigg\} \;, 
\label{eq:rpLSpaceU}
\end{align}
where $(U,S)\sim p(u)\cdot q(s)$.
Then, define
\begin{align}
& \LrICav \triangleq
 \min_{q(s)\in\apLSpaceS} \;
\max_{
\substack{
\encs:\Uset\times\Sset\rightarrow\Xset
\;, \\ p(u)\in \rpLSpaceU }} I_q(U;Y)   \;,
\label{eq:ALcvCIcoro}
\intertext{and}
& \LrIRav \triangleq
 \min_{q(s)\in\apLSpaceS} \;  \max_{
\substack{
\encs:\Uset\times\Sset\rightarrow\Xset
\;, \\ p(u) \,:\; \E_q\, \cost(\encs(U,S))\leq\plimit  }
} I_q(U;Y)   \;.
\label{eq:ALcvRIcoro}
\end{align}
Observe that $\LrICav\leq\LrIRav$, since the maximization constraint in (\ref{eq:ALcvCIcoro}) is taken for all $q\in\pLSpaceS$ (see (\ref{eq:rpLSpaceU})), while the maximization constraint in (\ref{eq:ALcvRIcoro}) is taken for a particular $q\in\pLSpaceS$.

\begin{lemma}
\label{lemm:ALCompoundPC}
Let  $\LcompoundP$ be a compound channel with causal SI available at the encoder, under per message
 input constraint $\plimit$ and state constraint $\Lambda$.
The random code capacity and the deterministic code capacity of 
 $\LcompoundP$ 
are bounded by
\begin{subequations}
\begin{align}
\LCcompoundP&\geq \LrICav \;, \\
\LrCcompoundP&\leq\LrIRav
\;.
\end{align}
\end{subequations}
Furthermore, if 
 $R<\LrICav$, then for some $a>0$ and sufficiently large $n$, there exists a $(2^{nR},n,e^{-a n})$ Shannon strategy code
 over $\LcompoundP$, 
under per message input  constraint $\plimit$. 
\end{lemma}
The proof of Lemma~\ref{lemm:ALCompoundPC} is given in Appendix~\ref{subsec:ALCompoundPC}. It can further be shown that if 
$\LCcompoundP>0$, then $\LrCcompoundP=\LCcompoundP=\LrIRav$. However, this will not be needed here.

\subsection{The AVC with Causal SI}
\subsubsection{Random Code Capacity}
We give lower and upper bounds on the random code capacity of 
 the AVC $\avc$ with causal SI under input and state constraints. 

We begin with a lemma, which 
 is a restatement of Ahlswede's Robustification Technique (RT) \cite{Ahlswede:86p} with some modification. 
\begin{lemma}[Ahlswede's RT] \cite{Ahlswede:86p}
\label{lemm:LRT}
Let $h:\Sset^n\rightarrow [0,1]$ be a given function. If, for some fixed $\alpha_n\in(0,1)$, and for all 
$ \qn(s^n)=\prod_{i=1}^n q(s_i)$, with 
$q\in\apLSpaceS$, 
\bieee
\label{eq:RTcondCs}
\sum_{s^n\in\Sset^n} \qn(s^n)h(s^n)\leq \alpha_n \;,
\eieee
then,
\bieee
\label{eq:RTresCs}
\frac{1}{n!} \sum_{\pi\in\Pi_n} h(\pi s^n)\leq \beta_n \;,\quad\text{for all $s^n\in\Sset^n$ such that $l^n(s^n)\leq\Lambda$} \;,
\eieee
where $\Pi_n$ is the set of all $n$-tuple permutations $\pi:\Sset^n\rightarrow\Sset^n$, and 
$\beta_n=(n+1)^{|\Sset|}\cdot\alpha_n$. 
\end{lemma}
Originally, Ahlswede's RT is stated so that (\ref{eq:RTcondCs}) holds for any $q(s)\in\pSpace(\Sset)$, without state constraint (see \cite{Ahlswede:86p}), but the claim holds also when state constraints are imposed, as here. For completeness, we give the proof of 
Lemma~\ref{lemm:LRT} in Appendix~\ref{app:LRT}.

\begin{theorem}
\label{theo:ALavcCr}
Let $\avc$ be an AVC with causal SI available at the encoder, under per message 
 input  constraint $\plimit$ and state constraint $\Lambda$. Then, 
\begin{enumerate}[1)]
\item
the random code capacity of $\avc$  is bounded by 
\begin{align}
\LrICav\leq \LrCav\leq \LrIRav \;,
\end{align}
where $\LrICav$ and $\LrIRav$ are given by (\ref{eq:ALcvCIcoro}) and  (\ref{eq:ALcvRIcoro}), respectively. 
\item
For $\plimit=\cost_{max}$, \ie when free of input constraints, the random code capacity of $\avc$ is given by
\begin{align} 
 \LrCav=\LrICav= \LrIRav \;.
\end{align}
\end{enumerate}
\end{theorem}
 Theorem~\ref{theo:ALavcCr} is proved in Appendix~\ref{app:ALavcCr}.
We further note that the result above holds when the input constraint is averaged over the message set as well.
The following lemma is the counterpart of a result from \cite{Ahlswede:78p}, stating that a polynomial size of the code collection $\{\code_\gamma\}$ is sufficient. This result is a key observation in Ahlswede's Elimination Technique (ET), presented in \cite{Ahlswede:78p}, where it is used as a basis for the deterministic code analysis. Here, it will 
 be used to determine a condition under which the deterministic code capacity is identical to the random code capacity of the AVC with causal SI under a state constraint.    
\begin{lemma} 
\label{lemm:LcorrSizeC}  
Let $R<\LrCav$. Consider a given  
 $(2^{nR},n,\eps_n)$ random 
code $\code^\Gamma=(\mu,\Gamma,\{\code_\gamma\}_{\gamma\in\Gamma})$
 for the AVC $\avc$ with causal SI, under per message input constraint $\plimit$ and state constraint $\Lambda$,
%
where $\lim_{n\rightarrow\infty} \eps_n=0$. 
Then, for every $\delta>0$, $0<\alpha<1$, and sufficiently large $n$, there exists a $(2^{nR},n)$ random 
 code $(\mu^*,\Gamma^*,\{\code_{\gamma}
\}_{\gamma\in\Gamma^*})$ 
such that for all $m\in [1:2^{nR}]$ and $q^n\in\pLSpaceSn$,
\bieee
&&\sum_{\gamma\in{\Gamma^*}} \mu^*(\gamma) \sum_{s^n\in\Sset^n} q^n(s^n) \cost^n(\enc_\gamma^n(m,s^n)) \leq \plimit+\alpha \;, \\
&&\err(q^n,\code^{\Gamma^*})\leq \delta \;,
\eieee
 with the following properties:
\begin{enumerate}
 \item 
The size of the code collection is bounded by
$
 |\Gamma^*|\leq n^2 
$. 
\item
\label{item:Lsubset}
 The code collection is a subset of the original code collection, \ie 
$
\Gamma^*\subseteq \Gamma 
$. 
\item
 The distribution $\mu^*$ 
 is uniform, \ie 
$
\mu^*(\gamma)=\frac{1}{|\Gamma^*|} 
$ 
for $\gamma\in\Gamma^*$. 
%
%
\end{enumerate} 
\end{lemma}
The proof of Lemma~\ref{lemm:LcorrSizeC} is given in Appendix~\ref{app:LET}.

\subsubsection{Deterministic Code Capacity}
Here, we consider the AVC with causal SI, under \emph{average} input constraint $\plimit$ and
 state constraint $\Lambda$. 
 We establish a lower bound on the capacity for this setting, and we find a condition under which the deterministic code capacity coincides with the random code capacity for the setting where the jammer is under a state constraint while the user if free of constraints.  
 For every encoding mapping $\encs(u,s)$, define an AVC $\Uavc=\{\xichannel\}$ without SI specified by
$
\xichannel(y|u,s)=\channel(y|\xi(u,s),s) 
$. 

Given a function $\encs:\Uset\times\Sset\rightarrow\Xset$ and a distribution $p\in\pSpace(\Uset)$, define
\bieee
\LambdaO(p,\encs)=\min\, \sum_{u\in\Uset}\sum_{s\in\Sset} p(u)J(s|u)l(s) \;,
\eieee
where the minimization is over all conditional distributions $J(s|u)$ that symmetrize $\xichannel$ 
(see Definition~\ref{def:symmetrizable}). Assume that  
\bieee
\label{eq:assumpC}
\max_{p(u)} \LambdaO(p,\encs)\neq \Lambda \,,\;\text{for all $\encs:\Uset\times\Sset\rightarrow\Xset$}\;.
\eieee
For every $\encs(u,s)$, define the following set. If $\Uchannel$ is symmetrizable, define
\begin{subequations}
\label{eq:plspaceU}
\begin{IEEEeqnarray}{rll}
\pLSpaceU \triangleq 
&\bigg\{    
 p\in\pSpace(\Uset) \,:\;& 
\E_q\, \cost(\encs(U,S))
\leq \plimit \,,\;\text{for all $q\in\apLSpaceS$, and} \IEEEnonumber\\&&
\text{for every $J(s|u)$ that symmetrizes $\Uchannel$}, \IEEEnonumber\\&&
 \sum_{u\in\Uset}\;\sum_{s\in\Sset}\, p(u)  J(s|u) l(s)> \Lambda
\bigg\}\;,
\label{eq:plspaceU1}
\intertext{  
and if $\Uchannel$ is non-symmetrizable, 
}
\pLSpaceU \triangleq
&\bigg\{    
 p\in\pSpace(\Uset) \,:\;& \E_q\, \cost(\encs(U,S))
\leq \plimit \,,\;\text{for all $q\in\apLSpaceS$}
\bigg\}\;, 
\end{IEEEeqnarray}
where $(U,S)\sim p(u)\cdot q(s)$.

The intuition behind the definition of $\pLSpaceU$ above  
 can be explained as follows. 
For a symmetrizable $\Uchannel$, the set defined in (\ref{eq:plspaceU1}) consists of distributions $p(u)$ 
 such that every jamming strategy $J(s|u)$, which symmetrizes $\Uchannel$, 
 violates the state constraint.
 That is, $\pLSpaceU$ consists of distributions for which the jammer is prohibited from using symmetrizing state strategies.
\end{subequations}

Then, let
\begin{align}
\label{eq:ICavceq}
&\LICavc\triangleq\,
\min_{
\text{\footnotesize{$
q(s)\in\pLSpaceS  
$}}
}\,
 \max_{
\substack{
\text{\footnotesize{$
\encs:\Uset\times\Sset\rightarrow\Xset
$}}
\,,\\
\text{\footnotesize{$
p(u)\in\pLSpaceU
$}}
  }}
\,  I_q(U;Y) \;.
\end{align}
Observe that 
 $\LICavc\leq\LrICav$ 
(\cf (\ref{eq:ALcvCIcoro}) 
 and (\ref{eq:ICavceq})). 


\begin{theorem}
\label{theo:ALavcCstateC}
Let $\avc$ be an AVC with causal SI, under average input constraint $\plimit$ and
 state constraint $\Lambda$. 
Suppose that (\ref{eq:assumpC}) holds. Then, 
\begin{enumerate}[{1)}]
\item
 the capacity of $\avc$ is lower bounded by 
\bieee 
\ALCavc \geq \LICavc \;.
\eieee

\item
For $\plimit=\cost_{max}$, \ie when free of input constraints,
if there exists a function $\encs:\Uset\times\Sset\rightarrow\Xset$, such that 
$\Uchannel$ 
 is non-symmetrizable,
 the deterministic code capacity is identical to the random code capacity, \ie
$
\LCavc =\LrCav>0 
$, 
 and it is given by 
\bieee 
\LCavc =\LrICav=\LrIRav \;. 
\eieee
\end{enumerate}
\end{theorem}

The proof of Theorem~\ref{theo:ALavcCstateC} is given in Appendix~\ref{app:ALavcCstateC}.

\begin{center}
\begin{figure*}[tb]
        \centering
        \includegraphics[scale=0.4,trim={0 0 0 1cm},clip]
				{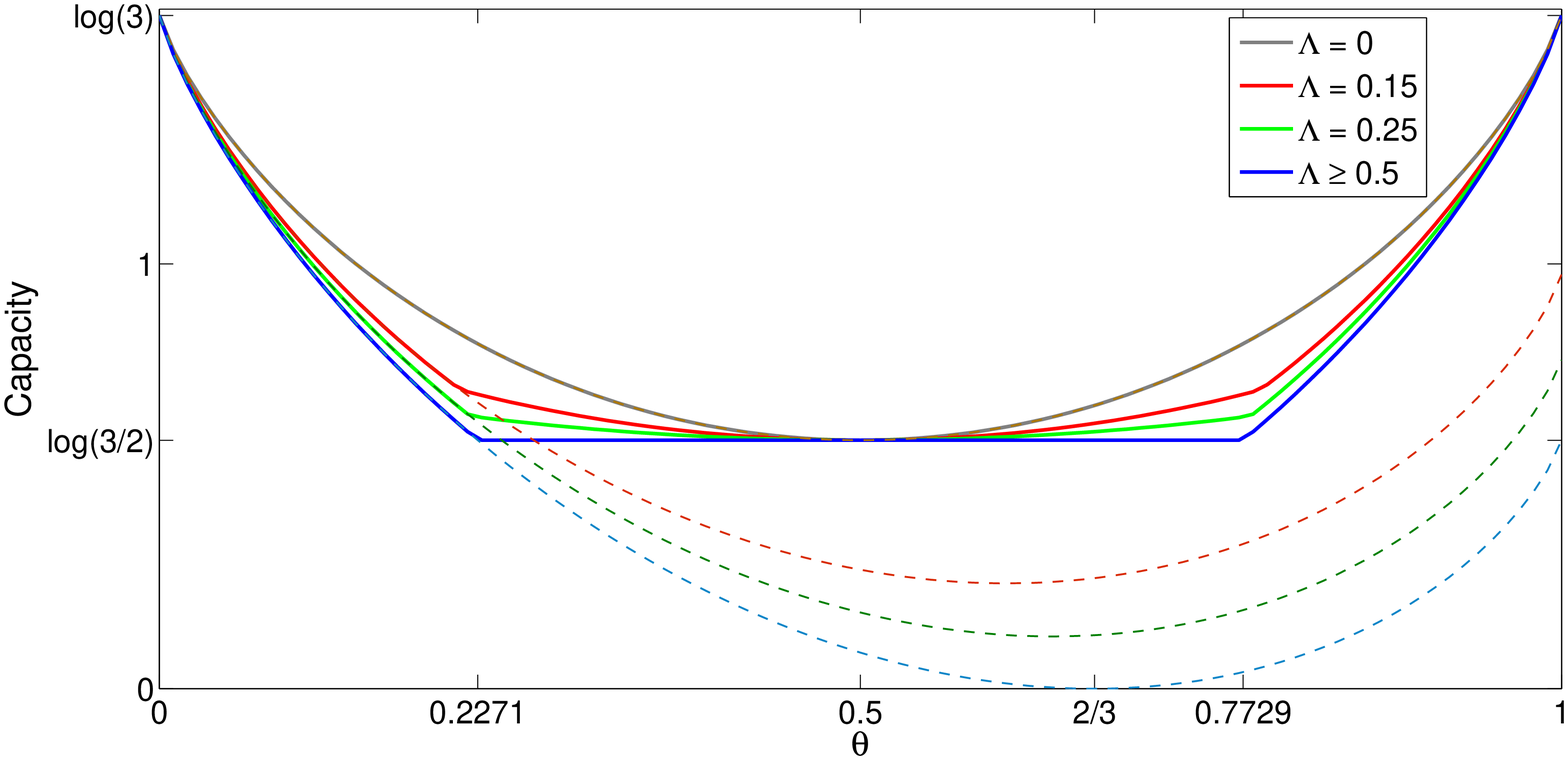}
        
\caption{The capacity of 
the arbitrarily varying noisy-typewriter channel as a function of the transition parameter $\theta$. 
The dashed lines correspond to the capacity $\LCavcig$ of the AVC without SI, and the solid lines correspond to the capacity $\LCavc$ of the AVC with causal SI. 
Each line corresponds to a state constraint $\Lambda=0,0.15,0.25$ and $\Lambda\geq 0.5$, from top to bottom. As the state constraint $\Lambda$ increases, the capacity decreases.
  }
\label{fig:AVNTC}
\end{figure*}
\end{center}

\section{Example}
To illustrate our results, we consider the following example of an AVC with causal SI, under a state constraint.
\begin{example}
\label{example:LAVNTC}
Consider an arbitrarily varying noisy-typewriter channel, defined by
\begin{subequations}
\label{eq:LNT}
\bieee
Y=X+Z \mod 3 \;,
\eieee
where $\Xset=\Zset=\Yset=\{0,1,2\}$. The additive noise is defined by 
$
Z=K\cdot S 
$,
with $S\in\{1,2\}$,
 and
\bieee
K\sim \Bern(\theta) \,,\; \theta>0 \;.
\eieee
\end{subequations}
Thus, $S$ chooses among two noisy-typewriter DMCs \cite{CoverThomas:06b}. The channel is under a state constraint $\Lambda$, with
\bieee
l(s)=\begin{cases}
0 &\text{if $s=1$}\;,\\
1 &\text{if $s=2$}\;.
\end{cases}
\eieee 

We have the following results. 
The capacity of the arbitrarily varying noisy-typewriter channel $\avcig$ without SI, under a state constraint $\Lambda$,  is given by 
\bieee
\label{eq:Lex2Cavcig}
\opC_{\cost_{max},\Lambda}(\avcig)
=
\begin{cases}
\log 3 -h(\theta)-\theta 	h(\Lambda)							&\text{if $0\leq \Lambda\leq \frac{1}{2}$}									\;,\\
\log 3 -h(\theta)-\theta 													&\text{if $\Lambda\geq\frac{1}{2}$}									\;,
\end{cases} \qquad
\eieee
for all $\theta>0$. 
The capacity of the arbitrarily varying noisy-typewriter $\avc$ with causal SI, under a state constraint $\Lambda$, is given by 
\begin{align}
\label{eq:Lex2Cavc}
& \opC_{\cost_{max},\Lambda}(\avc) 
 \nonumber\\
=&\begin{cases}
\log 3-\min\left(\, h(\theta)+\theta h(\Lambda) \,,\; h(\theta)+(1-\theta)h(\Lambda) \,,\; h(\theta*\Lambda)  \right) 
&\text{if $0\leq \Lambda< \frac{1}{2}$}\;,\\
\log 3-\min\left(\, h(\theta)+\theta  \,,\; h(\theta)+(1-\theta) \,,\; 1  \right) 
&\text{if $ \Lambda\geq \frac{1}{2}$}\;.
\end{cases}\qquad
\end{align}
The proof of these results is given in Appendix~\ref{app:LAVNTC}.
%
Figure~\ref{fig:AVNTC} depicts the capacity of the arbitrarily varying noisy-typewriter channel, as a function of the parameter $\theta$. The dashed lines correspond to the case where there is no SI, and the solid lines correspond to the case where causal SI is available at the encoder. 
 Since $\channel$ is symmetrizable if and only $\theta=\frac{2}{3}$,  the capacity without SI and without constraints is zero only for this value. This is equivalent to a modulo-additive DMC $Y=X+Z \mod 3$ where the noise $Z$ is uniform.
On the other hand, with causal SI, the capacity is symmetric around $\theta=\frac{1}{2}$, which resembles the behavior of a BSC, as $K\sim\text{Bernoulli}(\theta)$. 
Choosing the encoding function $\encs(u,s)=u\cdot s \mod 3$, with $\Uset=\{0,1,2\}$, we find that the DMC $\Uchannel$ is non-symmetrizable for all $\theta>0$, thus the capacity of the arbitrarily varying noisy-typewriter with causal SI is positive. Furthermore, 
the capacity is bounded by $\opC_{\cost_{max},\Lambda=1}(\avc)\leq\opC_{\cost_{max},\Lambda}(\avc)\leq\opC_{\cost_{max},\Lambda=0}(\avc)$, where
\begin{align}
\label{eq:Ex1lowL}
&\opC_{\cost_{max},\Lambda=1}(\avc)\geq \log 3 -1=\log\Big( \frac{|\Xset|}{2} \Big) \;,\\
\label{eq:Ex1upL}
&\opC_{\cost_{max},\Lambda=0}(\avc)=\log 3-h(\theta) \;,
\end{align}
by (\ref{eq:Lex2Cavc}).
The lower bound $\log\big( \frac{|\Xset|}{2} \big) $  is the capacity of the standard noisy-typewriter DMC, with $\theta=\frac{1}{2}$.
The upper bound (\ref{eq:Ex1upL}) is the capacity  when the state is known to both the encoder and the receiver.

\end{example}

In conclusion of this chapter,
 we have established lower and upper bounds on 
 the random code capacity, 
for the single-user AVC with causal SI at the encoder, under  input   and  state constraints.
 We have then established a lower bound on the deterministic code capacity, for the  AVC with causal SI at the encoder, under a state constraint and free of input constraint. For this case, we have also  
 stated a condition under which the deterministic code capacity coincides with the random code capacity. The next chapter deals with a multiple-user scenario.

\chapter{The Arbitrarily Varying Degraded Broadcast Channel}
\label{chap:AVDBC}
In this chapter, we address the arbitrarily varying degraded broadcast channel with causal SI available at the encoder. It is assumed that there are no constraints.
	 
\section{Definitions and Previous Results}
\label{sec:Bnotation}
	\subsection{Channel Description}
	\label{subsec:BCs}
A state-dependent discrete memoryless broadcast channel 
$(\Xset\times\Sset,\bc,\Yset_1,\Yset_2)$ consists of a finite input alphabet $\Xset$, two finite output alphabets $\Yset_1$ and $\Yset_2$, a finite state 
alphabet $\Sset$, and a collection of conditional pmfs $p(y_1,y_2|x,s)$ over $\Yset_1\times
\Yset_2$. The channel is memoryless without feedback, and therefore  
$
p(y_1^n,y_2^n|x^n,s^n)= \prod_{i=1}^n \bc(y_{1,i},y_{2,i}|x_i,s_i) 
$. 
The marginals $\sbc$ and $\wbc$ correspond to user 1 and user 2, respectively. 
  For state-dependent broadcast channels with causal SI, the channel input at time $i\in[1:n]$ may depend on the sequence of past and present states $s^i$. 

 Throughout this chapter, we assume that $\bc$ is a \emph{degraded broadcast channel} (DBC). Following the definitions by \cite{Steinberg:05p}, 
a state-dependent broadcast channel $\bc$ is said to be physically degraded if it can be expressed as 
\bieee
\bc(y_1,y_2|x,s)=W_{Y_1|X,S}(y_1|x,s)\cdot 
p(y_2|y_1) \;,
\eieee
 \ie $(X,S)\Cbar Y_1 \Cbar Y_2$ form a Markov chain. 
User 1 is then referred to as the \emph{stronger} user, whereas user 2 is referred to as the \emph{weaker} user. 
More generally, a broadcast channel is said to be stochastically degraded if  
 $\wbc(y_2|x,s)=\sum_{y_1\in\Yset_1}W_{Y_1|X,S}(y_1|x,s)\cdot 
\tp(y_2|y_1)$
for some conditional distribution $\tp(y_2|y_1)$. We note that the definition of degradedness in \cite
{Jahn:81p} is equivalent to the definition above when SI is not available, as assumed in \cite{Jahn:81p}. 
Our results apply to both the physically degraded and the stochastically degraded broadcast channels. 
Thus, for our purposes,
 there is no need to distinguish between the two, and we simply say that the broadcast channel is degraded.

The \emph{arbitrarily varying degraded broadcast channel} (AVDBC) is a discrete memoryless DBC $\bc$  with a state sequence of unknown distribution,  not necessarily independent nor stationary. That is, $S^n\sim \qn(s^n)$ with an unknown joint pmf $\qn(s^n)$ over $\Sset^n$.
In particular, $\qn(s^n)$ can give mass $1$ to some state sequence $s^n$. 
We denote the AVDBC with causal SI by $\avbc=\{\bc\}$.

To analyze the AVDBC with causal SI, we consider the
 \emph{compound  degraded broadcast channel}. 
Different models of a compound  DBC have been considered in the literature, as \eg in 
\cite{WLSSV:09p} and \cite{BPS:14a}.
Here, we define 
the compound  DBC as a discrete memoryless DBC with a discrete memoryless state, where the state distribution $q(s)$ is not known in exact, but rather belongs to a family of distributions $\Qset$, with $\Qset\subseteq \pSpace(\Sset)$. That is,  $S^n\sim\prod_{i=1}^n q(s_i)$, with an unknown pmf $q\in\Qset$ over $\Sset$.
 We denote the compound DBC with causal SI by $\Bcompound$. 

\subsection{Coding}
\label{subsec:coding}
We introduce some preliminary definitions, starting with the definitions of a deterministic code and a random code for the AVDBC $\avbc$ with 
 causal SI. Note that in general,
 the term `a code', unless mentioned otherwise, refers to a deterministic code.  

\begin{definition}[A code, an achievable rate pair and capacity region]
\label{def:Bcapacity}
A $(2^{nR_1},2^{nR_2},n)$ code for the AVDBC $\avbc$ with causal SI consists of the following;   
two  message sets $[1:2^{nR_1}]$ and $[1:2^{nR_2}]$,  where it is assumed throughout that $2^{nR_1}$ and $2^{nR_2}$ are integers,
a set of $n$ encoding functions 
$\enci:  [1:2^{nR_1}]\times[1:2^{nR_2}]\times \Sset^i \rightarrow \Xset$,   
$ i\in [1:n]$, 
 and two  decoding functions,
$
\dec_1: \Yset_1^n\rightarrow [1:2^{nR_1}]   
$
and 
$
\dec_2: \Yset_2^n\rightarrow [1:2^{nR_2}]   
$. 

At time $i\in [1:n]$, given a pair of messages $m_1\in [1:2^{nR_1}]$ and $m_2\in[1:2^{nR_2}]$ and a sequence $s^i$,
 the encoder transmits $x_i=\enci(m_1,m_2,s^i)$. The codeword is then given by 
\bieee
x^n= \encn(m_1,m_2,s^n) \triangleq \left(
\enc_1(m_1,m_2,s_1),\enc_2(m_1,m_2,s^2),\ldots,\enc_n(m_1,m_2,s^n)   \right) \;.
\eieee
Decoder $1$ receives the channel output $y_1^n$, and finds an estimate of the first message $\hm_1=g_1(y_1^n)$. Similarly, decoder 2   estimates  the second message with $\hm_2=g_2(y_2^n)$. 
We denote the code by $\code=\left(\encn(\cdot,\cdot,\cdot),\dec_1(\cdot),\dec_2(\cdot) \right)$.

 Define the conditional probability of error of $\code$ given a state sequence $s^n\in\Sset^n$ by  
\begin{align}
\label{eq:Bcerr}
\cerr(\code)
=&\frac{1}{2^{ n(R_1+R_2 )}}\sum_{m_1=1}^{2^ {nR_1}}\sum_{m_2=1}^{2^ {nR_2}}
\sum_{
\Dset(m_1,m_2)^c} 
 \nBC(y_1^n,y_2^n|\encn(m_1,m_2,s^n),s^n) \;,
\end{align}
where 
\bieee
\Dset(m_1,m_2)\triangleq
\big\{\, (y_1^n,y_2^n)\in\Yset_1^n\times\Yset_2^n:  \left(\dec_1(y_1^n),\dec_2(y_2^n) \right)= (m_1,m_2) \,\big\} \;.
\eieee
Now, define the average probability of error of $\code$ for some distribution $\qn(s^n)\in\pSpace^n(\Sset^n)$, 
\bieee
\err(\qn,\code)
=\sum_{s^n\in\Sset^n} \qn(s^n)\cdot\cerr(\code) \;.
\eieee
We say that $\code$ is a 
$(2^{nR_1},2^{nR_2},n,\eps)$ code for the AVDBC $\avbc$ if it further satisfies 
\bieee
\label{eq:Berr}
 \err(\qn,\code)\leq \eps \;,\quad\text{for all $\qn(s^n)\in\pSpace^n(\Sset^n)\,$.} 
\eieee 

  We say that a rate pair $(R_1,R_2)$ is achievable
if for every $\eps>0$ and sufficiently large $n$, there exists a  $(2^{nR_1},2^{nR_2},n,\eps)$ code. The operational capacity region is defined as the closure of the set of achievable rate pairs and it is denoted by $\BCavc$. 
 We use the term `capacity region' referring to this operational
meaning, and in some places we call it the deterministic code capacity region in order to emphasize that achievability is measured with respect to  deterministic codes.  
\end{definition} 

We proceed now to define the parallel quantities when using stochastic-encoder stochastic-decoders triplets with common randomness.
The codes formed by these triplets are referred to as random codes. 

\begin{definition}[Random code]
\label{def:BcorrC} 
A $(2^{nR_1},2^{nR_2},n)$ random code for the AVDBC $\avbc$ consists of a collection of 
$(2^{nR_1},2^{nR_2},n)$ codes $\{\code_{\gamma}=(\encn_\gamma,\dec_{1,\gamma},\dec_{2,\gamma})\}_{\gamma\in\Gamma}$, along with a probability distribution $\mu(\gamma)$ over the code collection $\Gamma$. 
We denote such a code by $\gcode=(\mu,\Gamma,\{\code_{\gamma}\}_{\gamma\in\Gamma})$.

Analogously to the deterministic case,  a $(2^{nR_1},2^{nR_2},n,\eps)$ random code has the additional requirement
\bieee
 \err(\qn,\gcode)=\sum_{\gamma\in\Gamma} \mu(\gamma)\sum_{s^n\in\Sset^n} \qn(s^n) \cerr(\code_\gamma)\leq \eps \,,\;\text{for all $\qn(s^n)\in\pSpace^n(\Sset^n)$} && \;. \qquad
\eieee
The capacity region achieved by random codes is denoted by $\BrCav$, and it 
 is referred to as the random code capacity region.
\end{definition}
 
%

Next, we write the definition of superposition coding \cite{Bergmans:73p} using  Shannon strategies \cite{Shannon:58p}. See also \cite{Steinberg:05p}, and the discussion after Theorem 4 therein. Here, we refer to such codes as Shannon strategy codes.
\begin{definition}[Shannon strategy codes] 
\label{def:BStratCode}
A $(2^{nR_1},2^{nR_2},n)$ Shannon strategy code for the AVDBC $\avbc$ with causal SI is a $(2^{nR_1},2^{nR_2},n)$ 
 code 
with an encoder that is composed of two strategy sequences
\bieee
u_1^n &:& [1:2^{nR_1}]\times [1:2^{nR_2}] \rightarrow \Uset_1^n \;, \\
u_2^n &:& [1:2^{nR_2}] \rightarrow \Uset_2^n \;,
\eieee
 and an encoding function $\encs(u_1,u_2,s)$, where $\encs:\Uset_1\times\Uset_2\times\Sset\rightarrow \Xset$,  
as well as a pair of decoding functions
$
\dec_1: \Yset_1^n\rightarrow [1:2^{nR_1}] 
$ and 
$
\dec_2: \Yset_2^n\rightarrow [1:2^{nR_2}] 
$. The codeword is then given by 
\bieee
\label{eq:BStratEnc}
x^n=\encs^n(u_1^n(m_1,m_2),u_2^n(m_2),s^n)\triangleq \big[\, \encs(u_{1,i}^n(m_1,m_2),u_{2,i}^n(m_2),s_i) \,\big]_{i=1}^n \;.
\eieee
We denote the code by $\code=\left(u_1^n,u_2^n,\encs,\dec_1,\dec_2 \right)$.
\end{definition}


\subsection{In the Absence of Side Information -- Inner Bound } 
\label{sec:BNOsi}
In this subsection, we briefly review known results for the case where the state is not known to the encoder or the decoder, \ie SI is not available. 

Consider a given AVDBC without SI, which we denote by $\Bavcig$.  
Let
\bieee
\label{eq:BrCavc0def}
\BrIRavig\triangleq
\bigcup_{p(x,u)} \bigcap_{q(s)}    
\left\{
\begin{array}{lll}
(R_1,R_2) \,:\; & R_2 &\leq   I_q(U;Y_2) \;, \\
								& R_1 &\leq   I_q(X;Y_1|U)  
\end{array}
\right\} 
\eieee  
In \cite[Theorem~2]{Jahn:81p}, Jahn introduced an inner bound for 
 the arbitrarily varying \emph{general} broadcast channel. In our case, where the broadcast channel is assumed to be degraded, Jahn's inner bound reduces to the following. 

\begin{theorem}[Jahn's Inner Bound] \cite{Jahn:81p}
\label{Btheo:avcC0R}
Let $\Bavcig$ be an AVDBC  without SI. Then, $\BrIRavig$ is an achievable rate region using random codes over $\Bavcig$, \ie
\bieee
\BrCavig \supseteq \BrIRavig \;.
\eieee
\end{theorem}


Now we move to the deterministic code capacity region. 
\begin{theorem}[Ahlswede's Dichotomy] \cite{Jahn:81p}
\label{theo:BavcC0}
The capacity region of   an AVDBC $\Bavcig$  without SI either coincides with the random code capacity region or else, its interior is empty.
That is, 
$\BCavcig = \BrCavig$ 
 or else, 
$\interior{\BCavcig}=\emptyset$. 
\end{theorem} 
By Theorem~\ref{Btheo:avcC0R} and Theorem~\ref{theo:BavcC0}, we have that $\BrIRavig$ is an achievable rate region, 
if the interior of the capacity region is non-empty. That is,
$
\BCavcig \supseteq \BrIRavig 
$, if 
$ \interior{\BCavcig}\neq\emptyset$. 
\begin{theorem} \cite{
Ericson:85p,CsiszarNarayan:88p,HofBross:06p}
\label{theo:Bsymm0}
For an AVDBC $\Bavcig$ without SI, the interior of the capacity region is non-empty, \ie  $\interior{\BCavcig}\neq\emptyset$, if and only if the marginal $\wbc$ is \emph{not} symmetrizable. 
\end{theorem}


\section{Results}
We present our results on the compound DBC and the AVDBC with causal SI. 

\subsection{The Compound DBC with Causal SI}  
\label{sec:Bcompound}
We now consider the case where the encoder has access to the state sequence in a causal manner, \ie
 the encoder has $S^i$. 
%

\subsubsection{Inner Bound}
 First, we provide an achievable rate region for the compound  DBC with causal SI. Consider a given compound DBC $\Bcompound$ with causal SI. 
Let 
\bieee  
\label{eq:BIRcompound} 
\BIRcompound \triangleq\bigcup_{p(u_1,u_2),\,\encs(u_1,u_2,s)}\, \bigcap_{q(s)\in\Qset} 
\left\{
\begin{array}{lll}
(R_1,R_2) \,:\; & R_2 &\leq   I_q(U_2;Y_2) \;, \\
								& R_1 &\leq   I_q(U_1;Y_1|U_2)  
\end{array}
\right\} 
\eieee
subject to $X=\encs(U_1,U_2,S)$, where $U_1$ and $U_2$ are auxiliary random variables, independent of $S$, and the union is over the pmf $p(u_1,u_2)$ and the set of all functions $\encs:\Uset_1\times\Uset_2\times\Sset\rightarrow\Xset$.
This can also be expressed as
\bieee
\label{eq:BIRcompoundEQ}
\BIRcompound=\bigcup_{p(u_1,u_2),\,\encs(u_1,u_2,s)} 
\left\{
\begin{array}{lll}
(R_1,R_2) \,:\; & R_2 &\leq   \inf_{q\in\Qset} I_q(U_2;Y_2) \;, \\
								& R_1 &\leq   \inf_{q\in\Qset} I_q(U_1;Y_1|U_2)  
\end{array}
\right\}\;.
\eieee 


\begin{lemma}
\label{lemm:BcompoundLowerB}
Let $\Bcompound$ be a compound DBC with causal SI available at the encoder. Then, $\BIRcompound$ is an achievable rate region for $\Bcompound$, \ie 
\bieee
\BCcompound \supseteq \BIRcompound \;.
\eieee
Specifically, if $(R_1,R_2)\in\BIRcompound$, then for some $a>0$ and sufficiently large $n$, there exists a $(2^{nR_1},2^{nR_2},n,e^{-an})$ Shannon strategy code over the compound DBC $\Bcompound$ with causal SI.
\end{lemma}
The proof of Lemma~\ref{lemm:BcompoundLowerB} is given in Appendix~\ref{app:BcompoundLowerB}.

\subsubsection{The Capacity Region}
\label{subsec:BcvC}
We determine the capacity region of the compound DBC $\Bcompound$ with causal SI available at the encoder. 
In addition, we give a condition, for which the inner bound in Lemma~\ref{lemm:BcompoundLowerB} coincides with the capacity region.
For every $q\in\Qset$, define
\bieee
\BICrp \triangleq 
\bigcup_{p(u_1,u_2),\,\encs(u_1,u_2,s)}
\left\{
\begin{array}{lll}
(R_1,R_2) \,:\; & R_2 &\leq   I_q(U_2;Y_2) \;, \\
								& R_1 &\leq   I_q(U_1;Y_1|U_2)  
\end{array}
\right\} \;,
\label{eq:BRPcausalIex}
\eieee
and let
\bieee  
\label{eq:BcvCI}
\BICcompound \triangleq  \bigcap_{q(s)\in\Qset} \BICrp && \;.
\eieee
 
Now, our condition is defined in terms of the following.
\begin{definition}
\label{def:Bcompoundachieve} 
We say that a function $\encs:\Uset_1\times\Uset_2\times\Sset\rightarrow\Xset$ and a set $\Dset\subseteq\pSpace(\Uset_1\times\Uset_2)$ achieve both $\BIRcompound$ and $\BICcompound$ if 
\begin{subequations}
\label{eq:Bcompoundachieve} 
\begin{align}  
\label{eq:BIRcompoundachieve} 
\BIRcompound =\bigcup_{p(u_1,u_2)\in\Dset}\, \bigcap_{q(s)\in\Qset} 
\left\{
\begin{array}{lll}
(R_1,R_2) \,:\; & R_2 &\leq   I_q(U_2;Y_2) \;, \\
								& R_1 &\leq   I_q(U_1;Y_1|U_2)  
\end{array}
\right\} \;,
\intertext{and}
\label{eq:BICcompoundachieve} 
\BICcompound = \bigcap_{q(s)\in\Qset}\, \bigcup_{p(u_1,u_2)\in\Dset} 
\left\{
\begin{array}{lll}
(R_1,R_2) \,:\; & R_2 &\leq   I_q(U_2;Y_2) \;, \\
								& R_1 &\leq   I_q(U_1;Y_1|U_2)  
\end{array}
\right\} \;,
\end{align}
\end{subequations}
subject to $X=\encs(U_1,U_2,S)$. That is, the unions in (\ref{eq:BIRcompound}) and (\ref{eq:BRPcausalIex})  can be restricted to the particular
function $\encs(u_1,u_2,s)$ and set of strategy distributions $\Dset$.
\end{definition}
Observe that by Definition~\ref{def:Bcompoundachieve},  given a function $\encs(u_1,u_2,s)$, if
a set $\Dset$ achieves both $\BIRcompound$ and $\BICcompound$, then every set $\Dset'$ with $\Dset\subseteq\Dset'\subseteq\pSpace(\Uset_1\times\Uset_2)$ achieves those regions, 
 and in particular,  $\Dset'=\pSpace(\Uset_1\times\Uset_2)$. Nevertheless, the condition defined below requires a certain property  that may hold for $\Dset$, but not for $\Dset'$. 

\begin{definition} 
\label{def:sCondQ}
Given a convex set $\Qset$ of state distributions, define the condition $\sCondQ$ by the following;
for some $\encs(u_1,u_2,s)$ and $\Dset$ that achieve both $\BIRcompound$ and $\BICcompound$,
 there exists $q^*\in\Qset$ which minimizes both $I_q(U_2;Y_2)$ and $I_q(U_1;Y_1|U_2)$, for all $p(u_1,u_2)\in\Dset$, 
\ie
\begin{IEEEeqnarray}{ll}
\sCondQ \,:\; &\text{For some $q^*\in\Qset$,}  \IEEEnonumber\\
 &q^*=\arg\min_{q\in\Qset} I_q(U_2;Y_2)=\arg\min_{q\in\Qset} I_q(U_1;Y_1|U_2) \;, \quad
\IEEEnonumber\\
& \forall p(u_1,u_2)\in\Dset 
 \;.
\label{eq:TQ}
\end{IEEEeqnarray}
\end{definition}

\begin{theorem}
\label{theo:BcvC} 
Let $\Bcompound$ be a compound DBC  with causal SI available at the encoder. Then,
\begin{enumerate}[{1)}]
\item
the capacity region of $\Bcompound$ follows 
\bieee 
\label{eq:BcvC}
\BCcompound= \BICcompound  \;,\;\,\text{if $\interior{\BCcompound}\neq\emptyset$}\;, 
\eieee
and it is identical to the corresponding random code capacity region, \ie  $\BrCcompound=\BCcompound$ if $\interior{\BCcompound}\neq\emptyset$.
\item
Suppose that $\Qset\subseteq\pSpace(\Sset)$ is a convex set of state distributions. If the condition $\sCondQ$ holds, the capacity region of  $\Bcompound$ is given by 
\bieee
\BCcompound=\BIRcompound=\BICcompound \;,
\eieee
and it is identical to the corresponding random code capacity region, \ie $\BrCcompound=\BCcompound$. 
\end{enumerate}
%
\end{theorem}
The proof of Theorem~\ref{theo:BcvC} is given in Appendix~\ref{app:BcvC}. 
\subsection{The AVDBC with Causal SI}

We give inner and outer bounds, on the random code capacity region and the deterministic code capacity region, for the 
AVDBC $\avbc$ with causal SI. We also provide conditions, for which the inner bound coincides with the outer bound.
\subsubsection{Random Code Inner and Outer Bounds}
Define 
\begin{align}
\label{eq:BIRcompoundP} 
&\BIRavc \triangleq\bigcup_{p(u_1,u_2),\,\encs(u_1,u_2,s)}\; \bigcap_{q(s)} \,
\left\{
\begin{array}{lll}
(R_1,R_2) \,:\; & R_2 &\leq   I_q(U_2;Y_2) \;, \\
								& R_1 &\leq   I_q(U_1;Y_1|U_2)  
\end{array}
\right\} \;,
\intertext{and}
\label{eq:BrICav}
&\BrICav\triangleq \bigcap_{q(s)}\;  \bigcup_{p(u_1,u_2),\,\encs(u_1,u_2,s)} \,
\left\{
\begin{array}{lll}
(R_1,R_2) \,:\; & R_2 &\leq   I_q(U_2;Y_2) \;, \\
								& R_1 &\leq   I_q(U_1;Y_1|U_2)  
\end{array}
\right\} \;.
\end{align}

Now, we define a condition in terms of the following.
\begin{definition}
\label{def:Bachieve} 
We say that a function $\encs:\Uset_1\times\Uset_2\times\Sset\rightarrow\Xset$ and a set $\Dset^{\rstarC}\subseteq\pSpace(\Uset_1\times\Uset_2)$ achieve both $\BIRavc$ and $\BrICav$ if 
\begin{subequations}
\label{eq:Bachieve} 
\begin{align}  
\label{eq:BIRachieve} 
\BIRavc =\bigcup_{p(u_1,u_2)\in\Dset^{\rstarC}}\, \bigcap_{q(s)} 
\left\{
\begin{array}{lll}
(R_1,R_2) \,:\; & R_2 &\leq   I_q(U_2;Y_2) \;, \\
								& R_1 &\leq   I_q(U_1;Y_1|U_2)  
\end{array}
\right\} \;,
\intertext{and}
\label{eq:BICachieve} 
\BrICav = \bigcap_{q(s)}\, \bigcup_{p(u_1,u_2)\in\Dset^{\rstarC}} 
\left\{
\begin{array}{lll}
(R_1,R_2) \,:\; & R_2 &\leq   I_q(U_2;Y_2) \;, \\
								& R_1 &\leq   I_q(U_1;Y_1|U_2)  
\end{array}
\right\} \;,
\end{align}
\end{subequations}
subject to $X=\encs(U_1,U_2,S)$. That is, the unions in (\ref{eq:BIRcompoundP}) and (\ref{eq:BrICav}) can be restricted to the particular function $\encs(u_1,u_2,s)$ and set of strategy distributions $\Dset^{\rstarC}$.
\end{definition}

\begin{definition} 
\label{def:sCond}
Define the condition $\sCond$ by the following;
for some $\encs(u_1,u_2,s)$ and $\Dset^{\rstarC}$ that achieve both $\BIRavc$ and $\BrICav$,
 there exists $q^*\in\pSpace(\Sset)$ which minimizes both $I_q(U_2;Y_2)$ and $I_q(U_1;Y_1|U_2)$, for all 
$p(u_1,u_2)\in\Dset^{\rstarC}$, 
\ie
\begin{IEEEeqnarray*}{ll}
\sCond \,:\; &\text{For some $q^*\in\pSpace(\Sset)$,}  \\ 
 &q^*=\arg\min_{q(s)} I_q(U_2;Y_2)=\arg\min_{q(s)} I_q(U_1;Y_1|U_2) 
\quad \forall p(u_1,u_2)\in\Dset^{\rstarC} 
\;.
\end{IEEEeqnarray*}
\end{definition}

\begin{theorem}
\label{theo:Bmain}
Let $\avbc$ be an AVDBC with causal SI available at the encoder. Then,  
\begin{enumerate}[{1)}]
\item
the random code capacity region of $\avbc$ is bounded by
\bieee
\BIRavc \subseteq \BrCav \subseteq \BrICav \;.
\eieee
\item 
If the condition $\sCond$ holds, the random code capacity region of  $\avbc$ is given by
	\bieee
	\label{eq:BrCavTight}
	\BrCav=\BIRavc=\BrICav \;.
	\eieee
\end{enumerate}
\end{theorem}
The proof of Theorem~\ref{theo:Bmain} is given in Appendix~\ref{app:Bmain}. 

The following lemma is a restatement of a result from \cite{Ahlswede:78p}, stating that a polynomial size of the code collection $\{\code_\gamma\}$ is sufficient. This result is a key observation in Ahlswede's Elimination Technique (ET), presented in \cite{Ahlswede:78p}, and it is significant for the deterministic code analysis.   
\begin{lemma} 
\label{lemm:BcorrSizeC}  
Consider a given  
 $(2^{nR_1},2^{nR_2},n,\eps_n)$ random 
code $\code^\Gamma=(\mu,\Gamma,\{\code_\gamma\}_{\gamma\in\Gamma})$
 for the AVDBC $\avbc$, 
where $\lim_{n\rightarrow\infty} \eps_n=0$. 
Then, for every $0<\alpha<1$ and sufficiently large $n$, there exists a $(2^{nR_1},2^{nR_2},n,\alpha)$ random 
 code $(\mu^*,\Gamma^*,\{\code_{\gamma}
\}_{\gamma\in\Gamma^*})$ with the following properties:
\begin{enumerate}
 \item 
The size of the code collection is bounded by
 $
 |\Gamma^*|\leq n^2 
$. 
\item
\label{item:Bsubset}
 The code collection is a subset of the original code collection, \ie 
$
\Gamma^*\subseteq \Gamma 
$. 
\item
 The distribution $\mu^*$ 
 is uniform, \ie 
$
\mu^*(\gamma)=\frac{1}{|\Gamma^*|} 
$, 
for $\gamma\in\Gamma^*$. 
\end{enumerate} 
\end{lemma} 
The proof of Lemma~\ref{lemm:BcorrSizeC} follows 
 the same lines  as in \cite[Section 4]{Ahlswede:78p} (see also \cite{WinshtokSteinberg:06c}).
 For completeness, we give the proof in Appendix~\ref{app:BET}.

\subsubsection{Deterministic Code Inner and Outer Bounds}
The next theorem characterizes the deterministic code capacity region, which demonstrates a dichotomy property. 
\begin{theorem}
\label{theo:BcorrTOdetC}
The capacity region of an AVDBC $\avbc$ with causal SI either coincides with the random code capacity region or else, it has an empty interior. 
That is, 
$\BCavc = \BrCav$ 
 or else, 
$ \interior{\BCavc}=\emptyset$. 
\end{theorem}
The proof of Theorem~\ref{theo:BcorrTOdetC} is given in Appendix~\ref{app:BcorrTOdetC}.
For every function $\encs':\Uset_2\times\Sset\rightarrow\Xset$, define a DMC $V_{Y_2|U_2,S}^{\encs'}$ specified by 
$V_{Y_2|U_2,S}^{\encs'}(y_2|u_2,s)=W_{Y_2|X,S}(y_2|\encs'(u_2,s),s)$.
%
%
\begin{coro}
\label{coro:BmainDbound}
The capacity region of $\avbc$ is bounded by
\bieee
&&\BCavc \supseteq \BIRavc \,,\;\text{if}\;\, \interior{\BCavc}\neq \emptyset \;, \label{eq:BmainInner} \\
&&\BCavc \subseteq \BrICav \;. \label{eq:BmainOuter}
\eieee
Furthermore,
if $V_{Y_2|U_2,S}^{\encs'}$ is non-symmetrizable for some $\encs':\Uset_2\times\Sset\rightarrow\Xset$, and the condition $\sCond$ holds, then 
 $\,\BCavc=\BIRavc=\BrICav$.
\end{coro}
The proof of Corollary~\ref{coro:BmainDbound} is given in Appendix~\ref{app:BmainDbound}.

To conclude this chapter, we have established inner and outer bounds, on the random code capacity region and the deterministic code capacity region, for the 
AVDBC $\avbc$ with causal SI. We also provided conditions, for which the inner bound coincides with the outer bound.


\section{Example}
\label{sec:exAVBSBC}
To illustrate the results above, we give the following example.
\begin{example} \cite[Section IV-A]{Steinberg:05p} 
\label{example:AVBSBC}
Consider an arbitrarily varying binary symmetric broadcast channel (BSBC), 
\bieee
Y_1&=&X+Z_S \mod 2 \;, \IEEEnonumber\\
Y_2&=&Y_1+V \mod 2 \;, \IEEEnonumber
\eieee
where $X,Y_1,Y_2,S,Z_S,V$ are binary, with values in $\{0,1\}$. The additive noises are distributed according to 
\bieee
Z_s&\sim& \text{Bernoulli}(\theta_s) \,,\; \text{for $s\in\{0,1\}$}\;, \IEEEnonumber\\
V&\sim& \text{Bernoulli}(\alpha) \;, \IEEEnonumber
\eieee
 with  $\theta_0\leq 1-\theta_1 \leq \frac{1}{2}$ and $\alpha<\frac{1}{2}$, where $V$ is independent of $(S,Z_S)$. 
It is readily seen the channel is  physically degraded. 
 Define the binary entropy function $h(x)=-x\log x-(1-x)\log(1-x)$, for $x\in [0,1]$, with logarithm to base $2$.

We have the following results.
The capacity region of the arbitrarily varying BSBC $\Bavcig$ without SI is given by
\bieee
\label{eq:Bex1Cavcig}
\BCavcig= \{(0,0)\} \;. 
\eieee
The capacity region of the arbitrarily varying BSBC $\avbc$ with causal SI is given by 
\bieee
\label{eq:Bex1Cavc}
\BCavc= \bigcup_{0\leq \beta\leq 1}
\left\{
\begin{array}{lll}
(R_1,R_2) \,:\; & R_2 &\leq   1-h(\alpha*\beta*\theta_1) \;, \\
								& R_1 &\leq   h(\beta*\theta_1)-h(\theta_1)
\end{array}
\right\}\;.
\eieee
It will be seen in the achievability proof that the parameter $\beta$ is related to the distribution of  $U_1$, and thus the RHS of 
(\ref{eq:Bex1Cavc}) can be thought  of as a union over Shannon strategies.
The analysis is given in Appendix~\ref{app:AVBSBC}. 

It is shown in Appendix~\ref{app:AVBSBC} that the condition $\sCond$ holds and $\BCavc=\BIRavc=\BrICav$.
Figure~\ref{fig:BBSC} provides a graphical interpretation. 
Consider a DBC $\bc$ with random parameters  with causal SI, governed by an i.i.d. state sequence, distributed according to $S\sim\text{Bernoulli}(q)$, for a given $0\leq q\leq 1$, and let $\BCrp$ denote the corresponding capacity region.
 Then, the analysis shows that  the condition $\sCond$ implies that there exists $0\leq q^*\leq 1$ such that $\BCavc=\opC(\avbc^{q^*})$, where $\opC(\avbc^{q^*})\subseteq \opC(\avbc^{q}) $ for every $0\leq q\leq 1$. Indeed,
looking at Figure~\ref{fig:BBSC}, it appears that the regions $\BCrp$, for $0\leq q\leq 1$, form a well ordered set, hence
 $\BCavc=\opC(\avbc^{q^*})$ with $q^*=1$.


\begin{center}
\begin{figure}[hbt]
        \centering
        \includegraphics[scale=0.52,trim={2.5cm 0 0 0},clip
				]{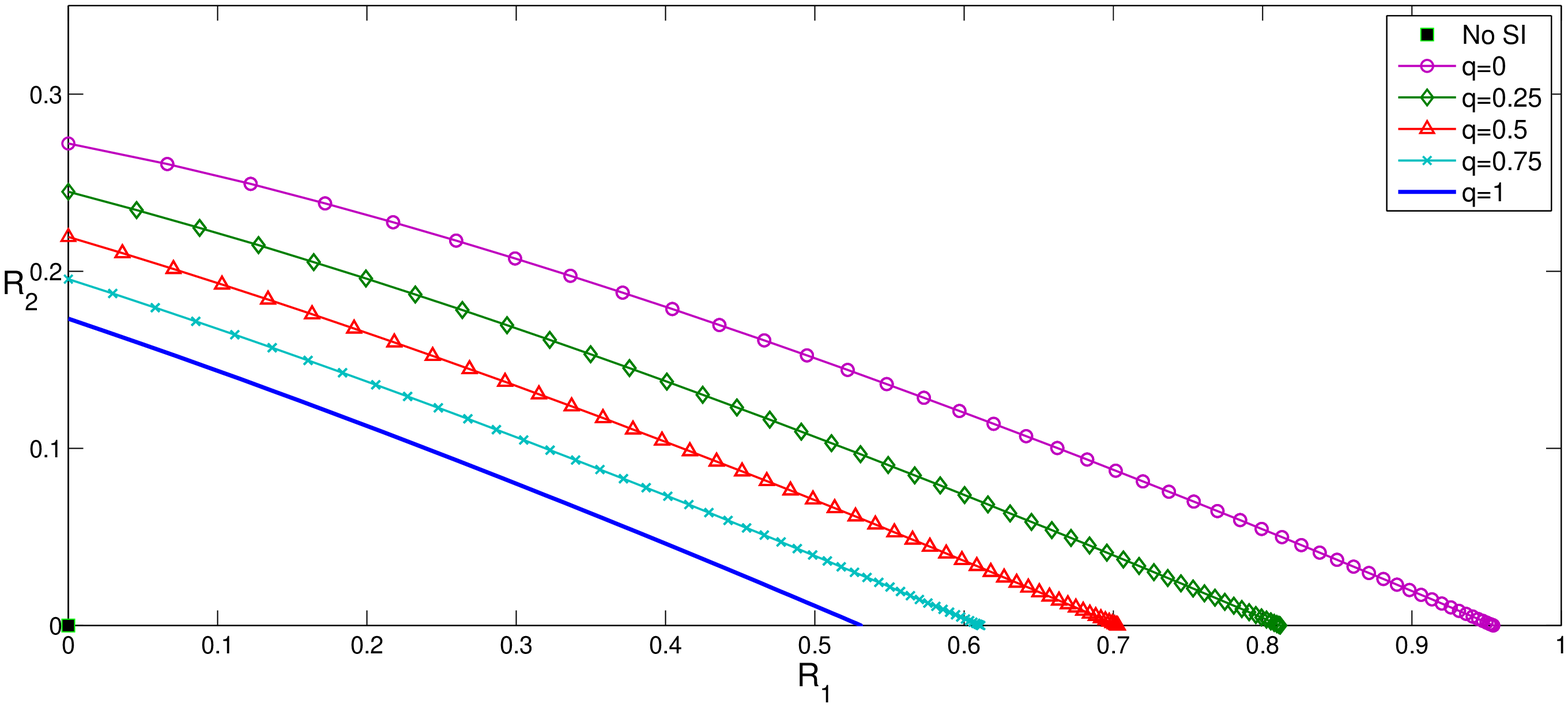}
        
\caption{The capacity region of the AVDBC in Example~\ref{example:AVBSBC}, the arbitrarily varying binary symmetric broadcast channel. The area under the thick blue line is the  capacity region of the AVDBC $\avbc$ with causal SI, with $\theta_1=0.005$, 
$\theta_2=0.9$, and 
$\alpha=0.2$. The black square at the origin stands for the capacity region of the AVDBC $\Bavcig$ without SI, $\BCavcig=\{(0,0)\}$.
The curves depict $\BCrp$ for $q=0,0.25,0.5,0.75,1$, where the capacity region of $\avbc$ is given by $\BCavc=\BrICav=\BCrp$ for $q=1$ 
(see (\ref{eq:BrICav})). 
  }
\label{fig:BBSC}
\end{figure}
\end{center}
\end{example}

\vspace{-0.35cm}
\begin{appendices}
\chapter{Input and State Constraints: Proofs}
\vspace{-0.75cm}
Observe that it suffices to prove the lower bound for the strict input constraint, and the upper bound for the average input constraint. This follows from the fact that
 the capacity under  average input constraint is at least as high as the corresponding capacity under per message input constraint, \ie $\LCcompound\leq\ALCcompound$ and $\LrCcompound\leq\ALrCcompound$.
\section{
Proof of Lemma~\ref{lemm:ALCompoundPC}}
\label{subsec:ALCompoundPC}

\subsubsection{ Lower Bound} 
We construct a code based on Shannon strategies, and decode using joint typicality  with respect to a state type, which is ``close" to some  $q\in\pLSpaceS$. 

We begin with the following definitions. Basic method of types concepts are defined as in \cite[Chapter 2]{CsiszarKorner:82b}; including the definition of a type $\hP_{x^n}$ of a sequence $x^n$; a joint type $\hspace{-0.1cm}\hP_{x^n,y^n}$ and a conditional type $\hP_{x^n|y^n}$ of a pair of sequences $(x^n,y^n)$; and 
 a $\delta$-typical set $\tset(P_{X,Y})$ with respect to a distribution $P_{X,Y}(x,y)$. 
 We also define a set of state types $\tQ$ by  
\bieee
\label{eq:tQ}
\tQ=\{ \hP_{s^n} \,:\; s^n\in\Aset^{ \delta_1 
}(q) \;\text{ for some  $q\in\pLSpaceS$}\, \} \;,
\eieee
where 
\bieee
\label{eq:delta1}
\delta_1 \triangleq
\frac{\delta}{2\cdot |\Sset|} \;.
\eieee 
Namely, $\tQ$ is the set of types that are $\delta_1$-close 
 to some state distribution $q(s)$ in $\pLSpaceS$.
A code $\code$ for the compound channel with causal SI is constructed as follows.

\emph{Codebook Generation}:
 Fix the distribution $P_U(u)$ and the function $\encs(u,s)$ that achieve  
$\inR_{low,\plimit-\eps,\Lambda+\eps}^{\rstarC}(\avc)$, 
where $\eps>0$ is arbitrarily small. Generate $2^{nR}$ independent sequences $u^n(m)$, $m\in[1:2^{nR}]$, at random, each according to $\prod_{i=1}^n P_U(u_i)$. Reveal the codebook to the encoder and the decoder.

\emph{Encoding}:
A message $m\in [1:2^{nR}]$ is encoded as follows. 
If 
\bieee
\label{eq:QdirectState}
\sum_{\tsn\in\Sset^n} q^n(\tsn) \cost^n(\encs^n(u^n(m),\tsn))\leq \plimit \,,\;\text{for all $q\in\pLSpaceS$}\;,
\eieee 
where $q^n(s^n)=\prod_{i=1}^n q(s_i)$, then
 transmit at time $i\in [1:n]$, $x_i=\encs(u_i(m),s_i)$.
Otherwise, if (\ref{eq:QdirectState}) fails to hold for some $q\in\pLSpaceS$, transmit $x^n$ $=$ $(a,\ldots,a)$, with an idle symbol $a\in\Xset$ with $\cost(a)=0$.


\emph{Decoding}: As $y^n$ is received, the decoder finds a unique $\hm\in[1:2^{nR}]$ such that
$
 (u^n(\hm),y^n)\in\tset(P_U P^{q}_{Y|U}) $, 
for some $q\in\tQ$, 
where 
 \bieee
\label{eq:LUchannelYL}
P^{q}_{Y|U}(y|u)=\sum_{s\in\Sset} q(s)\channel\left( y| \encs(u,s),s \right) \;.
\eieee
If there is none, or more than one such $\hm\in[1:2^{nR}]$, then the decoder declares an error.

\emph{Analysis of Probability of Error}:
Due to symmetry, we may assume without loss of generality that the user sent the message 
$m=1$. Let $q(s)\in\pLSpaceS$ denote the \emph{actual} state distribution chosen by the jammer.

The error event is bounded by the union of the events below. 
Define 
\bieee
\Eset_1&=&\{ U^n(1)\notin \Aset^{\nicefrac{\delta}{2}}(P_U) \} \;, \\
\Eset_2 &=&\{ (U^n(1),Y^n)\notin \tset(P_U P^{q'}_{Y|U}) \;\text{ for all $q'\in\tQ$} \} \;, 
\label{eq:SlemmCompoundCerrEv1}\\
\Eset_3 &=&\{ (U^n(m),Y^n)\in \tset(P_U P^{q'}_{Y|U}) \;\text{ for some $m\neq 1,\, q'\in\tQ$} \} \;.
\label{eq:SlemmCompoundCerrEv2}
\eieee
 Then, 
 the probability of error is bounded by
\begin{align}
\label{eq:SlemmCompoundCerrBound}
 \err(q,\code) 
\leq& \prob{\Eset_1}+ \cprob{\Eset_2}{\Eset_1^c }+ \cprob{\Eset_3}{\Eset_1^c } \;,
\end{align}
 where the conditioning on $M=1$ is omitted for convenience of notation. 
The first term in the RHS of (\ref{eq:SlemmCompoundCerrBound}) tends to zero exponentially as $n\rightarrow\infty$, by the law of large numbers and Chernoff's bound. As for the other terms,
 observe that given that the event $\Eset_1^c$ occurs, \ie $U^n(1)\in\Aset^{\nicefrac{\delta}{2}}(P_U)$, 
 we have that for a sufficiently small $\delta>0$, the requirement
 \bieee
\sum_{s^n\in\Sset^n}q^n(s^n) \cost^n(\encs^n(U^n(1),s^n))=\frac{1}{n} \sum_{i=1}^n \sum_{s\in\Sset} q(s)
\cost(\encs(U_i(1),s))\leq\plimit
\eieee
 is held for all $q\in\pLSpaceS$. Hence,
\bieee
\label{eq:SLencsDirLow}
X^n=\encs^n(U^n(1),S^n) \;.
\eieee
	
As for the second term in the RHS of (\ref{eq:SlemmCompoundCerrBound}), we now claim that  the event $\Eset_2$ implies that 
	$(U^n(1),Y^n)\notin \Aset^{\nicefrac{\delta}{2}}(P_U P^{q''}_{Y|U})$ for all $q''\in\pLSpaceS$.
	This claim is due to the following. 
Suppose that $(U^n(1),Y^n)\in \Aset^{\nicefrac{\delta}{2}}(P_U P^{q''}_{Y|U})$ for some $q''\in\pLSpaceS$.  
Then, for a sufficiently large $n$, there exists a type $q'(s)$ such that 
\bieee
\label{eq:qdelta1}
|q'(s)-q''(s)|\leq \delta_1 \;, 
\eieee
 for all $s\in\Sset$, and by the definition in (\ref{eq:tQ}), $q'\in\tQ$.  Then, (\ref{eq:qdelta1}) implies that
\bieee
|P_{Y|U}^{q'}(y|u)-P_{Y|U}^{q''}(y|u)|\leq |\Sset|\cdot \delta_1=\frac{\delta}{2} \;,
\eieee
for all $u\in\Uset$ and $y\in\Yset$ (see (\ref{eq:delta1}) and (\ref{eq:LUchannelYL})). Hence, 
$(U^n(1),Y^n)\in \Aset^{\delta}(P_U P^{q'}_{Y|U})$. It follows that if  $(U^n(1),Y^n)\notin \Aset^{\delta}(P_U P^{q'}_{Y|U})$ for all $q'\in\tQ$, then 
	$(U^n(1),Y^n)\notin \Aset^{\nicefrac{\delta}{2}}(P_U P^{q''}_{Y|U})$ for all $q''\in\pLSpaceS$. 
		Thus,
	\bieee
	\label{eq:SllnRL}
	\cprob{\Eset_2}{\Eset_1^c}&\leq& \cprob{(U^n(1),Y^n)\notin \Aset^{\nicefrac{\delta}{2}}(P_U P^{q''}_{Y|U}) \;\text{ for all $q''\in\pLSpaceS$} }{\Eset_1^c} \IEEEnonumber \\
	&\leq& \cprob{(U^n(1),Y^n)\notin \Aset^{\nicefrac{\delta}{2}}(P_U P^{q}_{Y|U})  }{\Eset_1^c} \;.
	\eieee
The RHS of (\ref{eq:SllnRL}) exponentially tends to zero as $n\rightarrow\infty$ by the law of large numbers and  Chernoff's bound. 
	
	We move to the third term in the RHS of (\ref{eq:SlemmCompoundCerrBound}). 
	By the union of events bound 
and the fact that the number of type classes in $\Sset^n$ is bounded by $(n+1)^{|\Sset|}$,  
 we have that  
\begin{align}
\label{eq:SE2poly}
&\cprob{\Eset_3}{\Eset_1^c}
\leq (n+1)^{|\Sset|}\cdot \sup_{q'\in\tQ} \cprob{
(U^n(m),Y^n)\in \tset(P_U P^{q'}_{Y|U}) \;\text{ for some $m\neq 1$} 
}{\Eset_1^c} \nonumber\\
\leq& (n+1)^{|\Sset|}\cdot 2^{nR} \cdot \sup_{q'\in\tQ}\left[  \sum_{u^n\in\Uset^n} P_{U^n}(u^n) \cdot \sum_{y^n \,:\; (u^n,y^n)\in \tset(P_U P^{q'}_{Y|U})} P_{Y^n}^q(y^n)
\right] \;,
\end{align}
where we have defined
	$
	P_Y^q(y)
	=\sum\limits_{u\in\Uset,s\in\Sset} P_U(u)\cdot q(s)\cdot \channel(y|\encs(u,s),s) 
	$. 
This follows from (\ref{eq:SLencsDirLow}) and the fact that 
 $U^n(m)$ is independent of $Y^n$ for every $m\neq 1$. 
 Let $y^n$ satisfy $(u^n,y^n)\in \tset(P_U P^{q'}_{Y|U})$. Then, $\,y^n\in\Aset^{\delta_2}(P_Y^{q'})$ with $\delta_2\triangleq |\Uset|\cdot\delta$. By Lemmas 2.6 and 2.7 in
 \cite{CsiszarKorner:82b},
\bieee
\label{eq:SpYbound}
P_{Y^n}^q(y^n)=2^{-n\left(  H(\hP_{y^n})+D(\hP_{y^n}||P_Y^q)
\right)}\leq 2^{-n H(\hP_{y^n})}
\leq 2^{-n\left( H_{q'}(Y) -\eps_1(\delta) \right)} \;,
\eieee
where $\eps_1(\delta)\rightarrow 0$ as $\delta\rightarrow 0$. Therefore, by (\ref{eq:SE2poly})$-$(\ref{eq:SpYbound}), along with 
 \cite[Lemma 2.13]{CsiszarKorner:82b},
\begin{align}
& \prob{\Eset_3|\Eset_1^c}           																									
\leq
 \;(n+1)^{|\Sset|}\cdot \sup_{q'\in\pLSpaceS} 
2^{-n[ I_{q'}(U;Y) 
-R-\eps_2(\delta) ]} \label{eq:SLexpCR} \;,
\end{align}
with $\eps_2(\delta)\rightarrow 0$ as $\delta\rightarrow 0$, 
 The RHS of (\ref{eq:SLexpCR})
 exponentially tends to zero as $n\rightarrow\infty$, provided that $R<\min_{q'\in\pLSpaceS} I_{q'}(U;Y)
-\eps_2(\delta)$. 
%
\qed

\subsubsection{Upper Bound}
Assume to the contrary that there exists an achievable rate $R>\LrIRav$ using random codes. Thus, for some $q^*(s)\in\pLSpaceS$, we have that 
$R>\inC_{\plimit}(\avc^{q^*})$, where $\inC_{\plimit}(\avc^{q})\triangleq \max\limits_{
\text{\footnotesize{$
\encs(u,s),p(u) \,:\; \E_{q}\, \cost(\encs(U,S))\leq\plimit
$}}
} I_{q}(U;Y)$. 

The achievability assumption implies that for every $\eps>0$ and sufficiently large $n$, there exists a $(2^{nR},n)$ random code 
$\code^\Gamma$ for the compound channel $\LcompoundP$ 
 such that $\err(\qn,\code^\Gamma)<\eps$ for all i.i.d. state distributions $q\in\pLSpaceS$. 
 If such a code would exist, it could have been used over a random parameter channel with $S^n\sim \prod_{i=1}^n q^*(s_i)$, 
with causal SI,
 achieving a rate $R>\inC_{\plimit}(\avc^{q^*})$. This stands in contradiction to Shannon's fundamental result in 
\cite{Shannon:58p}, 
 hence the assumption is false. 
\qed 

\section{Proof of Lemma~\ref{lemm:LRT}}
\label{app:LRT}
We state the proof of our modified version of Ahlswede's RT \cite{Ahlswede:78p}. The proof follows the lines of \cite[Subsection IV-B]{Ahlswede:78p}. 
Let $\widetilde{s}^{\;n}\in\Sset^n$ such that $l^n(\widetilde{s}^{\;n})\leq\Lambda$. Denote the type of  $\widetilde{s}^{\;n}\in\Sset^n$ by $\hq
$.  
Observe that 
$
\hq\in \apLSpaceS 
$. 

Given a permutation $\pi\in\Pi_n$,
\bieee
\sum_{s^n\in\Sset^n} q^n(s^n) h(s^n)=\sum_{s^n\in\Sset^n} q^n(\pi s^n) h(\pi s^n)=\sum_{s^n\in\Sset^n} q^n(s^n) h(\pi s^n) \;,
\eieee
for every i.i.d. state distribution $q^n(s^n)=\prod_{i=1}^n q(s_i)$, with  $q\in\apLSpaceS$,
where the first equality holds since $\pi$ is a bijection, and the second equality holds since $q^n$ is i.i.d.
Hence, taking $q=\hq$,
\bieee
\sum_{s^n\in\Sset^n} \hq^{\;n}(s^n) h(s^n)=\frac{1}{n!} \sum_{\pi\in\Pi_n} \sum_{s^n\in\Sset^n} \hq^{\;n}(s^n) h(\pi s^n) \;,
\eieee
and by (\ref{eq:RTcondCs}),
\bieee
\sum_{s^n\in\Sset^n}  \hq^{\;n}(s^n) \left[\frac{1}{n!}\sum_{\pi\in\Pi_n} h(\pi s^n)\right] \leq \alpha_n \;.
\eieee
Then,
\bieee
\sum_{s^n \,:\; \hP_{s^n}=\hq}  \hq^{\;n}(s^n) \left[\frac{1}{n!}\sum_{\pi\in\Pi_n} h(\pi s^n)\right] \leq \alpha_n \;.
\eieee
The expression in the square brackets is identical for all sequences $s^n$ of type $\hq$. Thus, 
\bieee
\label{eq:rtineq1}
\left[\frac{1}{n!}\sum_{\pi\in\Pi_n} h(\pi \widetilde{s}^{\;n})\right]\cdot
\sum_{s^n \,:\; \hP_{s^n}=\hq}  \hq^{\;n}(s^n)  \leq \alpha_n \;.
\eieee
The second sum is the probability of 
 the type class of $\hq$, hence
\bieee
\label{eq:rtineq2}
\sum_{s^n \,:\; \hP_{s^n}=\hq}  \hq^{\;n}(s^n) \geq \frac{1}{(n+1)^{|\Sset|}} \;,
\eieee
by \cite[Theorem 11.1.4]{CoverThomas:06b}. The proof follows from (\ref{eq:rtineq1}) and (\ref{eq:rtineq2}). \qed

\section{Proof of Theorem~\ref{theo:ALavcCr}}
\label{app:ALavcCr}

Consider the AVC $\avc$ per message 
 input constraint $\plimit$  and state constraint $\Lambda$, as specified by (\ref{eq:StateCn}) and (\ref{eq:inputCstrict}). 

\subsection*{Part 1}
\subsubsection{Lower Bound} 
We use Ahlswede's RT twice, as follows. Let 
$R<\inR^{\rstarC}_{low,\plimit-2\delta,\Lambda+2\delta}(\avc)$,
 where $\delta>0$ is arbitrarily small.
Consider the compound channel with causal SI, under input constraint $\plimit$, with $\Qset=\pLSpaceS$, hence $\Qset\subseteq \overline{\pSpace}_{\Lambda+2\delta}(\Sset)$.
According to Lemma~\ref{lemm:ALCompoundPC}, 
for some $\theta>0$ and sufficiently large $n$,  
 there exists a  $(2^{nR},n)$ Shannon strategy code $\code=(U^n(m),\encs(u,s),\dec(y^n))$ for the compound channel $\LcompoundP$ with causal SI, such that 
\bieee
\label{eq:LrAVcosti}
\sum_{s^n\in\Sset^n} q^n(s^n) \cdot\E\, 
\cost^n(\encs^n(U^n(m),&&s^n))  \leq\plimit-2\delta \,,\; 
\text{for all $m\in [1:2^{nR}]$ }\;.
\eieee
and 
\bieee
\label{eq:LrAVerrDirect}
\E\,
\err(q,\code)&&=\sum_{s^n\in\Sset^n} q^n(s^n) \cdot\E\, \cerr(\code)  \leq e^{-2\theta n} \;,
\eieee
for all i.i.d. state distributions $q^n(s^n)=\prod_{i=1}^n q(s_i)$, with $q\in\apLSpaceS$.
 The expectation in Equations (\ref{eq:LrAVcosti}) and (\ref{eq:LrAVerrDirect}) is on the ensemble of codebooks, corresponding to the
 independent i.i.d. random sequences $U^n(m)$, $m\in [1:2^{nR}]$, as set in the proof of Lemma~\ref{lemm:ALCompoundPC}. 

Given such a Shannon strategy code, 
 we have that (\ref{eq:RTcondCs}) is satisfied with  $h_0(s^n)=\E\, \cerr(\code)$  and $\alpha_n=e^{-2\theta n}$.  
Consequently, by Lemma~\ref{lemm:LRT}, for a sufficiently large $n$,
\bieee
\label{eq:ALdetErrC}
\frac{1}{n!} \sum_{\pi\in\Pi_n} \E\, P_{e|\pi s^n}^{(n)}(\code)&&\leq (n+1)^{|\Sset|}e^{-2\theta n} 
\leq e^{-\theta n}  \;,
\eieee
for all $s^n\in\Sset^n$ with $l^n(s^n)\leq\Lambda$.  

On the other hand, for every Shannon strategy  code $\code=(u^n(m),\encs(u,s),g(y^n))$, and for every $\pi\in\Pi_n$,
\bieee
P_{e|\pi s^n}^{(n)}(\code)  &\stackrel{(a)}{=}&
\frac{1}{2^{ nR }}\sum_{m=1}^{2^ {nR}}
\sum_{y^n:\dec(y^n)\neq m}  \nchannel(y^n|\encs^n(u^n(m),\pi s^n),\pi s^n) \IEEEnonumber\\
&\stackrel{(b)}{=}&\frac{1}{2^{ nR }}\sum_{m=1}^{2^ {nR}}
\sum_{ y^n:\dec(\pi y^n)\neq m}  \nchannel(\pi y^n|\encs^n(u^n(m),\pi s^n),\pi s^n) \IEEEnonumber\\
&\stackrel{(c)}{=}&\frac{1}{2^{ nR }}\sum_{m=1}^{2^ {nR}}
\sum_{y^n:\dec(\pi y^n)\neq m}  \nchannel( y^n|\pi^{-1}\encs^n(u^n(m),\pi s^n), s^n) \;,\qquad
\label{eq:Lcerrpi}
\eieee
where $(a)$ is obtained by plugging 
 $\pi s^n$ and $x^n=\encs^n(\cdot,\cdot)$ in (\ref{eq:cerr});
in $(b)$ we simply change the order of summation over $y^n$; and $(c)$ holds because the channel is memoryless. Note that for a Shannon strategy code,  $x_i=\encs(u_i,s_i)$, $i\in[1:n]$, by Definition~\ref{def:StratCode} (see (\ref{eq:StratEnc})).
 Thus,
$\pi^{-1}\encs^n(u^n(m),\pi s^n)$ $=$ $\encs^n(\pi^{-1} u^n(m),s^n) 
$, and 
\bieee
 P_{e|\pi s^n}^{(n)}(\code)
&=&\frac{1}{2^{ nR }}\sum_{m=1}^{2^ {nR}}
\sum_{y^n:\dec(\pi y^n)\neq m} \nchannel( y^n|\encs^n(\pi^{-1} u^n(m),s^n), s^n) \;.\qquad
\label{eq:LcorrErrC}
\eieee
The last expression suggests the use of permutations applied to the encoding \emph{strategy sequence} and the channel output sequence.

Then, consider the $(2^{nR},n)$ random code $\code^\Pi$, specified by 
\bieee
\label{eq:LCpi}
f_\pi^n(m,s^n)= \encs^n(\pi^{-1} U^n(m),s^n) \;,\quad g_\pi(y^n)=\dec(\pi y^n) \;,\quad \pi\in\Pi_n \;,
\eieee 
with a uniform distribution $\mu(\pi)=\frac{1}{|\Pi_n|}=\frac{1}{n!}$. 
Such permutations can be implemented without knowing $s^n$, hence this coding scheme does not violate the causality requirement. 

 From (\ref{eq:LcorrErrC}), we see that 
\bieee 
\cerr(\code^\Pi)=\sum_{\pi\in\Pi_n} \mu(\pi) \cdot\E\, P_{e|\pi s^n}^{(n)}(\code) \;,
\eieee
for all $s^n\in\Sset^n$ with $l^n(s^n)\leq \Lambda$. Therefore, together with (\ref{eq:ALdetErrC}), we have that the probability of error of the random code $\code^\Pi$ is bounded by 
\bieee 
\label{eq:LrandErr}
\err(\qn,\code^{\Pi})\leq e^{-\theta n} \;,
\eieee 
for every $\qn(s^n)\in\pLSpaceSn$. 

It is left for us to verify that the random code $\code^{\Pi}$ obeys the input constraint.
To this end, we apply Ahlswede's RT again. 
Let $m\in [1:2^{nR}]$ and $q(s)\in\pLSpaceS$, and let a sequence of i.i.d. random variables $\oS_1,\ldots,\oS_n\sim q(s)$.
Define the random variables
\bieee
\Phi_i(m)=\cost(\encs(U_i(m),\oS_i)) \,,\;\text{for $i\in [1:n]$}\;.
\eieee
 Then, $\Phi_1(m),\ldots,\Phi_n(m)$ are i.i.d. as well, and by (\ref{eq:LrAVcosti}), $\E \Phi_1(m)\leq \plimit-2\delta$. Hence,
for every $m\in [1:2^{nR}]$ and $q\in\pLSpaceS$, 
\begin{align*}
&\prob{\cost^n(\encs^n(U^n(m),\oS^n))>\plimit-\delta} 
= \prob{\frac{1}{n} \sum_{i=1}^n \Phi_i(m)>\plimit-\delta }\leq 2^{-n\cdot\overline{\dE}(\plimit,\Lambda)} \;,
\end{align*}
where
$
\overline{\dE}(\plimit,\Lambda) \triangleq \min\limits_{m\in [1:2^{nR}],q\in\pLSpaceS} \min\limits_{P_{\Phi'} \,:\; \E \Phi'>\plimit-\delta} D(P_{\Phi'} || P_{\Phi_1(m)}) 
$, 
by standard large deviations considerations (see \eg \cite[pp. 362--364]{CoverThomas:06b}).
On the other hand,
\begin{align}
&\prob{\cost^n(\encs^n(U^n(m),\oS^n))>\plimit-\delta} 
= \sum_{s^n\in\Sset^n} q^n(s^n) h_m(s^n) \;,
\end{align}
where $h_m(s^n)=\prob{\cost^n(\encs^n(U^n(m),s^n))>\plimit-\delta}$.
Thus,  by  Lemma~\ref{lemm:LRT},
\bieee
\label{eq:Lhbound}
\frac{1}{n!} \sum_{\pi\in\Pi_n} h_m(\pi s^n) &&\leq (n+1)^{|\Sset|}\cdot 2^{-n\cdot\overline{\dE}(\plimit,\Lambda)} \leq e^{-\theta' n} \;, 
\eieee
for all $s^n\in\Sset^n$ with $l^n(s^n)\leq\Lambda$, 
 for some $\theta'>0$ and sufficiently large $n$.

Then, 
\begin{align}
&\frac{1}{n!} \sum_{\pi\in\Pi_n} \E\, \cost^n(\encs^n(U^n(m),\pi s^n)) \nonumber\\
=& \frac{1}{n!} \sum_{\pi\in\Pi_n}
 h_m(\pi s^n) 
\cdot  \E\left(\, \cost^n(\encs^n(U^n(m),\pi s^n)) \,\Big|\; \cost^n(\encs^n(U^n(m),\pi s^n))>\plimit-\delta
 \,\right)
\nonumber\\
&+ \frac{1}{n!} \sum_{\pi\in\Pi_n}
(1- h_m(\pi s^n)) 
\cdot  \E\left(\, \cost^n(\encs^n(U^n(m),\pi s^n)) \,\Big|\; \cost^n(\encs^n(U^n(m),\pi s^n))\leq\plimit-\delta
 \,\right) \;.
\label{eq:LrandInC}
\end{align}
To bound the first sum in the RHS of (\ref{eq:LrandInC}), we use (\ref{eq:Lhbound}) and the fact that $\cost^n(x^n)\leq \cost_{max}$, for all 
$x^n\in\Xset^n$.
As for the second sum in the RHS of (\ref{eq:LrandInC}),  observe that the expectation in the last line is bounded by $(\plimit-\delta)$.
Hence, 
\bieee
\frac{1}{n!} \sum_{\pi\in\Pi_n} \E\,\cost^n(\encs^n(U^n(m),\pi s^n))
&\leq& \cost_{max}\cdot e^{-\theta' n}  +\plimit-\delta \;.
\eieee
It follows that for a sufficiently large $n$,
\begin{align}
\frac{1}{n!} \sum_{\pi\in\Pi_n} \E\, \cost^n(f_\pi^n(m,s^n)) 
=&\frac{1}{n!} \sum_{\pi\in\Pi_n} \E\, \cost^n(\encs^n(U^n(m),\pi s^n))
\leq \plimit \;,
\label{eq:LinB}
\end{align}
where the 
equality is due to (\ref{eq:LCpi}),
and the fact that the input constraint is additive (see (\ref{eq:LInConstraintStrict})).

Thus, it  follows from (\ref{eq:LrandErr}) and (\ref{eq:LinB}) that $\code^\Pi$ is a $(2^{nR},n,e^{-\theta n})$ random 
 code for the AVC $\avc$ with causal SI at the encoder, under input constraint $\plimit$ and state constraint $\Lambda$. 
\qed

\subsubsection{Upper Bound}
Assume to the contrary that there exists an achievable rate 
$
R> 
\inR^{\rstarC}_{up,\plimit+\delta,\Lambda-\delta}(\avc) 
$, 
 using random codes over the AVC $\avc$ with causal SI, under input constraint $\plimit$ and state constraint $\Lambda$,
where $\delta>0$ is arbitrarily small. 
That is, for every $\eps>0$ and sufficiently large $n$,
there exists a $(2^{nR},n)$ random code $\code^\Gamma=(\mu,\Gamma,\{\code_\gamma\}_{\gamma\in\Gamma})$ for the AVC $\avc$ with causal SI, such that 
\bieee
&& \sum_{\gamma\in\Gamma} \mu(\gamma) \left[ 
 \sum_{s^n\in\Sset^n} q^n(s^n) \cost^n(\enc^n_\gamma(m,s^n)) \right] \leq\plimit
\;,
\label{eq:LconvRandInC}
\\
&& \err(q^n,\code^\Gamma)\leq\eps \;,
\label{eq:StateConverse1b}
\eieee
 for all $m\in [1:2^{nR}]$ and $q^n(s^n)\in\pLSpaceSn$. 
	In particular, for a kernel,
	$
	\cerr(\code^\Gamma)\leq\eps 
	$, 
	for all $s^n\in\Sset^n$ such that $l^n(s^n)\leq\Lambda$.
	
	Consider using the random code $\code^\Gamma$ over the compound channel $\avc^{\overline{\pSpace}_{\Lambda-\delta}(\Sset)}$ with causal SI under input constraint $\plimit+\delta$, where $\delta>0$ is arbitrarily small. Let $\oq(s)\in\overline{\pSpace}_{\Lambda-\delta}(\Sset)$ be a given state distribution. Then, 
	 define a sequence of i.i.d. random variables $\oS_1,\ldots,\oS_n\sim \oq(s)$.  
	Letting 
	$\oq^n(s^n)\triangleq\prod_{i=1}^n \oq(s_i)$, the probability of error is bounded by
	\bieee
	\err(\oq,\code^\Gamma)
	&\leq& 
	\sum_{s^n\,:\; l^n(s^n)\leq\Lambda} \oq^{n}(s^n) \cerr(\code^\Gamma)+\prob{l^n(\oS^{n})>\Lambda} \;.\qquad 
	\eieee
	Then, the first sum is bounded by (\ref{eq:StateConverse1b}), and the second term vanishes as well by the law of large numbers, since
	$\oq(s)\in\overline{\pSpace}_{\Lambda-\delta}(\Sset)$. 

	As for the input constraint, define a random variable $L\in\Gamma$, with $L\sim\mu(\ell)$. Then, for every $m\in [1:2^{nR}]$,
	\begin{align}
	&\sum_{\gamma\in\Gamma} \mu(\gamma) \sum_{s^n\in\Sset^n} \oq^{n}(s^n)
	 \cost^n(\enc_\gamma^n(m,s^n)) 
	=
	\E_{\oq} \, \cost^n(\enc^n_L(m,\oS^{n})) \\
  =& \prob{l^n(\oS^{n})\leq\Lambda} \cdot \E_{\oq} \hspace{-0.1cm} \left( \cost^n(\enc^n_L(m,\oS^{n})) \,\big|\;				l^n(\oS^{n})\leq\Lambda \right) \nonumber\\
	&+\prob{l^n(\oS^{n})>\Lambda} \cdot \E_{\oq} \hspace{-0.1cm}  \left( \cost^n(\enc^n_L(m,\oS^{n})) \,\big|\;				l^n(\oS^{n})>\Lambda \right) \\
	\leq&  \E_{\oq} \hspace{-0.1cm} \left( \cost^n(\enc^n_L(m,\oS^{n})) \,\big|\;				l^n(\oS^{n})\leq\Lambda \right)+ \cost_{max}\cdot \eps_n 
	\leq \plimit+\delta \;,
	\label{eq:LconvIC2}
	\end{align}
	with $\eps_n\rightarrow 0$ as $n \rightarrow \infty$. The first inequality follows from the law of large numbers, and last inequality is obtained by 	
	applying (\ref{eq:LconvRandInC}) to the  state distribution 
	$
	q^n(s^n)=\cprob{\oS^{n}=s^n}{l^n(\oS^{n})\leq\Lambda} 
	$, 
	which is readily seen to satisfy $q^n\in\pLSpaceSn$.

	It follows that the random code $\code^\Gamma$ achieves a rate $R> \inR^{\rstarC}_{up,\plimit+\delta,\Lambda-\delta}(\avc)$ over the compound channel $\avc^{\overline{\pSpace}_{\Lambda-\delta}(\Sset)}$ under  input constraint $\plimit+\delta$, for an arbitrarily small $\delta>0$, in contradiction to Lemma~\ref{lemm:ALCompoundPC}. 
	We deduce that the assumption is false, and $R>\LrICav$ cannot be achieved.
\qed

\subsection*{Part 2}
For $\plimit=\cost_{max}$, it follows from  (\ref{eq:ALcvCIcoro}) and  (\ref{eq:ALcvRIcoro}) that
$\LrICav=\LrIRav=\min\limits_{q(s)\in\pLSpaceS} \max\limits_{p(u),\encs(u,s)} I_q(U;Y)$. Hence, the proof follows from part 1. \qed

\section{Ahlswede's Elimination Technique}    
\label{app:LET}

\begin{proof}[Proof of Lemma~\ref{lemm:LcorrSizeC}]
The proof is an extension of \cite[Section~4]{Ahlswede:78p}. 
Consider the AVC $\avc$ with causal SI, under per message 
 input constraint $\plimit$ and state constraint $\Lambda$.
 Let $\dK>0$ be an integer, chosen later, and define the random variables
\bieee
\label{eq:LLi}
L_1,L_2,\ldots,L_{\dK} \;\,\text{i.i.d. $\sim\mu(\ell)$} \;.
\eieee
Fix $m\in [1:2^{nR}]$ and  $s^n\in\Sset^n$, and define the random variables
\begin{align}
&\Phi_j(m,s)=\cost^n( f_{L_j}^n(m,s^n)) \;,\quad j\in [1:\dK] \;,
\intertext{
and
}
&\Psi_j(s^n)= \cerr(\code_{L_j}) \;,\quad j\in [1:\dK] \;,
\end{align}
which correspond to the code $\code_{L_j}=(f_{L_j}^n,g_{L_j})$ in the code collection $\{\code_\gamma=(f^n_\gamma,g_\gamma)\}_{\gamma\in\Gamma}$. 
Since $\code^\Gamma$ is a $(2^{nR},n,\eps_n)$ random code over the AVC $\avc$ with causal SI, under per message 
 input constraint $\plimit$ and state constraint $\Lambda$, we have that 
$ \sum_\gamma\mu(\gamma)\sum_{s^n} \qn(s^n) \cerr(\code_\gamma)\leq \eps_n$, for all $\qn(s^n)\in\pLSpaceSn$. In particular, for a kernel, we have that  for a given $m\in [1:2^{nR}]$ and $s^n\in\Sset^n$ with $l^n(s^n)\leq\Lambda$,
\bieee
\label{eq:LPhiIneqC}
&&\E \Phi_j(m,s^n)=\sum_{\gamma\in\Gamma} \mu(\gamma) \cost^n(f_\gamma^n(m,s^n)) \leq \plimit \;,
\intertext{
and
}
\label{eq:LPsiIneqC}
&&\E \Psi_j(s^n)=\sum_{\gamma\in\Gamma} \mu(\gamma)\cdot \cerr(\code_\gamma) \leq \eps_n \;,
\eieee
for all $j\in[1:k]$. 
Now take $n$ to be large enough so that $\eps_n<\alpha$. 

Consider the code $\code^{\Gamma^*}=(\mu^*,\Gamma^*=[1:\dK],\{\code_{L_j}\}_{j=1}^k)$ formed by a random collection of codes, with $\mu^*(j)=\frac{1}{\dK}$. The event that a ``bad code" is chosen is bounded by the union of the following events.
Denote the event that the input constraint is violated 
 by
\bieee
\Aset_1&=\Bigg\{& \frac{1}{\dK} \sum_{j=1}^{\dK} \Phi_j(m,s^n)> \plimit+\delta 
\,,\; 
\label{eq:LETa1}
\IEEEnonumber\\
&&\text{for some $(m,s^n)\in [1:2^{nR}]\times \Sset^n$ with $l^n(s^n)\leq\Lambda$}
\Bigg\}\;,
\intertext{and denote the event that the error requirement is violated by}
\Aset_2&=\Bigg\{& \frac{1}{\dK} \sum_{j=1}^{\dK} \Psi_j(s^n)\geq \alpha \,,\; \text{for some $s^n\in\Sset^n$ with $l^n(s^n)\leq\Lambda$}
 \Bigg\} \;,
\label{eq:LETa2}
\eieee
where $0<\alpha<1$ and $\delta>0$ are arbitrarily small. 
Then, by the union of events bound
\begin{align}
&\prob{\Aset_1\cup\Aset_2} 
\leq \prob{\Aset_1}+\prob{\Aset_2} 
\;.
\label{eq:LETa1Ua2}
\end{align}

Keeping $m$ and $s^n$ fixed, the random variables $\Phi_j(m,s^n)$ and $\Psi_j(s^n)$  are each i.i.d., due to (\ref{eq:LLi}).
Consider the first term in the RHS of (\ref{eq:LETa1Ua2}),  $\prob{\Aset_1}$ (see (\ref{eq:LETa1})). By standard large deviations considerations,
we have that
\bieee
\prob{\frac{1}{\dK}\sum_{j=1}^{\dK} \Phi_j(m,s^n)> \plimit+ \delta}&\leq&  2^{-\dK \cdot (\dE_j(s^n)-\eps')} \;,
\eieee
with
\bieee
\dE_j(s^n)\triangleq \min_{P_{\Phi'}\,:\; \E \Phi'>\plimit+\alpha} D(P_{\Phi'}||P_{\Phi_1(m,s^n)}) \;,
\eieee
 (see \eg \cite[pp. 362--364]{CoverThomas:06b}), where $\eps'>0$ is arbitrarily small. 
Thus, the first term in the RHS of (\ref{eq:LETa1Ua2}) is bounded by 
\bieee
\prob{\Aset_1}&\leq& 
\sum_{m\in [1:2^{nR}]}\sum_{s^n\in\Sset^n \,:\; l^n(s^n)\leq\Lambda}\prob{ \frac{1}{k} \sum_{j=1}^{\dK} \Phi_j(m,s^n)> \plimit+\delta} \\
&\leq&
2^{nR}\cdot |\Sset|^n\cdot 2^{-\dK \cdot (\dE_j(s^n)-\eps')} \;.
\label{eq:eBound1}
\eieee
Since $2^{nR}\cdot|\Sset|^n$ grows only exponentially in $n$, choosing
 \bieee
\dK=n^2
\eieee
 results in a super exponential decay.

As for the second term in the RHS of (\ref{eq:LETa1Ua2}), $\prob{\Aset_2}$ (see (\ref{eq:LETa2})). The technique known as  Bernstein's trick \cite{Ahlswede:78p} is now applied.  
\bieee
\prob{\sum_{j=1}^{\dK} \Psi_j(s^n)\geq \dK\alpha} &\stackrel{(a)}{\leq}&
\E\left\{ \exp\left[ \beta \left(
\sum_{j=1}^\dK \Psi_j(s^n)- k\alpha
\right)
\right]
\right\} \\
&=& e^{-\beta \dK \alpha}\cdot
\E\left\{  \prod_{j=1}^\dK e^{\beta  \Psi_j(s^n)} \right\} \\
&\stackrel{(b)}{=}& e^{-\beta \dK \alpha}\cdot
\prod_{j=1}^\dK \E\left\{   e^{\beta  \Psi_j(s^n)} \right\} \\
&\stackrel{(c)}{\leq}& e^{-\beta \dK \alpha}\cdot
\prod_{j=1}^\dK \E\left\{   1+e^\beta \cdot\Psi_j(s^n) \right\} \\
&\stackrel{(d)}{\leq}& e^{-\beta \dK \alpha}\cdot
 \left(   1+e^\beta \eps_n \right)^\dK 
\eieee
where $(a)$ is an application of Chernoff's inequality; $(b)$ follows from the fact that $\Psi_j(s^n)$ are independent; $(c)$ holds since $e^{\beta x}\leq1+e^\beta x$, for $\beta>0$ and $0\leq x\leq 1$; $(d)$ follows from (\ref{eq:LPsiIneqC}). We take $n$ to be large enough for $1+e^\beta \eps_n\leq e^\alpha$ to hold. Thus, choosing $\beta=2$, we have that 
\bieee
\prob{\frac{1}{\dK}\sum_{j=1}^\dK \Psi_j(s^n)\geq \alpha} &\leq&
e^{- \alpha k }=e^{-\alpha n^2} \;,
\eieee
for all $s^n\in\Sset^n$ with $l^n(s^n)\leq\Lambda$. Hence, the second term in the RHS of (\ref{eq:LETa1Ua2})
is bounded by 
\bieee
\prob{\Aset_2}\leq
\sum_{s^n\in\Sset^n \,:\; l^n(s^n)\leq\Lambda} \prob{  \frac{1}{\dK} \sum_{j=1}^\dK \Psi_j(s^n) \,\geq\, \alpha
 } 
& \leq& |\Sset|^n \cdot e^{-\alpha n^2} \;.
\label{eq:eBound2}
\eieee

By (\ref{eq:LETa1Ua2}), (\ref{eq:eBound1}) and (\ref{eq:eBound2}), we have that probability that either the input constraint or the error requirement are violated decays super exponentially with  blocklength, namely $\prob{\Aset_1\cup\Aset_2} \sim 2^{-\theta n^2}$, for some $\theta>0$.
 It follows that 
%
 there exists a random code $\code^{\Gamma^*}=(\mu^*,\Gamma^*,\{\code_{\gamma_j}\}_{j=1}^k)$ 
 for the AVC $\avc$, such that for all $m\in [1:2^{nR}]$ and $q^n\in\pLSpaceSn$,
\begin{IEEEeqnarray}{l}
\sum_{\gamma\in\Gamma^*}\mu^*(\gamma) \sum_{s^n\in\Sset^n} q^n(s^n)
 \cost^n(\enc^n_\gamma(m,s^n))\leq\plimit+\delta \,,\;\text{}
\intertext{
and
} 
\err(\qn, \code^{\Gamma^*})=\sum_{s^n\in\Sset^n} \qn(s^n) \cerr(\code^{\Gamma^*})\leq \alpha \;,
\end{IEEEeqnarray}
as we were set to prove.
\end{proof}

\section{Proof of Theorem \ref{theo:ALavcCstateC}}
\label{app:ALavcCstateC}

\subsection*{Part 1}

Consider the AVC $\avc=\{\channel\}$ with causal SI, under average input constraint $\plimit$ and
  state constraint $\Lambda$. 
Then, for every encoding mapping $\encs(u,s)$, consider the AVC $\Uavc=\{\xichannel\}$ \emph{without SI}, under state constraint $\Lambda$, and \emph{free} of input constraint.
%
Hence, any coding scheme employed over the AVC $\Uavc$ without SI can also be employed  over the AVC $\avc$ with causal SI, using the encoding function $\encs(u,s)$,  provided that the input constraint $\plimit$ on $\avc$ is satisfied.

Let a type $P_U(u)$ and a function $\encs(u,s)$  achieve $\inR_{low,\plimit-2\eps,\Lambda+\eps}(\avc)$, where $\eps>0$ is arbitrarily small.
 Hence,
$P_U(u)\in  \pSpace_{\plimit-2\eps,\Lambda+\eps,\encs}(\Uset)$. 
 By \cite[Theorem 2]{CsiszarNarayan:88p}, for every $\eps>0$ and sufficiently large $n$, if
\bieee
R< \min_{q(s)\in\apLSpaceS} I_q(U;Y) \;,
\eieee
then there exists a $(2^{nR},n,\eps_1)$ code $\code_{\encs}=(u^n(m),g(y^n))$  over the AVC $\Uavc$ without SI, under state constraint $\Lambda$.
The code constructed in \cite{CsiszarNarayan:88p} is formed by a random selection of $2^{nR}$ independent codewords $u^n(m)$, for $m\in[1:2^{nR}]$, with uniform distribution over the type class of $P_U$ 
 (see proof of Lemma 3 in \cite{CsiszarKorner:82b}).

Consider the code $\code'$ over the AVC $\avc$, as described below.

\emph{Encoding:}
 To send a  message $m\in [1:2^{nR}]$, do as  as follows.
If  
\bieee
\label{eq:LdetC1}
\frac{1}{2^{nR}}\sum_{\tm=1}^{2^{nR}}
\cost^n(\encs^n(u^n(\tm),$ $\tsn))\leq\plimit \;,
\eieee
 for all $\tsn\in\Sset^n$ with $l^n(\tsn)\leq\Lambda$ then, at time $i\in [1:n]$, transmit 
$
x_i=\encs(u_i(m),s_i) 
$. 
Otherwise, 
 transmit 
$
x^n=(a,\ldots,a) 
$, 
 with an idle input symbol $a\in\Xset$, with $\cost(a)=0$.

\emph{Decoding:} 
Use 
 the decoder of the original code $\code_{\encs}$, namely 
$
\hm=g(y^n) 
$. 

\emph{Analysis of Probability of Error:} 
%
 Assume without loss of generality that the user sent the message $M=1$. 
Denote $\dM\triangleq 2^{nR}$. For every $s^n\in\Sset$, define a sequence of $\dM$ random variables
given by $Z_m(s^n)=\cost^n(\encs^n(U^n(m),s^n))$, for $m\in [1:\dM]$, and
consider the event
\begin{align}
 & \Eset_1=
\bigg\{\, \sum_{m=1}^{2^{nR}} Z_m(s^n)
>\plimit \,,\;\text{for some $s^n\in\Sset^n$ with 
$l^n(s^n)\leq\Lambda$} 
\,\bigg\} \;.
\label{eq:ALEset0TiDet}
\end{align}
 Then, the probability of error $\err(q^n,\code')=\prob{g(Y^n)\neq 1}$ is bounded as follows,
\begin{align}
\err(q^n,\code') 
=&\prob{\Eset_1}\cdot\cprob{g(Y^n)\neq 1}{\Eset_1} 
+\prob{\Eset_1^c}\cdot\cprob{g(Y^n)\neq 1}{\Eset_1^c} \nonumber\\
\leq& \prob{\Eset_1}+\cprob{g(Y^n)\neq 1}{\Eset_1^c} \;,
\label{eq:ALerrAVCLow1}
\end{align}
where the conditioning on $M=1$ is omitted to simplify notation. 

Now, we bound the first term in the RHS of (\ref{eq:ALerrAVCLow1}). Fix $s^n\in\Sset^n$ with $l^n(s^n)\leq\Lambda$. 
 Recall that  $U^n(1),\ldots,U^n(\dM)$ is a sequence of vectors that are independent of each other, where each vector has the same distribution.
 Therefore, for every given $s^n\in\Sset^n$ with $l^n(s^n)\leq\Lambda$, the sequence $Z_1(s^n),\ldots,Z_{\dM}(s^n)$ is i.i.d., hence
\begin{align}
\prob{\sum_{m=1}^{\dM} Z_m(s^n)>\plimit} 
\leq 2^{-\dM\cdot \dF_0(\plimit)}=2^{-2^{nR}\cdot \dF_0(\plimit)} \;,
\end{align}
where
$
\dF_0(\plimit) \triangleq \min\limits_{
 s^n\in\Sset^n \,:\;  l^n(s^n)\leq\Lambda
}
 \min\limits_{
P_{Z'} \,:\; \E Z'>\plimit} D(P_{Z'} || P_{Z_m(s^n)}) 
>0$, 
by standard large deviations considerations (see \eg \cite[pp. 362--364]{CoverThomas:06b}).
Thus, applying the union bound to (\ref{eq:ALEset0TiDet}), we have that
\bieee
\prob{\Eset_1}&\leq& |\Sset^n|\cdot  \max_{s^n\in\Sset^n \,:\;  l^n(s^n)\leq\Lambda} \prob{\sum_{m=1}^{\dM} Z_m(s^n)>\plimit} \leq
|\Sset^n|\cdot 2^{-2^{nR}\cdot \dF_0(\plimit)} \;.
\eieee
Hence, $\prob{\Eset_1}$ decays to zero double exponentially as $n\rightarrow\infty$.

As for the second term in the RHS of (\ref{eq:ALerrAVCLow1}), the probability $\cprob{g(Y^n)\neq 1}{\Eset_1^c}$ vanishes as well, due to the following. Given that the event $\Eset_1^c$ occurred, we have that 
$
X^n=\encs^n(U^n(1),S^n) 
$. 
Then, applying the results by \cite{CsiszarNarayan:88p}, we have that $\code_{\encs}$ is a $(2^{nR},n,\eps_1)$ code over the AVC $\Uavc$, where $\eps_1>0$ is arbitrarily small. It thus follows that $\cprob{g(Y^n)\neq 1}{\Eset_1^c}\leq\eps_1$.
%
\qed

\subsection*{Part 2}
The converse part is a direct consequence of Theorem~\ref{theo:ALavcCr}, by which $\LCavc\leq\LrCav=\LrICav=\LrIRav$, for $\plimit=\cost_{max}$.
In the proof of the direct part,  
  the lemma below is used as a tool. 
\begin{lemma} \cite{CsiszarNarayan:88p}
\label{lemm:RPencsPos}  
If $\channel(y|x,s)$ is non-symmetrizable, then  
 for every $p\in\pSpace(\Xset)$ with $p(x)>0$ for all $x\in\Xset$, 
we have that
$
\min_{q\in\pSpace(\Sset)} I_q(X;Y)>0 
$. 
\end{lemma}

Now, assume that there exists a function $\encs'(u,s)$, such that 
$V_{Y|U,S}^{\encs'}(y|u,s)=\channel(y|\encs'(u,s),s)$ is non-symmetrizable. 
  We  show that every rate $R<\LrICav$ can be achieved.  
%
The assumption above, along with Lemma~\ref{lemm:RPencsPos} and 
\cite[Theorem 2]{CsiszarNarayan:88p}, 
imply that 
 the capacity without constraints is positive, \ie 
$
\opC_{\cost_{max},l_{max}}(\avc)>0 
$. 
 This, in turn,  allows us to use Ahlswede's ET \cite{Ahlswede:78p} 
 using the random code constructed in the proof of Theorem~\ref{theo:ALavcCr} to construct a deterministic code (see  \cite[Section V]{CsiszarNarayan:88p}).

Let $R<\inR^{\rstarC}_{low,\plimit,\Lambda
\,}\hspace{-0.1cm}(\avc)$. 
 By Theorem~\ref{theo:ALavcCr},  for some $\theta>0$ and sufficiently large $n$,  
 there exists a $(2^{nR},n,e^{-\theta n})$ random code for the AVC $\avc$ with causal SI, under  
state constraint $\Lambda$. 
Thus, by Lemma~\ref{lemm:LcorrSizeC},  for every $\eps_1>0$ and sufficiently large $n$, 
 there exists a $(2^{nR},n,\eps_1)$ random  code 
$
\code^\Gamma= \bigg(\, \mu(\gamma)=\frac{1}{k} \,,\; \Gamma=[1:k] \,,\;  \{\code_\gamma=(\enc_\gamma^n,g_\gamma)  \}_{\gamma\in \Gamma}  \,\bigg) 
$, 
 for the AVC $\avc$ under 
 state constraint $\Lambda$, 
with 
$
k=|\Gamma|\leq n^2 
$. 

Next, we claim that the code index $\gamma\in [1:k]$ can be reliably sent  over the AVC $\avc$ with causal SI, under 
 state constraint $\Lambda$. 
Consider a 
 code for the index $\gamma\in [1:k]$, with a blocklength $\nu$ and rate $\bR$.
Since $k$ is polynomial at most, 
 such a code requires a negligible blocklength, \ie
 $
\nu=o(n) 
$. 
Therefore, the jammer is virtually free of state constraints during this transmission. 
However, as deduced above, 
 the capacity without state constraints is positive, under the assumptions of part 2
 of the theorem, and thus 
 for every $\eps_2>0$ and sufficiently large $\nu$, there exists a $(2^{\nu\bR},\nu,\eps_2)$ deterministic code 
$ 
\code_{\text{i}}=(\tfnu,\gnu)  
$ to send $\gamma\in [1:k]$, where $\nu=o(n)$ and $\bR>0$. 

Now, consider a code 
 formed by the concatenation of $\code_{\text{i}}$ as a prefix to a corresponding code in the code collection $\{\code_\gamma\}_{\gamma\in\Gamma}$. The encoder sends both $\gamma$ and $m$, 
 by transmitting $\tfnu(\gamma,s^\nu)$ and then 
 $x^n=\enc_\gamma^n(m,s_{\nu+1},\ldots,s_{\nu+n})$. 
Subsequently, decoding is performed in two stages as well; the index is estimated first, with  
$\hgamma=\gnu(y_1,\ldots,y_{\nu})$, and the message is then estimated by  
$\widehat{m}=g_{\hgamma}(y_{\nu+1},\ldots,y_{\nu+n})$.  
By the union of events bound, the probability of error is then bounded by $\eps=\eps_1+\eps_2$. 
That is, the concatenated code is a $(2^{(\nu+n)\tR_n},\nu+n,\eps)$ code over the AVC $\avc$ with causal SI, under  
 state constraint $\Lambda$, where $\nu=o(n)$, and  
  the rate   $\tR_n=\frac{n}{\nu+n}\cdot R$ approaches $R$ as $n\rightarrow \infty$. 
\qed


\section{Analysis of Example~\ref{example:LAVNTC}}
\label{app:LAVNTC}
We rely on the analysis of Erez and Zamir in \cite{ErezZamir:00p}. They considered Shannon's model \cite{Shannon:58p} of a channel with random parameters with causal SI, where the state sequence $S^n$ is i.i.d. according to a given distribution $q(s)$. In \cite{ErezZamir:00p},
Erez and Zamir consider a modulo-additive channel, 
\bieee
\label{eq:AdditiveModulo}
Y=X+Z_S \mod |\Xset| \;,
\eieee
with $\Xset=\Zset=\Yset=\{0,1,\ldots,|\Xset|-1\}$, such that given $S=s$, the additive noise is distributed according to 
$
Z_s\sim p(z|s) 
$. 
Let $\Uset$ be the index set for the set of all functions $\encs_u: \Sset\rightarrow\Xset$. 
It is shown in \cite{ErezZamir:00p} that the capacity of the modulo-additive random parameter channel $\rp$ with causal SI is given by
\bieee
\label{eq:ModuloAdditiveC}
\ICrp=\log|\Xset|-\min_{u\in\Uset} H(Z_S-\encs_u(S)) \;.
\eieee
For $u\in\Uset$ that achieves the minimum above, $\encs_u(S)$ is interpreted as the minimum error-entropy predictor of $Z_S$. 
The DMC $\channel$ in Example~\ref{example:LAVNTC} is a special case of their model.

First, consider the arbitrarily varying  noisy-typewriter channel $\avcig$ without SI, under a state constraint $\Lambda$, 
when free of input constraints, \ie $\plimit=\cost_{max}$. We calculate the 
random code capacity given by (\ref{eq:Lminimax}), due to \cite{CsiszarNarayan:88p1}.
Consider a given $0\leq q\leq 1$, and let 
\bieee
S=\begin{cases}
1 &\text{w.p. $1-q$}\;,\\
2 &\text{w.p. $q$}\;.
\end{cases}
\eieee
 The entropy of the additive noise $Z$ is then given by 
\bieee
H_q(Z)=h(\theta)+\theta h(q) \;,
\eieee
hence,
\bieee
\label{eq:Lex2ICrpig}
\ICrpig\triangleq\max_{p(x)} I_q(X;Y) 
=\log 3-h(\theta)-\theta h(q) \;.
\eieee
Minimizing over $0\leq q\leq \Lambda$ yields
\bieee
\label{eq:Lex2LrICavig}
\LrICavig=\min_{0\leq q\leq \Lambda} \ICrpig=\begin{cases}
\log 3 -h(\theta)-\theta 	h(\Lambda)							&\text{if $0\leq \Lambda\leq \frac{1}{2}$}									\;,\\
\log 3 -h(\theta)-\theta 													&\text{if $\Lambda\geq\frac{1}{2}$}									\;.
\end{cases}\qquad
\eieee
and by Theorem~\ref{theo:LavcC0R}, due to \cite{CsiszarNarayan:88p1}, 
 the random code capacity of the AVC $\avcig$ without SI, under state constraint $\Lambda$, is given by
 $\LrCavig=\LrICavig$.

We now claim that $\avcig$ is non-symmetrizable for all $\theta\neq\frac{2}{3}$, which will imply that the deterministic code capacity is given by 
$\LCavcig=\LrICavig$, by Theorem~\ref{theo:LavcC0stateC}, due to \cite{CsiszarNarayan:88p}.
Assume to the contrary that $\avcig$ is symmetrizable and there exists $J(s|x)$ that satisfies (\ref{eq:symmetrizable}). In particular, denoting 
$\alpha_x=J(2|x)$ for $x\in\{0,1,2\}$, we have that both of the following relations hold for $y\in\{0,1,2\}$,
\begin{subequations}
\label{eq:EX1symm}
\bieee
&&(1-\alpha_1)\cdot \channel(y|0,1)+\alpha_1\cdot \channel(y|0,2) \IEEEnonumber\\&&\quad= 
  (1-\alpha_0)\cdot \channel(y|1,1)+\alpha_0\cdot \channel(y|1,2) \;,
\eieee
and
\bieee
&&(1-\alpha_2)\cdot \channel(y|0,1)+\alpha_2\cdot \channel(y|0,2) \IEEEnonumber\\&&\quad=   
  (1-\alpha_0)\cdot \channel(y|2,1)+\alpha_0\cdot \channel(y|2,2) \;.
\eieee
\end{subequations}
 Taking $y=0$, we have
$
1-\theta=\alpha_0\cdot \theta=(1-\alpha_0)\cdot\theta 
$. 
Since $\theta>0$, this 
 can only hold for $\alpha_0=\frac{1}{2}$ and $\theta=\frac{2}{3}$. Thus, for $\theta\neq\frac{2}{3}$, the AVC $\avcig$ without SI is non-symmetrizable, and by Theorem~\ref{theo:LavcC0stateC}, $\LCavcig=\LrICavig$. 

For $\theta=\frac{2}{3}$, we have that 
\bieee
\LrICavig=
\begin{cases}
\frac{2}{3}(1-h(\Lambda))						&\text{if $0\leq \Lambda\leq \frac{1}{2}$}									\;,\\
0													&\text{if $\Lambda\geq\frac{1}{2}$}									\;.
\end{cases}
\eieee
Since the capacity without constraints is zero, Theorem~\ref{theo:symm0} implies that $\channel$ is symmetrizable for this value of $\theta$. Substituting $y=0$ and $y=1$ in (\ref{eq:EX1symm}), we find that $\channel$ can only be symmetrized by $J(s|x)$ such that $\alpha_0=\alpha_1=\alpha_2=\frac{1}{2}$, hence $\sum_{x,s} p(x)J(s|x) l(s)=\frac{1}{2}$ for all $p$.
It then follows that $\LCavcig=\LrICavig$.
Therefore, when SI is not available,  $\LCavcig=\LrICavig$ for all values of $\theta>0$, and the capacity is thus given by  
 (\ref{eq:Lex2Cavcig}).  

Now, consider the arbitrarily varying noisy-typewriter channel $\avc$ with causal SI, under state constraint $\Lambda$.  
We use the formula  in
 (\ref{eq:ModuloAdditiveC}) (by \cite{ErezZamir:00p}) to find an explicit expression for $\ICrp$. There are nine mappings $\encs_u:\Sset\rightarrow\Xset$. 
For $\encs_1(s)=0$, $\encs_2(s)=1$ and $\encs_3(s)=2$, we have 
\bieee
H(Z-\encs_u(S))=H(Z)=h(\theta)+\theta h(q) \,,\; u=1,2,3.
\eieee
For $\encs_4(s)=s$, $\encs_5(s)=s+1$ and $\encs_6(s)=s+2$, we have 
\bieee
H(Z-\encs_u(S))=H\left(  (K-1)\cdot S \right)=h(\theta)+(1-\theta) h(q) \,,\; u=4,5,6.\quad
\eieee
For $\encs_7(s)=2s$, $\encs_8(s)=2s+1$ and $\encs_9(s)=2s+2$, we have 
\bieee
H(Z-\encs_u(S))=H\left(  (K-2)\cdot S \right)=h(\theta*q) \,,\; u=7,8,9 \;,
\eieee
where $\theta*q=\theta(1-q)+(1-\theta )q$. Therefore,
\bieee
\ICrp &=&\log 3-\min\left(\, h(\theta)+\theta h(q) \,,\; h(\theta)+(1-\theta) h(q) \,,\; h(\theta*q)  \right)\;.
\eieee
Therefore, 
\begin{align}
\label{eq:LNTrC}
&\LrICav \nonumber\\
=&\begin{cases}
\log 3-\min\left(\, h(\theta)+\theta h(\Lambda) \,,\; h(\theta)+(1-\theta)h(\Lambda) \,,\; h(\theta*\Lambda)  \right) 
&\text{if $0\leq \Lambda< \frac{1}{2}$}\;,\\
\log 3-\min\left(\, h(\theta)+\theta  \,,\; h(\theta)+(1-\theta) \,,\; 1  \right) 
&\text{if $ \Lambda\geq \frac{1}{2}$}\;.
\end{cases}\qquad
\end{align}
and by part 2 of Theorem~\ref{theo:ALavcCr}, the random code capacity of the AVC $\avc$ with causal SI, under state constraint $\Lambda$,
is given by $\LrCav=\LrICav=\LrIRav$.

Let us examine the condition in part 2 of Theorem~\ref{theo:ALavcCstateC}. Assume to the contrary that $\Uavc$ is symmetrizable.
In particular, taking $\encs_{u_1}(s)=s$ and $\encs_{u_2}(s)=2s$, \ie $u_1=4$ and $u_2=7$, we have that for some $\beta_u$,
where $\beta_u=J(2|u)$,
\begin{multline}
(1-\beta_7)\channel(y|1,1)+\beta_7\channel(y|2,2)=\\
(1-\beta_4)\channel(y|2,1)+\beta_4\channel(y|1,2) \;.
\end{multline}  
Thus, for $y=0$, we get $0=\theta$, which contradicts our assumption that $\theta>0$, and by part 2 of Theorem~\ref{theo:ALavcCstateC},
$\LCavc=\LrICav$.
 \qed

\chapter{AVDBC with Causal SI: Proofs}
\section{Proof of Lemma~\ref{lemm:BcompoundLowerB}}
\label{app:BcompoundLowerB}
We show that every rate pair $(R_1,R_2)\in\BIRcompound$ can be achieved using deterministic codes over the compound DBC $\Bcompound$ with causal SI. 
We construct a code based on superposition coding with Shannon strategies, and decode using joint typicality  with respect to a channel state type, which is ``close" to some  $q\in\Qset$. 

Define a set $\tQ$ of state types
\bieee
\tQ= \left\{ \emp_{s^n} \,:\; s^n\in\Aset^{\delta_1}(q) \,,\;\text{for some $q\in\Qset$}  \right\} \;,
\label{eq:BtQ}
\eieee
where
\bieee
\label{eq:Bdelta1CompNoSI}
\delta_1\triangleq \frac{\delta}{2\cdot |\Sset|} \;.
\eieee
That is, $\tQ$ is the set of types that are $\delta_1$-close to some state distribution $q(s)$ in $\Qset$. 
Now, a code for the compound DBC with causal SI is constructed as follows

\emph{Codebook Generation:} Fix the distribution $P_{U_1,U_2}(u_1,u_2)=p(u_1,u_2)$ and the function $\encs(u_1,u_2,s)$. Generate $2^{nR_2}$ independent sequences at random,
\bieee
u_2^n(m_2) \sim  \prod_{i=1}^n P_{U_2}(u_{2,i}) \,,\;\text{for $m_2\in[1:2^{nR_2}]$}\;.
\eieee
For every $m_2\in[1:2^{nR_2}]$, generate $2^{nR_1}$ sequences at random,  
 \bieee
u_1^n(m_1,m_2) \sim \prod_{i=1}^n P_{U_1|U_2}(u_{1,i}|u_{2,i}(m_2)) \,,\;\text{for $m_1\in[1:2^{nR_1}]$}\;,
\eieee
conditionally independent given $u_2^n(m_2)$.

\emph{Encoding}: To send a pair of messages $(m_1,m_2)\in [1:2^{nR_1}]\times [1:2^{nR_2}]$, transmit at time $i\in[1:n]$,
\bieee
x_i=\encs\left( u_{1,i}(m_1,m_2),u_{2,i}(m_2),s_i \right) \;.
\eieee

\emph{Decoding}: Let
 \begin{align}
\label{eq:BUchannelYL1causal}
P^{q}_{U_1,U_2,Y_1,Y_2}(u_1,u_2,y_1,y_2)=\sum_{s\in\Sset} q(s)
 P_{U_1,U_2}(u_1,u_2) 
 \bc\left( y_1, y_2| \encs(u_1,u_2,s),s \right) \;.
 \end{align}
Observing $y_2^n$, decoder 2 finds a unique $\tm_2\in[1:2^{nR_2}]$ such that
\bieee
 (u_2^n(\tm_2),y_2^n)\in\tset(P_{U_2} P^{q}_{Y_2|U_2}) \;,\quad\text{for some $q\in\tQ$} \;.
\eieee  
If there is none, or more than one such $\tm_2\in[1:2^{nR_2}]$, then decoder 2 declares an error.

Observing $y_1^n$, decoder 1 finds a unique pair of messages $(\hm_1,\hm_2)\in [1:2^{nR_1}]\times [1:2^{nR_1}]$ such that
\bieee
(u_2^n(\hm_2),u_1^n(\hm_1,\hm_2),y_1^n)\in\tset( P_{U_1,U_2} P^q_{Y_1|U_1,U_2} )
 \;,\quad\text{for some $q\in\tQ$} \;.
\eieee
If there is none, or more than such pair $(\hm_1,\hm_2)$, then decoder 1 declares an error.

\emph{Analysis of Probability of Error}:
By the union of events bound,
\bieee
\label{eq:BcompLerr12}
\err(q,\code)\leq \prob{\tM_2\neq 1}+\prob{(\hM_1,\hM_2)\neq (1,1)} \;,
\eieee
where the conditioning on $(M_1,M_2)=(1,1)$ is omitted for convenience of notation.
The error event for decoder 2 is the union of the following events. 
\bieee
\Eset_{2,1} &=&\{ (U_2^n(1),Y_2^n)\notin \tset(P_{U_2} P^{q'}_{Y_2|U_2}) \;\text{ for all $q'\in\tQ$} \} \;, \\
\Eset_{2,2} &=&\{ (U_2^n(m_2),Y^n)\in \tset(P_{U_2} P^{q'}_{Y_2|U_2}) \;\text{ for some $m_2\neq 1,\, q'\in\tQ$} \} \;.
\eieee
 Then, by the union of events bound,
\bieee
\label{eq:Bub2}
 \prob{\tM_2\neq 1}\leq \prob{\Eset_{2,1}}+ \prob{\Eset_{2,2}} \;.
\eieee  
 Considering the first term, we claim that the event $\Eset_{2,1}$ implies that $(U_2^n(1),Y_2^n)\notin \Aset^{\delta}(P_{U_2} P^{q''}_{Y_2|U_2})$ for all $q''\in\Qset$. 
	Suppose that there exists $q''\in\Qset$ that satisfies 
	$(U_2^n(1),Y_2^n)\in \Aset^{\nicefrac{\delta}{2}}(P_{U_2} P^{q''}_{Y_2|U_2})$. 
	Then, for a sufficiently large $n$, there exists a type $q'(s)$ such that 
	\bieee
	|q'(s)-q''(s)|\leq \delta_1 \;.
	\eieee
	It can then be inferred that $q'\in\tQ$ (see (\ref{eq:BtQ})), and 
	\bieee
	|P_{Y_2|U_2}^{q'}(y_2|u_2)-P_{Y_2|U_2}^{q'}(y_2|u_2)|\leq |\Sset|\cdot\delta_1=\frac{\delta}{2} \;,
	\eieee
	for all $u_2\in\Uset_2$ and $y_2\in\Yset_2$ (see  (\ref{eq:Bdelta1CompNoSI}) and (\ref{eq:BUchannelYL1causal})).
	Hence,	$(U_2^n(1),Y_2^n)\in \Aset^{\delta}(P_{U_2} P^{q'}_{Y_2|U_2})$. 
	Equivalently, if  $(U_2^n(1),Y_2^n)\notin \Aset^{\delta}(P_{U_2} P^{q'}_{Y_2|U_2})$ for all $q'\in\tQ$, then 
	$(U_2^n(1),Y_2^n)\notin \Aset^{\nicefrac{\delta}{2}}(P_{U_2} P^{q''}_{Y_2|U_2})$ for all $q''\in\Qset$. 
	Thus,
	\bieee 
	\label{eq:BllnRL}
	\prob{\Eset_{2,1}}&\leq& \prob{(U_2^n(1),Y_2^n)\notin \Aset^{\nicefrac{\delta}{2}}(P_{U_2} P^{q''}_{Y_2|U_2}) \;\text{ for all $q''\in\Qset$} } \IEEEnonumber \\
	&\leq& \prob{(U_2^n(1),Y_2^n)\notin \Aset^{\nicefrac{\delta}{2}}(P_{U_2} P^{q}_{Y_2|U_2})  } \;.
	\eieee
The last expression  tends to zero exponentially as $n\rightarrow\infty$ by the law of large numbers and Chernoff's bound.
	
	Moving to the second term in the RHS of (\ref{eq:Bub2}), 
	we use the classic method of types considerations to bound $\prob{\Eset_{2,2}}$. 
	By the union of events bound 
and the fact that the number of type classes in $\Sset^n$ is bounded by $(n+1)^{|\Sset|}$ 
 \cite[Lemma 2.2]{CsiszarKorner:82b}, we have that  
\begin{align}
&\prob{\Eset_{2,2}} \nonumber\\
\leq& (n+1)^{|\Sset|}
\cdot \sup_{q'\in\tQ} \prob{
(U_2^n(m_2),Y_2^n)\in \tset(P_{U_2} P^{q'}_{Y_2|U_2}) \;\text{ for some $m_2\neq 1$} 
}\;.
\label{eq:BE2poly}
\end{align}
For every $m_2\neq 1$,
\begin{align}
\prob{
(U_2^n(m_2),Y_2^n)\in \tset(P_{U_2} P^{q'}_{Y_2|U_2}) 
}
=&  \sum_{u_2^n\in\Uset_2^n} P_{U_2^n}(u_2^n) \cdot \prob{(u_2^n,Y_2^n)\in\tset(P_{U_2} P^{q'}_{Y_2|U_2})} 													\nonumber\\
=&  \sum_{u_2^n\in\Uset_2^n} P_{U_2^n}(u_2^n) \cdot \sum_{y_2^n \,:\; (u_2^n,y_2^n)\in \tset(P_{U_2} P^{q'}_{Y_2|U_2})} P_{Y_2^n}^q(y_2^n) \;,
\label{eq:BE2bound0} 
\end{align}
where the first equality holds since $U_2^n(m_2)$ is independent of $Y_2^n$ for every $m_2\neq 1$. Let 
 $(u_2^n,y_2^n)\in \tset(P_{U_2} P^{q'}_{Y_2|U_2})$. Then, $\,y_2^n\in\Aset^{\delta_2}(P_{Y_2}^{q'})$ with $\delta_2\triangleq |\Uset_2|\cdot\delta$. By Lemmas 2.6 and 2.7 in
 \cite{CsiszarKorner:82b},
\bieee
\label{eq:BpYbound}
P_{Y_2^n}^q(y_2^n)=2^{-n\left(  H(\hP_{y_2^n})+D(\hP_{y_2^n}||P_{Y_2}^q)
\right)}\leq 2^{-n H(\hP_{y_2^n})}
\leq 2^{-n\left( H_{q'}(Y_2) -\eps_1(\delta) \right)} \;,
\eieee
where $\eps_1(\delta)\rightarrow 0$ as $\delta\rightarrow 0$. Therefore, by (\ref{eq:BE2poly})$-$(\ref{eq:BpYbound}),
\begin{align}
& \prob{\Eset_{2,2}}           																										\nonumber\\
\leq& (n+1)^{|\Sset|}																															\nonumber\\
& \cdot \sup_{q'\in\tQ} \left[
2^{nR_2} \cdot \sum_{u_2^n\in\Uset_2^n} P_{U_2^n}(u_2^n)\cdot |\{y_2^n\,:\; (u_2^n,y_2^n)\in\tset(P_{U_2} P_{Y_2|U_2}^{q'})\}| \cdot 
 2^{-n\left( H_{q'}(Y_2) -\eps_1(\delta) \right)}		 \right]	  												\nonumber\\
\leq& (n+1)^{|\Sset|}\cdot \sup_{q'\in\Qset} 
2^{-n[ I_{q'}(U_2;Y_2) 
-R_2-\eps_2(\delta) ]} \label{eq:BexpCR2} \;,
\end{align}
with $\eps_2(\delta)\rightarrow 0$ as $\delta\rightarrow 0$, 
where the last inequality is due to \cite[Lemma 2.13]{CsiszarKorner:82b}. The RHS of (\ref{eq:BexpCR2})
  tends to zero exponentially as $n\rightarrow\infty$, provided that $R_2<\inf_{q'\in\Qset} I_{q'}(U_2;Y_2)
-\eps_2(\delta)$.  

Now, consider the error event of decoder 1. For every $(m_1,m_2)\in [1:2^{nR_1}]\times [1:2^{nR_1}]$, define the events
\begin{align*}
&\Eset_{1,2}(m_2)=\{ (U_2^n(m_2),Y_1^n)\in\Aset^{\delta_3}(P_{U_2} P^{q'}_{Y_1|U_2})  \,,\;\text{for some $q'\in\tQ$}  \}\;,\\
&\Eset_{1,1}(m_1,m_2)=\{ (U_2^n(m_2),U_1^n(m_1,m_2),Y_1^n)\in\tset(P_{U_2,U_1} P^{q'}_{Y_1|U_2,U_1})  \,,\;\text{for some $q'\in\tQ$}\}\;,
\end{align*}
where $\delta_3\triangleq |\Uset_1|\delta$. 
Then, the error event is bounded by
\begin{align}
\left\{(\hM_1,\hM_2)\neq (1,1)\right\}
&\subseteq \,\Eset_{1,1}(1,1)^c \,\cup\; \bigcup_{m_1\neq 1}\Eset_{1,1}(m_1,1)
\cup\; \bigcup_{ \substack{m_1\in [1:2^{nR_1}]\\ m_2\neq 1}
}\Eset_{1,1}(m_1,m_2)																																\nonumber\\
&\subseteq\,\Eset_{1,1}(1,1)^c \,\cup\; \bigcup_{m_1\neq 1}\Eset_{1,1}(m_1,1)
\cup\; \bigcup_{  m_2\neq 1}
\Eset_{1,2}(m_2)
  \;,
\end{align}
where the last line follows from the fact that if the event $\Eset_{1,1}(m_1,m_2)$ occurs, then $\Eset_{1,2}(m_2)$ occurs as well.
Thus, by the union of events bound,
\begin{align}
\prob{(\hM_1,\hM_2)\neq (1,1)}
\leq&\prob{ \Eset_{1,1}(1,1)^c} + \sum_{m_2\neq 1} \prob{\Eset_{1,2}(m_2) }
+ \sum_{m_1\neq 1} \prob{\Eset_{1,1}(m_1,1) } \nonumber\\
\leq&
2^{-\theta n}+2^{-n\left(\inf\limits_{q'\in\Qset} I_{q'}(U_2;Y_1)-R_2-\eps_3(\delta) \right)}+\sum_{m_1\neq 1}\prob{\Eset_{1,1}(m_1,1) }
\;,
\label{eq:Blowerexp1}
\end{align}
where the last inequality follows from the law of large numbers and type class considerations used before, with $\eps_3(\delta)\rightarrow 0$ as $\delta\rightarrow 0$.  
Since the compound  DBC is assumed to be degraded, we have that $I_{q'}(U_2;Y_1)\geq I_{q'}(U_2;Y_2)$ for all $q'\in\pSpace(\Sset)$.
Thus,  taking $R_2<\inf_{q'\in\Qset} I_{q'}(U_2;Y_2)
-\eps_2(\delta)$ guarantees that the middle term in the RHS of (\ref{eq:Blowerexp1}) 
 tends to zero  exponentially as $n\rightarrow\infty$.
 It remains for us to bound the last sum. Using similar type class considerations, we have that for every $q'\in\tQ$ and $m_1\neq 1$,
\begin{align}
&\prob{ ( U_2^n(1), U_1^n(m_1,1),Y_1^n)\in \tset(P_{U_2,U_1} P^{q'}_{Y_1|U_2,U_1})}  \nonumber\\
&= \sum_{(u_2^n,u_1^n,y_1^n)\in \tset(P_{U_2,U_1} P^{q'}_{Y_1|U_2,U_1})} 
P_{U_2^n}(u_2^n)\cdot    P_{U_1^n|U_2^n}(u_1^n|u_2^n)\cdot  P^q_{Y_1^n|U_2^n}(y_1^n|u_2^n)\nonumber\\
&\leq 2^{n(H_{q'}(U_2,U_1,Y_1)+\eps_4(\delta))}\cdot 2^{-n(H(U_2)-\eps_4(\delta))}\cdot 2^{-n(H(U_1|U_2)-\eps_4(\delta))}\cdot 2^{-n(H_{q'}(Y_1|U_2)-\eps_4(\delta))} 																									\nonumber\\
&= 2^{-n(I_{q'}(U_1;Y_1|U_2)-4 \eps_4(\delta))} \;,
\label{eq:BI1}
\end{align}
where $\eps_4(\delta)\rightarrow 0$ as $\delta\rightarrow 0$. 
Therefore, the sum term in the RHS of (\ref{eq:Blowerexp1}) is bounded by 
\begin{align}
&\sum_{m_1\neq 1}\prob{\Eset_{1,1}(m_1,1) }\\
&=\sum_{m_1\neq 1} \prob{
(U_2^n(1),U_1^n(m_1,1),Y_1^n)\in\tset(P_{U_2,U_1} P^{q'}_{Y_1|U_2,U_1})  \,,\;\text{for some $q'\in\tQ$}\}
}\\
&\leq (n+1)^{|\Sset|} \cdot 2^{-n\left(\inf\limits_{q'\in\Qset} I_{q'}(U_1;Y_1|U_2)-R_1- \eps_5(\delta) \right)} \;,
\label{eq:BexpCR1}
\end{align}
where the last line follows from (\ref{eq:BI1}), and $\eps_5(\delta)\rightarrow 0$ as $\delta\rightarrow 0$. The last expression  tends to zero  exponentially as $n\rightarrow\infty$ and $\delta\rightarrow 0$ provided that $R_1<\inf_{q'\in\Qset} I_{q'}(U_1;Y_1|U_2)-\eps_5(\delta)$.  

The probability of error, averaged over the class of the codebooks, exponentially decays to zero  as $n\rightarrow\infty$. Therefore, there must exist a $(2^{nR_1},2^{nR_2},n,\eps)$ deterministic code, for a sufficiently large $n$.
\qed

\section{Proof of Theorem~\ref{theo:BcvC}}            
\label{app:BcvC}   
\section*{Part 1}
At the first part of the theorem it is assumed that the interior of the capacity region is non-empty, \ie $\interior{\BCcompound}\neq\emptyset$. Denote the marginal compound channels with causal SI, corresponding to user 1 and user 2, by
\bieee
\compound_1=\{ \Qset, \sbc \} \,,\;\text{and
}\quad 
\compound_2=\{ \Qset, \wbc \}\;,
\eieee
respectively. 
Since the compound  DBC is assumed to be degraded, this means that 
\bieee
\label{eq:BApos}
\opC(\compound_1)\geq \opC(\compound_2)>0 \;.
\eieee
 
 \begin{proof}[Achievability proof]
 We show that every rate pair $(R_1,R_2)\in\BICcompound$ can be achieved using a code based on Shannon strategies with the addition of a codeword \emph{suffix}. 
 At time $i=n+1$, having completed the transmission of the messages, the type of the state sequence $s^n$ is known to the encoder.
Following the assumption that the interior of the capacity region is non-empty,  the type of $s^n$  can be reliably communicated to both receivers 
 as a suffix, while the blocklength is increased by  $\nu>0$ additional channel uses, where $\nu$ is small compared to $n$. 
The receivers  first estimate the type of $s^n$, and then use joint typicality  with respect to the estimated type. 
  The details are provided below.  
	
By (\ref{eq:BApos}), we have that 
for every $\eps_1>0$ and sufficiently large blocklength $\nu$, there exists a 
$(2^{\nu \tR_1},$ $2^{\nu \tR_2},\nu,\eps_1)$ code $\tcode=(\tf^\nu,\tg_1,\tg_2)$
 for the transmission of a type  $\hP_{s^n}$ at  positive rates $\tR_1>0$ and $\tR_2>0$. 
 Since the total number of types is polynomial in $n$ (see \cite{CsiszarKorner:82b}),  the type $\hP_{s^n}$ can be transmitted at a negligible rate, with a blocklength that grows a lot slower than $n$, \ie   
\bieee
\nu=o(n) \;.
\eieee
We now construct a code $\code$ over the compound DBC with causal SI, such that the blocklength is $n+o(n)$, and the rate 
$R'_n$ approaches $R$ as $n\rightarrow\infty$. 

\emph{Codebook Generation:} Fix the distribution $P_{U_1,U_2}(u_1,u_2)=p(u_1,u_2)$ and the function $\encs(u_1,u_2,s)$. Generate $2^{nR_2}$ independent sequences $u_2^n(m)$, $m\in[1:2^{nR_2}]$, at random, each according to
 $\prod_{i=1}^n P_{U_2}(u_{2,i})$.
For every $m_2\in [1:2^{nR_2}]$, generate $2^{nR_1}$ sequences at random, 
\bieee
u_1^n(m_1,m_2) \sim \prod_{i=1}^n P_{U_1|U_2}(u_{1,i}|u_{2,i}(m_2)) \,,\;\text{for $m_1\in [1:2^{nR_1}]$}\;,
\eieee
conditionally independent given $u_2^n(m_2)$.  
 Reveal the codebook of the message pair $(m_1,m_2)$ and the codebook of the type  $\hP_{s^n}$  to the encoder and the decoders.

\emph{Encoding}: To send a message pair $(m_1,m_2)\in [1:2^{nR_1}]\times [1:2^{nR_2}]$, transmit at time $i\in[1:n]$,
\bieee
x_i=\encs\left( u_{1,i}(m_1,m_2),u_{2,i}(m_2),s_i\right) \;.
\eieee
At time $i\in [n+1:n+\nu]$, knowing the sequence of previous states $s^n$,  transmit 
\bieee
x_i=\tf_i(\hP_{s^n},\, s_{n+1},\ldots,s_{n+i} ) \;,
\eieee
where $\hP_{s^n}$ is the type of the sequence $(s_1,\ldots,s_n)$. 
That is, the encoded type $\hP_{s^n}$ is transmitted as a suffix of the codeword.
We note that the type of the sequence $(s_{n+1},\ldots,s_{n+i})$ is not necessarily $\hP_{s^n}$, and it 
 is irrelevant for that matter since by (\ref{eq:BApos}),  there exists a  $(2^{\nu \tR_1},2^{\nu \tR_2},\nu,\eps_1)$ code $\tcode=(\tf^\nu,\tg_1,\tg_2)$ for the transmission of $\hP_{s^n}$ over the compound  DBC with causal SI, with $\tR_1>0$ and $\tR_2>0$.

\emph{Decoding}: Decoder 2 receives the output sequence $y_2^{n+\nu}$. As a pre-decoding step, the receiver decodes the last $\nu$ output symbols, and finds an estimate of the type of the state sequence, 
\bieee
\hq_2=\tg_2(y_{2,n+1},\ldots,y_{2,n+\nu}) \;.
\eieee
 Then, given the output sequence $y_2^n$,  decoder 2 finds a unique $\tm_2\in[1:2^{nR_2}]$ such that
\bieee
 (u_2^n(\tm_2),y_2^n)\in\tset(P_{U_2} P^{\hq_2}_{Y_2|U_2}) \;.
\eieee  
If there is none, or more than one such $\tm_2\in[1:2^{nR_2}]$, then decoder 2 declares an error.

Similarly, decoder 1 receives  $y_1^{n+\nu}$ and begins with decoding the type of the state sequence,
\bieee
\hq_1=\tg_1(y_{1,n+1},\ldots,y_{1,n+\nu}) \;.
\eieee
Then, decoder 1 finds a unique pair of messages $(\hm_1,\hm_2)\in [1:2^{nR_1}]\times [1:2^{nR_2}]$ such that 
\bieee
(u_2^n(\hm_2),u_1^n(\hm_1,\hm_2),y_1^n)\in \tset(P_{U_2,U_1} P^{\hq_1}_{Y_1|U_2,U_1}) \;.
\eieee
If there is none, or more than one such pair  $(\hm_1,\hm_2)\in [1:2^{nR_1}]\times [1:2^{nR_2}]$, then decoder 1 declares an error.

\emph{Analysis of Probability of Error}:
By symmetry, 
we may assume without loss of generality that the users sent $(M_1,M_2)=(1,1)$. 
Let $q(s)\in\Qset$ denote the actual state distribution chosen by the jammer, and let $\qn(s^n)=\prod_{i=1}^n q(s_i)$. 
Then, by the union of events bound, the probability of error is bounded by 
\bieee
\err(q,\code)\leq \prob{\tM_2 \neq 1} + \prob{(\hM_1,\hM_2)\neq (1,1)} \;,
\eieee
where the conditioning on $(M_1,M_2)=(1,1)$ is omitted for convenience of notation. 

Define the  events  
\begin{align}
&\Eset_{1,0} =\{ \hq_1\neq \hP_{S^n} \} 						 																		\label{eq:BE11}\\
&\Eset_{1,1}(m_1,m_2,q')=\{ (U_2^n(m_2),U_1^n(m_1,m_2),Y_1^n)\in \tset(P_{U_2,U_1} P^{q'}_{Y_1|U_2,U_1}) \}   \label{eq:BE12}\\
&\Eset_{1,2}(m_2,q')=\{ (U_2^n(m_2),Y_1^n)\in \Aset^{\delta_1}(P_{U_2} P^{q'}_{Y_1|U_2}) \} 		\;,					\label{eq:BE13}\\
\intertext{ 
and
 }
&\Eset_{2,0} =\{ \hq_2\neq \hP_{S^n} \} 																								\label{eq:BE21}\\
&\Eset_{2,1}(m_2,q') =\{ (U_2^n(m_2),Y_2^n)\in \tset(P_{U_2} P^{q'}_{Y_2|U_2}) \;,						\label{eq:BE22} 
\end{align}
for every $m_1\in [1:2^{nR_1}]$, $m_2\in [1:2^{nR_2}]$, 
 and $q'\in\pSpace(\Sset)$, where $\delta_1=|\Uset_1|\delta$. The error event of decoder 2 is bounded by
\bieee
\left\{ \tM_2\neq 1 \right\}&\subseteq&\Eset_{2,0}\cup\Eset_{2,1}(1,\hq_2\,)^c\cup
\bigcup_{m_2\neq 1} \Eset_{2,1}(m_2,\hq_2\,) \IEEEnonumber\\
&=& 
\Eset_{2,0} \,\cup\,
\left(  \Eset_{2,0}^c\cap\Eset_{2,1}(1,\hq_2\,)^c  \right) \,\cup\,
\left( \bigcup_{m_2\neq 1} \Eset_{2,0}^c\cap\Eset_{2,1}(m_2,\hq_2\,) \right) 
 \;. \IEEEnonumber
\eieee
 By the union of events bound,
\begin{align}
&\prob{\tM_2\neq 1}																																	\nonumber\\	
&\leq 
 \prob{\Eset_{2,0}} +
\prob{\Eset_{2,0}^c\cap\Eset_{2,1}(1,\hq_2\,)^c }                             
+  \prob{\bigcup_{m_2\neq 1} \Eset_{2,0}^c\cap  \Eset_{2,1}(m_2,\hq_2\,)} \;.
\label{eq:BcvcErr1}
\end{align}

 Since the code $\tcode$ for the transmission of the type is a $(2^{\nu \tR_1},2^{\nu \tR_2},\nu,\eps_1)$ code,
where $\eps_1>0$ is arbitrarily small, we have that the probability of erroneous decoding of the type is bounded by
 \bieee
\label{eq:decq}
\prob{\Eset_{1,0}\cup \Eset_{2,0}}\leq\eps_1 \;.
\eieee
 Thus, 
the first term in the RHS of (\ref{eq:BcvcErr1}) is bounded by $\eps_1$. Then, we maniplute the last two term as follows.  
\begin{align}
\prob{\tM_2\neq 1}
\leq& \;
\sum_{s^n\in\Aset^{\delta_2}(q)}\qn(s^n)
\cprob{\Eset_{2,0}^c\cap\Eset_{2,1}(1,\hq_2\,)^c  }{S^n=s^n}
																																			\nonumber\\
&+
\sum_{s^n\notin\Aset^{\delta_2}(q)} \qn(s^n)
\cprob{\Eset_{2,0}^c\cap\Eset_{2,1}(1,\hq_2\,)^c }{S^n=s^n}
																																			\nonumber\\
&+
\sum_{s^n\in\Aset^{\delta_2}(q)} \qn(s^n)
\cprob{\bigcup_{m_2\neq 1} \Eset_{2,0}^c\cap  \Eset_{2,1}(m_2,\hq_2\,)}{S^n=s^n}
																																			\nonumber\\
&+
\sum_{s^n\notin\Aset^{\delta_2}(q)} \qn(s^n)
\cprob{\bigcup_{m_2\neq 1} \Eset_{2,0}^c\cap  \Eset_{2,1}(m_2,\hq_2\,)}{S^n=s^n}  +\eps_1			\;,	
\label{eq:BEcompound}																														
\end{align}
where 
\bieee
\delta_2 \triangleq\frac{1}{2|\Sset|
}\cdot\delta 	\;.		\label{eq:Bdelta1}
\eieee
 Next we show that the first and the third sums in (\ref{eq:BEcompound}) 
 tend to zero as $n\rightarrow\infty$. 

Consider a given $s^n\in\Aset^{\delta_2}(q)$. For notational convenience, denote
 \bieee
\label{eq:BempC}
q''=\hP_{s^n} \;.
\eieee 
 Then, by the definition of the $\delta$-typical set, we have that
\bieee
&&|q''(s)-q(s)|\leq\delta_2 \;\text{for all $s\in\Sset\,,\;$ and }\;\,
 q''(s)=0 \;\text{when $q(s)=0$}\;. 																																						\IEEEnonumber
\eieee
It follows that
\bieee
|P_{U_2}(u_2)P_{Y_2|U_2}^{q''}(y|u_2)-P_{U_2}(u_2)P_{Y_2|U_2}^q(y_2|u_2)|&\leq&  \delta_2\cdot\sum_{s,u_1} P_{U_1|U_2}
(u_1|u_2) \wbc(y_2|\encs(u_1,u_2,s),s)\IEEEnonumber\\
&\leq& \delta_2\cdot\sum_{s,u_1} P_{U_1|U_2}(u_1|u_2)
= \delta_2\cdot|\Sset|=\frac{\delta}{2} \;,
\label{eq:BpUYclose}
\eieee
for all $u_2\in\Uset_2$ and $y_2\in\Yset_2$, where the last equality follows from (\ref{eq:Bdelta1}).

Consider the first sum in the RHS of (\ref{eq:BEcompound}). 
Given a state sequence  $s^n\in\Aset^{\delta_2}(q)$,  we have that
\begin{align}
&\cprob{\Eset_{2,0}^c\cap\Eset_{2,1}(1,\hq_2\,) }{S^n=s^n}    \nonumber\\ =& 
\cprob{\Eset_{2,0}^c\cap\Eset_{2,1}(1,\hP_{s^n}\,) }{S^n=s^n} \nonumber\\ =&
\cprob{\Eset_{2,0}^c\cap\Eset_{2,1}(1,q''\,) }{S^n=s^n}       \nonumber\\ =&
\cprob{\Eset_{2,0}^c }{\Eset_{2,1}(1,q''), S^n=s^n}\cdot\cprob{\Eset_{2,1}(1,q'')\,) }{S^n=s^n} 
 \;, \label{eq:BEsum1}
\end{align}
where the first equality follows from (\ref{eq:BE21}),  
 and the second equality follows from (\ref{eq:BempC}). 
Then,
\bieee
\cprob{\Eset_{2,0}^c\cap\Eset_{2,1}(1,\hq_2\,) }{S^n=s^n}
&\leq& \cprob{\Eset_{2,1}(1,q'')\, }{S^n=s^n} 																\IEEEnonumber\\
&=& \prob{\, (U_2^n(1),Y_2^n)\notin\Aset^{\delta}(P_{U_2} P^{q''}_{Y_2|U_2})  \,\big|\; S^n=s^n  }
\;. \qquad\quad
\label{eq:BE2bound1} 
\eieee
 Now, suppose that $(U_2^n(1),Y_2^n)\in\Aset^{\nicefrac{\delta}{2}}(P_{U_2} P^q_{Y_2|U_2})$, where $q$ is the actual state distribution. By (\ref{eq:BpUYclose}), in this case we have that $(U_2^n(1),Y_2^n)\in\tset(P_{U_2} P^{q''}_{Y_2|U_2})$. 
 Hence, by (\ref{eq:BE2bound1}), we have that
\begin{align}
\label{eq:BE2causal}
&\cprob{\Eset_{2,0}^c\cap\Eset_{2,1}(1,\hq_2\,)^c }{S^n=s^n} \nonumber\\
\leq&
\prob{\, (U_2^n(1),Y_2^n)\notin\Aset^{\nicefrac{\delta}{2}}(P_{U_2} P^q_{Y_2|U_2})  \,\big|\; S^n=s^n  }
\;.
\end{align}
The first sum in the RHS of (\ref{eq:BEcompound}) 
 is then bounded as follows. 
\begin{align}
&\sum_{s^n\in\Aset^{\delta_2}(q)} \qn(s^n)\cprob{\Eset_{2,0}^c\cap\Eset_{2,1}(1,\hq_2\,)^c }{S^n=s^n} 											\IEEEnonumber\\
&\leq
\sum_{s^n\in\Aset^{\delta_2}(q)} \qn(s^n) \prob{\, (U_2^n(1),Y_2^n)\notin\Aset^{\nicefrac{\delta}{2}}(P_{U_2} P^q_{Y_2|U_2})  \,\big|\; S^n=s^n } 				\IEEEnonumber\\
&\leq
\sum_{s^n\in\Sset^n} \qn(s^n) \prob{\, (U_2^n(1),Y_2^n)\notin\Aset^{\nicefrac{\delta}{2}}(P_{U_2} P^q_{Y_2|U_2})  \,\big|\; S^n=s^n } 										\IEEEnonumber\\
&=
\prob{\, (U_2^n(1),Y_2^n)\notin\Aset^{\nicefrac{\delta}{2}}(P_{U_2} P^q_{Y_2|U_2})  }\leq \eps_2 \;,
\label{eq:2sum1}
\end{align}
for a sufficiently large $n$, where the last inequality follows from  the law of large numbers.

We bound the third sum in the RHS of (\ref{eq:BEcompound}) 
 using similar arguments. If $(U_2^n(m_2),Y_2^n)\in\tset(P_{U_2} P_{Y_2|U_2}^{q''})$, then
 $(U_2^n(m_2),Y_2^n)\in\Aset^{\nicefrac{3\delta}{2}}(P_{U_2} P_{Y_2|U_2}^q)$, due to (\ref{eq:BpUYclose}). Thus, for every $s^n\in\Aset^{\delta_2}(q)$,
\begin{align}
&\cprob{\bigcup_{m_2\neq 1} \Eset_{2,0}^c\cap\Eset_{2,1}(m_2,\hq_2\,) }{S^n=s^n}      																				\nonumber\\
&\leq  \sum_{m_2\neq 1} \cprob{\Eset_{2,1}(m_2,q'')}{S^n=s^n}  																													
																						 \nonumber\\
&= \sum_{m_2\neq 1} \prob{\, (U_2^n(m_2),Y_2^n)\in\tset(P_{U_2} P^{q''}_{Y_2|U_2}) 
\,\big|\; S^n=s^n }
																																																				\nonumber\\
&\leq \sum_{m_2\neq 1} \prob{\, (U_2^n(m_2),Y_2^n)\in\Aset^{\nicefrac{3\delta}{2}}(P_{U_2} P^q_{Y_2|U_2}) 
  \,\big|\; S^n=s^n } \;.
\end{align}
This, in turn, implies that the third sum in the RHS of (\ref{eq:BEcompound}) 
is bounded by
\begin{align}
&\sum_{s^n\in\Aset^{\delta_2}(q)} \qn(s^n)\cprob{\bigcup_{m_2\neq 1} \Eset_{2,0}^c\cap\Eset_{2,1}(m_2,\hq\,) }{S^n=s^n} 		\nonumber\\
&\leq
\sum_{s^n\in\Sset^n}\sum_{m_2\neq 1} \qn(s^n)\cdot \prob{\, (U_2^n(m_2),Y_2^n)\in\Aset^{\nicefrac{3\delta}{2}}(P_{U_2} P^q_{Y_2|U_2}) 
 \,\big|\; S^n=s^n }\nonumber\\
&=
\sum_{m_2\neq 1} \prob{\, (U_2^n(m_2),Y_2^n)\in\Aset^{\nicefrac{3\delta}{2}}(P_{U_2} P^q_{Y_2|U_2}) \, }\nonumber\\
&\leq   2^{-n[ I_{q}(U_2;Y_2)-R_2-\eps_2(\delta) ]} \label{eq:2sum2} \;, 
\end{align} 
with $\eps_2(\delta)\rightarrow 0$ as $\delta\rightarrow 0$. The last inequality follows from standard type class considerations. The RHS of  (\ref{eq:2sum2})  tends to zero as $n\rightarrow\infty$, provided that 
\bieee
R_2< I_{q}(U_2;Y_2)-\eps_2(\delta) \;,
\eieee
 for some $p(u_1,u_2)$ and $\encs(u_1,u_2,s)$. 
Then, it follows from the law of large numbers that the second and fourth sums in the RHS of  (\ref{eq:BEcompound})
tend to zero as $n\rightarrow\infty$. Thus, by 
 (\ref{eq:2sum1}) and (\ref{eq:2sum2}), we have that 
the probability of error of decoder 2, $\prob{\tM_2\neq 1}$, tends to zero as $n\rightarrow\infty$.

Now, consider the error event of decoder 1,
\bieee
&&\left\{(\hM_1,\hM_2)\neq (1,1)\right\} 										\nonumber\\
&\subseteq&\, \Eset_{1,0} \,\cup\;  \,\Eset_{1,1}(1,1,\hq_1)^c \,\cup\; 
\bigcup_{(m_1,m_2)\neq (1,1)}  \Eset_{1,1}(m_1,m_2,\hq_1)   \nonumber\\
&=&\Eset_{1,0} \,\cup\;  \,\Eset_{1,1}(1,1,\hq_1)^c \,\cup\;
 \bigcup_{m_1\neq 1}\Eset_{1,1}(m_1,1,\hq_1)
\cup\; \bigcup_{ \substack{m_1\in [1:2^{nR_1}]\\ m_2\neq 1}
}\Eset_{1,1}(m_1,m_2,\hq_1)																	 \nonumber\\
&\subseteq\,& \Eset_{1,0} \,\cup\;  \Eset_{1,1}(1,1,\hq_1)^c \,\cup\; 
\bigcup_{m_1\neq 1}\Eset_{1,1}(m_1,1,\hq_1) \cup\;
 \bigcup_{  m_2\neq 1} \Eset_{1,2}(m_2,\hq_1) 							 \nonumber\\
&=& \Eset_{1,0} 																																									
\,\cup\; \left( \Eset_{1,0}^c \cap \Eset_{1,1}(1,1,\hq_1)^c \right)														
\,\cup\;\bigcup_{m_1\neq 1} \left( \Eset_{1,0}^c \cap \Eset_{1,1}(m_1,1,\hq_1) \right) 				
\IEEEnonumber\\&&
\,\cup\; \bigcup_{  m_2\neq 1} \left( \Eset_{1,0}^c \cap \Eset_{1,2}(m_2,\hq_1) \right)
  \;,
\eieee 
where the second inclusion follows from the fact that if the event $\Eset_{1,1}(m_1,m_2,\hq_1)$ occurs, then $\Eset_{1,2}(m_2,\hq_1)$ occurs as well.
Thus, by the union of events bound,
\begin{align}
&\prob{(\hM_1,\hM_2)\neq (1,1)} 																									\nonumber\\
&\leq \prob{\Eset_{1,0}}  +
\prob{ \Eset_{1,0}^c \cap \Eset_{1,1}(1,1,\hq_1)^c} +  \prob{\bigcup_{m_2\neq 1} \Eset_{1,0}^c \cap \Eset_{1,2}(m_2,\hq_1) }																																		\nonumber\\
&+  \prob{\bigcup_{m_1\neq 1}  \Eset_{1,0}^c \cap \Eset_{1,1}(m_1,1,\hq_1) } 
\;.
\label{eq:Blowerexp1c}
\end{align}
By (\ref{eq:decq}), the first term is bounded by $\eps_1$, and 
as done above, we write
\begin{align}
&\prob{(\hM_1,\hM_2)\neq (1,1)} \nonumber\\
\leq& 
\sum_{s^n\in\Aset^{\delta_2}(q)}\qn(s^n)
\cprob{ \Eset_{1,0}^c \cap \Eset_{1,1}(1,1,\hP_{s^n})^c}{S^n=s^n}
																																			\nonumber\\
&+
\sum_{s^n\in\Aset^{\delta_2}(q)} \qn(s^n)
\cprob{\bigcup_{m_2\neq 1}  \Eset_{1,0}^c \cap \Eset_{1,2}(m_2,\hP_{s^n})}{S^n=s^n}
\nonumber\\
&+
\sum_{s^n\in\Aset^{\delta_2}(q)} \qn(s^n)
\cprob{\bigcup_{m_1\neq 1}  \Eset_{1,0}^c \cap \Eset_{1,1}(m_1,1,\hP_{s^n})}{S^n=s^n}
\nonumber\\
& +3\cdot\prob{S^n\notin\Aset^{\delta_2}(q)}+\eps_1 \;,
\label{eq:BEcausal1}
\end{align}
where $\delta_2$ is given by (\ref{eq:Bdelta1}). 
 By the law of large numbers, the probability $\prob{S^n\notin\Aset^{\delta_2}(q)}$ tends to zero as $n\rightarrow\infty$.
As for the sums, we use similar arguments to those used above.

 We have that for a given $s^n\in\Aset^{\delta_2}(q)$,
\begin{align}
&|P_{U_1,U_2}(u_1,u_2) P^{q''}_{Y_1|U_1,U_2}(y_1|u_1,u_2)-P_{U_1,U_2}(u_1,u_2)P^{q}_{Y_1|U_1,U_2}(y_1|u_1,u_2)| \nonumber\\
&\leq 
\delta_2 \cdot \sum_{s\in\Sset} \sbc(y_1|\encs(u_1,u_2,s) 
\leq |\Sset|\cdot \delta_2=\frac{\delta}{2} \;,
\label{eq:Btset1}
\end{align}
with $q''=\hP_{s^n}$, where the last equality follows from (\ref{eq:Bdelta1}).

The first sum in the RHS of (\ref{eq:BEcausal1}) is bounded by 
\begin{align}
&\sum_{s^n\in\Aset^{\delta_2}(q)}\qn(s^n)
\cprob{ \Eset_{1,0}^c \cap \Eset_{1,1}(1,1,\hP_{s^n})^c}{S^n=s^n} \nonumber\\
&\leq \sum_{s^n\in\Sset^n}\qn(s^n)
\cprob{ (U_2^n(1),U_1^n(1,1),Y_1^n )\notin \Aset^{\nicefrac{\delta}{2}}(P_{U_1,U_2} P^{q}_{Y_1|U_1,U_2})  }{S^n=s^n}\nonumber\\
&=
\prob{ (U_2^n(1),U_1^n(1,1),Y_1^n )\notin \Aset^{\nicefrac{\delta}{2}}(P_{U_1,U_2} P^q_{Y_1|U_1,U_2})}\leq \eps_2\;.
\end{align}
 The last inequality follows from the law of large numbers, with a sufficiently large $n$.

The second sum in the RHS of (\ref{eq:BEcausal1}) is bounded by 
\begin{align}
\sum_{s^n\in\Aset^{\delta_2}(q)} \qn(s^n)
\cprob{\bigcup_{m_2\neq 1}  \Eset_{1,0}^c \cap \Eset_{1,2}(m_2,\hP_{s^n})}{S^n=s^n} 
\leq 2^{-n( I_q(U_2;Y_1)-R_2-\eps_3(\delta)} \;.
\end{align}
with $\eps_3(\delta)\rightarrow 0$ as $n\rightarrow \infty$ and $\delta\rightarrow 0$. This is obtained following exactly the same analysis as for decoder 2. Then, the second sum tends to zero provided that  
\bieee
R_2 \leq I_q(U_2;Y_1)-\eps_3(\delta) \;.
\eieee
Since the compound  DBC is assumed to be degraded, the requirement $R_2<I_q(U_2;Y_2)$ suffices.

The third sum in the RHS of (\ref{eq:BEcausal1}) is bounded by 
\begin{align}
&\sum_{s^n\in\Aset^{\delta_2}(q)} \qn(s^n)
\cprob{\bigcup_{m_1\neq 1}  \Eset_{1,0}^c \cap \Eset_{1,1}(m_1,1,\hP_{s^n})}{S^n=s^n}\\
&\leq 
\sum_{s^n\in\Aset^{\delta_2}(q)} \sum_{m_1\neq 1}  \qn(s^n)
\cprob{  \Eset_{1,1}(m_1,1,\hP_{s^n})}{S^n=s^n} \;.
\end{align}
For every $s^n\in\Aset^{\delta_2}(q)$, it follows from (\ref{eq:Btset1}) that the event $\Eset_{1,1}(m_1,1,\hP_{s^n})$
implies that 
\bieee
(U_2^n(1),U_1^n(m_1,1),Y_1^n)\in\Aset^{\nicefrac{3\delta}{2}}(P^q_{U_2,U_1,Y_1}) \;.
\eieee
  Thus,  the sum is bounded by 
\begin{align}
&\sum_{s^n\in\Aset^{\delta_2}(q)} \qn(s^n)
\cprob{\bigcup_{m_1\neq 1}  \Eset_{1,0}^c \cap \Eset_{1,1}(m_1,1,\hP_{s^n})}{S^n=s^n}\\
&\leq 
2^{-n(I_q(U_1;Y_1|U_2)-R_1-\delta_3)}
\end{align}
where $\delta_3\rightarrow 0$ as $\delta\rightarrow 0$.

We conclude that the RHS of both (\ref{eq:BEcompound}) and (\ref{eq:BEcausal1}) tend to zero as $n\rightarrow \infty$. Thus, the overall probability of error, averaged over the class of the codebooks, 
 decays to zero  as $n\rightarrow\infty$. Therefore, there must exist a $(2^{nR_1},2^{nR_2},n,\eps)$ deterministic code, for a sufficiently large $n$.
\end{proof}

\begin{proof}[Converse proof]
Assume to the contrary that there exists an achievable rate pair
$(R_1,R_2)\notin 
\bigcap_{q(s)\in\Qset} \BICrp$  using random codes over the compound DBC $\Bcompound$ with causal SI. Hence, for some state distribution $q^*(s)$ in the closure of $\Qset$, we have that $(R_1,R_2)\notin\BICrpS$.

The achievability assumption implies that for every $\eps>0$ and sufficiently large $n$,
there exists a $(2^{nR_1},2^{nR_2},n)$ random code $\code^\Gamma$ for the compound  DBC $\Bcompound$ with causal SI, with
  $\err(q,\code^\Gamma)\leq \eps$ for all i.i.d. state distributions $q(s)\in\Qset$, and in particular, for $q^*(s)$, since $\err(q,\code^\Gamma)$ is continuous in $q$.
	
	Consider the DBC $\BrpS$ with causal SI where the state sequence is i.i.d. according $q^*(s)$.
	If such a random code $\code^\Gamma$ would exist, then it could have been used over the DBC $\BrpS$,  achieving a rate pair $(R_1,R_2)\notin\BICrpS$. This is a contradiction, since
	the random code capacity region of $\BrpS$ is given by $\BICrpS$ \cite[Theorem 4]{Steinberg:05p}. 
	We deduce that the assumption is false, and $(R_1,R_2)\notin 
\bigcap_{q(s)\in\Qset} \BICrp$ cannot be achieved.
\end{proof}

\section*{Part 2}
We show that when  the set of state distributions $\Qset$ is convex,  and the condition $\sCondQ$ holds, 
the capacity region of the compound DBC $\Bcompound$ with causal SI is given by $\BCcompound=\BrCcompound=\BIRcompound=\BICcompound$ 
(and this holds regardless of whether the interior of the capacity region is empty or not). 

Due to part 1, we have that
\bieee
\label{eq:Bcompound2up}
 \BrCcompound \subseteq \BICcompound \;.
\eieee 
By Lemma~\ref{lemm:BcompoundLowerB},
\bieee
\BCcompound \supseteq \BIRcompound \;.
\eieee
Thus, 
\bieee
\label{eq:BcompoundTightineq}
 \BIRcompound\subseteq  \BCcompound \subseteq \BrCcompound \subseteq \BICcompound \;. 
\eieee

To conclude the proof, we show that the condition $\sCondQ$ implies that $\BIRcompound\supseteq\BICcompound$, hence the inner and outer bounds coincide. By Definition~\ref{def:Bcompoundachieve}, if a function $\encs(u_1,u_2,s)$ and a set 
$\Dset$ achieve $\BIRcompound$ and $\BICcompound$, then
\begin{subequations}
\label{eq:BcompoundachieveEq} 
\begin{align}  
\label{eq:BIRcompoundachieveEq} 
\BIRcompound =\bigcup_{p(u_1,u_2)\in\Dset}\,  
\left\{
\begin{array}{lll}
(R_1,R_2) \,:\; & R_2 &\leq  \min_{q\in\Qset} I_q(U_2;Y_2) \;, \\
								& R_1 &\leq  \min_{q\in\Qset} I_q(U_1;Y_1|U_2)  
\end{array}
\right\} \;,
\intertext{and}
\label{eq:BICcompoundachieveEq} 
\BICcompound = \bigcap_{q(s)\in\Qset}\, \bigcup_{p(u_1,u_2)\in\Dset} 
\left\{
\begin{array}{lll}
(R_1,R_2) \,:\; & R_2 &\leq   I_q(U_2;Y_2) \;, \\
								& R_1 &\leq   I_q(U_1;Y_1|U_2)  
\end{array}
\right\} \;.
\end{align}
\end{subequations}
Hence, when the condition $\sCondQ$ holds, we have  by Definition~\ref{def:sCondQ} 
 that for some  $\encs(u_1,u_2,s)$, $\Dset\subseteq\pSpace(\Uset_1\times\Uset_2)$, and $q^*\in\Qset$,
\bieee
\BIRcompound&=&  \bigcup_{p(u_1,u_2)\in\Dset}
\left\{
\begin{array}{lll}
(R_1,R_2) \,:\; & R_2 &\leq   I_{q^*}(U_2;Y_2) \;, \\
								& R_1 &\leq   I_{q^*}(U_1;Y_1|U_2)  
\end{array}
\right\} \\
&\supseteq& \BICcompound \;,
\eieee
where the last line follows from (\ref{eq:BICcompoundachieveEq}).
%
%
\qed

\section{Proof of Theorem~\ref{theo:Bmain}}
\label{app:Bmain}

\subsection*{Part 1}
First, we explain the general idea. 
As in Chapter~\ref{chap:avcC}, we devise  a causal version of Ahlswede's Robustification Technique 
(RT)  \cite{Ahlswede:86p,WinshtokSteinberg:06c}. Namely, we use codes for the compound  DBC to construct a random code for the AVDBC using randomized permutations. However, in our case, the causal nature of the problem imposes a difficulty, and the application of the RT is not straightforward.

In \cite{Ahlswede:86p,WinshtokSteinberg:06c}, the state information is non-causal and a random code is defined via permutations of the codeword symbols. This cannot be done here, because the SI is provided to the encoder in a causal manner. 
We resolve this difficulty using Shannon strategy codes for the compound  DBC to construct a random code for the AVDBC, applying permutations to the \emph{strategy sequence} $(u_1^n,u_2^n)$, which is an integral part of the Shannon strategy code, and is independent of the channel state. The details are given below.

\subsubsection{Inner Bound}
We show that the region defined in (\ref{eq:BIRcompoundP}) can be achieved by random codes over the 
AVDBC $\avbc$ with causal SI, \ie $\BCavc \supseteq \BIRavc$.
%
%
The proof relies on similar ideas to those in the proof of Theorem~\ref{theo:ALavcCr} in Appendix~\ref{theo:ALavcCr}. 
We start with Ahlswede's RT, stated below. Let $h:\Sset^n\rightarrow [0,1]$ be a given function. If, for some fixed $\alpha_n\in(0,1)$, and for all 
$ \qn(s^n)=\prod_{i=1}^n q(s_i)$, with 
$q\in\pSpace(\Sset)$, 
\bieee
\label{eq:BRTcondC}
\sum_{s^n\in\Sset^n} \qn(s^n)h(s^n)\leq \alpha_n \;,
\eieee
then,
\bieee
\frac{1}{n!} \sum_{\pi\in\Pi_n} h(\pi s^n)\leq \beta_n \;,\quad\text{for all $s^n\in\Sset^n$} \;,
\eieee
where $\Pi_n$ is the set of all $n$-tuple permutations $\pi:\Sset^n\rightarrow\Sset^n$, and 
$\beta_n=(n+1)^{|\Sset|}\cdot\alpha_n$. 

According to Lemma~\ref{lemm:BcompoundLowerB}, 
 for every $(R_1,R_2)\in\BIRavc$, there exists a  $(2^{nR_1},$ $2^{nR_2},$ $n,$ $e^{-2\theta n})$ Shannon strategy code for the compound DBC $\BcompoundP$ with causal SI, for some $\theta>0$ and sufficiently large $n$. 
Given such a Shannon strategy code $\code=$ $(u_1^n(m_1,m_2),$ $u_2^n(m_2),$ $\encs(u_1,u_2,s),$ $\dec_1(y_1^n),$ $\dec_2(y_2^n))$,
 we have that (\ref{eq:BRTcondC}) is satisfied with  $h(s^n)=\cerr(\code)$  and $\alpha_n=e^{-2\theta n}$.  
As a result, Ahlswede's RT tells us that
\bieee
\label{eq:BdetErrC}
\frac{1}{n!} \sum_{\pi\in\Pi_n} P_{e|\pi s^n}^{(n)}(\code)\leq (n+1)^{|\Sset|}e^{-2\theta n} 
\leq e^{-\theta n}  \;,\quad\text{for all $s^n\in\Sset^n$} \;,
\eieee 
for a sufficiently large $n$, such that $(n+1)^{|\Sset|}\leq e^{\theta n}$.  

On the other hand, for every $\pi\in\Pi_n$,  
\begin{align}
&P_{e|\pi s^n}^{(n)}(\code)																											\nonumber\\
 &\stackrel{(a)}{=}
\frac{1}{2^{ n(R_1+R_2) }}\sum_{m_1,m_2}
\sum_{(\pi y_1^n,\pi y_2^n)\notin\Dset(m_1,m_2)} 
 \nBC(\pi y_1^n,\pi y_2^n|\encs^n(u_1^n(m_1,m_2),u_2^n(m_2),\pi s^n),\pi s^n) 	\nonumber\\
 &\stackrel{(b)}{=}
\frac{1}{2^{ n(R_1+R_2) }}\sum_{m_1,m_2}
\sum_{(\pi y_1^n,\pi y_2^n)\notin\Dset(m_1,m_2)} 
 \nBC( y_1^n, y_2^n|\pi^{-1} \encs^n(u_1^n(m_1,m_2),u_2^n(m_2),\pi s^n), s^n) \;,\nonumber\\
 &\stackrel{(c)}{=}\frac{1}{2^{ n(R_1+R_2) }}\sum_{m_1,m_2}
\sum_{
(\pi y_1^n,\pi y_2^n)\notin\Dset(m_1,m_2)}
 \nBC( y_1^n, y_2^n| \encs^n(\pi^{-1} u_1^n(m_1,m_2),\pi^{-1} u_2^n(m_2), s^n), s^n)
\label{eq:Bcerrpi}
\end{align}
where $(a)$ is obtained by  plugging $\pi s^n$ and $x^n=\encs^n(\cdot,\cdot,\cdot)$ in (\ref{eq:Bcerr}) 
 and then changing the order of summation over $(y_1^n,y_2^n)$; $(b)$ holds because the broadcast channel is memoryless; and $(c)$ follows from that fact that  for a Shannon strategy code,  $x_i=\encs(u_{1,i},u_{2,i},s_i)$, $i\in[1:n]$, by Definition~\ref{def:BStratCode}. 
The last expression suggests the use of permutations applied to the encoding \emph{strategy sequence} and the channel output sequences.

Then, consider the $(2^{nR_1},2^{nR_2},n)$ random code $\code^\Pi$, specified by 
\begin{subequations}
\label{eq:BCpi}
\begin{align}
f_\pi^n(m_1,m_2,s^n)&= \encs^n(\pi^{-1} u_1^n(m_1,m_2),\pi^{-1} u_2^n(m_2),s^n) \;,\\
\intertext{and}
g_{1,\pi}(y_1^n)&=\dec_1(\pi y_1^n)
\;,\quad g_{2,\pi}(y_2^n)=\dec(\pi y_2^n) \;,
\end{align}
\end{subequations}
for $\pi\in\Pi_n$,
with a uniform distribution $\mu(\pi)=\frac{1}{|\Pi_n|}=\frac{1}{n!}$. 
Such permutations can be implemented without knowing $s^n$, hence this coding scheme does not violate the causality requirement. 

 From (\ref{eq:Bcerrpi}), 
 we see that 
\bieee 
\cerr(\code^\Pi)=\sum_{\pi\in\Pi_n} \mu(\pi) P_{e|\pi s^n}^{(n)}(\code) \;,
\eieee
for all $s^n\in\Sset^n$, and therefore, together with (\ref{eq:BdetErrC}), we have that the probability of error of the random code $\code^\Pi$ is bounded by 
\bieee 
\err(\qn,\code^{\Pi})\leq e^{-\theta n} \;,
\eieee 
for every $\qn(s^n)\in\pSpace^n(\Sset^n)$. That is, $\code^\Pi$ is a $(2^{nR_1},2^{nR_2},n,e^{-\theta n})$ random 
 code for the AVDBC $\avbc$ with causal SI at the encoder. 
This completes the proof of the inner bound. 
\qed 

\subsubsection{Outer Bound}
We show that the capacity region of the AVDBC $\avbc$ with causal SI is included within the region defined in
(\ref{eq:BrICav}), \ie $\BrCav\subseteq \BrICav$. 

The random code capacity region of the AVDBC is included within the random code capacity region of the compound  DBC, namely
\bieee
\label{eq:Outercomp2}
\BrCav \subseteq \BrCcompoundP \;.
\eieee
By Theorem~\ref{theo:BcvC} we have that $\BCcompound\subseteq\BICcompound$. Thus,   
with $\Qset=\pSpace(\Sset)$,
\bieee
\label{eq:Outercomp11}
\BrCcompoundP\subseteq\BrICav \;.
\eieee
It follows from (\ref{eq:Outercomp2}) and (\ref{eq:Outercomp11}) that $\BrCav\subseteq\BrICav$. Since the random code capacity region always includes  the deterministic code capacity region, we have that
$
\BCavc \subseteq \BrICav 
$ as well.
\qed 

\subsection*{Part 2}
The second equality, $\BIRavc=\BrICav$, follows from  part 2 of Theorem~\ref{theo:BcvC}, taking $\Qset=\pSpace(\Sset)$. 
By part 1, $\BIRavc\subseteq \BrCav\subseteq\BrICav$, hence the proof follows. \qed

\section{Proof of Lemma~\ref{lemm:BcorrSizeC}}
\label{app:BET}

The proof follows the lines of \cite[Section~4]{Ahlswede:78p}. 
 Let $\dK>0$ be an integer, chosen later, and define the random variables
\bieee
\label{eq:BLi}
L_1,L_2,\ldots,L_{\dK} \;\,\text{i.i.d. $\sim\mu(\ell)$} \;.
\eieee
Fix $s^n$, and define the random variables
\bieee
\Omega_j(s^n)= \cerr(\code_{L_j}) \;,\quad j\in [1:\dK] \;,
\eieee
which is the conditional probability of error of the code $\code_{L_j}$ given the state sequence $s^n$. 

Since $\code^\Gamma$ is a $(2^{nR_1},2^{nR_2},n,\eps_n)$ code, we have that 
$ \sum_\gamma\mu(\gamma)\sum_{s^n} \qn(s^n) \cerr(\code_\gamma)\leq \eps_n$, for all $\qn(s^n)$. In particular, for a kernel, we have that  
\bieee
\label{eq:BPsiIneqC}
\E \Omega_j(s^n)=\sum_{\gamma\in\Gamma} \mu(\gamma)\cdot \cerr(\code_\gamma) \leq \eps_n \;,
\eieee
for all $j\in[1:k]$.

Now take $n$ to be large enough so that $\eps_n<\alpha$. 
Keeping $s^n$ fixed,  we have that the random variables $\Omega_j(s^n)$ are i.i.d., due to (\ref{eq:BLi}).
{ Next the technique known as Bernstein's trick \cite{Ahlswede:78p} is applied. }
\bieee
\prob{\sum_{j=1}^{\dK} \Omega_j(s^n)\geq \dK\alpha} &\stackrel{(a)}{\leq}&
\E\left\{ \exp\left[ \beta \left(
\sum_{j=1}^\dK \Omega_j(s^n)- k\alpha
\right)
\right]
\right\} \\
&=& e^{-\beta \dK \alpha}\cdot
\E\left\{  \prod_{j=1}^\dK e^{\beta  \Omega_j(s^n)} \right\} \\
&\stackrel{(b)}{=}& e^{-\beta \dK \alpha}\cdot
\prod_{j=1}^\dK \E\left\{   e^{\beta  \Omega_j(s^n)} \right\} \\
&\stackrel{(c)}{\leq}& e^{-\beta \dK \alpha}\cdot
\prod_{j=1}^\dK \E\left\{   1+e^\beta \cdot\Omega_j(s^n) \right\} \\
&\stackrel{(d)}{\leq}& e^{-\beta \dK \alpha}\cdot
 \left(   1+e^\beta \eps_n \right)^\dK 
\eieee
where $(a)$ is an application of Chernoff's inequality; $(b)$ follows from the fact that $\Omega_j(s^n)$ are independent; $(c)$ holds since $e^{\beta x}\leq1+e^\beta x$, for $\beta>0$ and $0\leq x\leq 1$; $(d)$ follows from (\ref{eq:BPsiIneqC}). We take $n$ to be large enough for $1+e^\beta \eps_n\leq e^\alpha$ to hold. Thus, choosing $\beta=2$, we have that 
\bieee
\prob{\frac{1}{\dK}\sum_{j=1}^\dK \Omega_j(s^n)\geq \alpha} &\leq&
e^{- \alpha k } \;,
\eieee
for all $s^n\in\Sset^n$. Now, by the union of events bound, we have that
\bieee
\prob{
\max_{s^n}  \frac{1}{\dK} \sum_{j=1}^\dK \Omega_j(s^n) \,\geq\, \alpha
 }&=&\prob{
\exists s^n :  \frac{1}{\dK} \sum_{j=1}^\dK \Omega_j(s^n) \,\geq\, \alpha
 }\\
&\leq&\sum_{s^n\in\Sset^n} \prob{  \frac{1}{\dK} \sum_{j=1}^\dK \Omega_j(s^n) \,\geq\, \alpha
 }\\
&\leq& |\Sset|^n \cdot e^{-\alpha \dK} \;.
\label{eq:BeBound}
\eieee
Since $|\Sset|^n$ grows only exponentially in $n$, choosing $k=n^2$ results in a super exponential decay. 

Consider the code $\code^{\Gamma^*}=(\mu^*,\Gamma^*=[1:k],\{\code_{L_j}\}_{j=1}^k)$ formed by a random collection of codes, with $\mu^*(j)=\frac{1}{k}$.  It follows 
 that  the conditional  probability of error given  $s^n$, which is given by  
\bieee 
\cerr(\code^{\Gamma^*})=
\frac{1}{k} \sum_{j=1}^{k} \cerr(\code_{L_j}) \;,
\eieee
 exceeds $\alpha$ with a super exponentially small probability $\sim e^{-\alpha n^2}$, for all $s^n\in\Sset^n$.
 Thus, there exists a random code $\code^{\Gamma^*}=(\mu^*,\Gamma^*,\{\code_{\gamma_j}\}_{j=1}^k)$ 
 for the AVBC $\avbc$, such that 
\bieee
\err(\qn, \code^{\Gamma^*})=\sum_{s^n\in\Sset^n} \qn(s^n) \cerr(\code^{\Gamma^*})\leq \alpha \;,\quad\text{for all $\qn(s^n)
\in\pSpace(\Sset^n)$}\;.
\eieee

\qed

\section{Proof of Theorem~\ref{theo:BcorrTOdetC}}
\label{app:BcorrTOdetC}
\begin{proof}[Achievability proof]
 To show achievability, 
we follow the lines of \cite{Ahlswede:78p}, with the required adjustments.
  We use the random code constructed in the proof of  Theorem~\ref{theo:Bmain} 
	to construct a deterministic code.

Let $(R_1,R_2)\in\BrCav$, and consider the case where $\interior{\BCavc}\neq \emptyset$. Namely,
\bieee
\label{eq:Brpos}
\opC(\avc_1)\geq \opC(\avc_2)>0 \;,
\eieee
where $\avc_1=\{\sbc\}$ and $\avc_2=\{\wbc\}$ denote the marginal AVCs with causal SI of the stronger user and the weaker user, respectively.  
By Lemma~\ref{lemm:BcorrSizeC},  for every $\eps_1>0$ and sufficiently large $n$, 
 there exists a $(2^{nR_1},2^{nR_2},n,\eps_1)$ random  code  
$
\code^\Gamma=\big(\mu(\gamma)=\frac{1}{k},\Gamma=[1:k],\{\code_\gamma \}_{\gamma\in \Gamma}\big) 
$, 
where
$\code_\gamma=(\encn_\gamma,\dec_{1,\gamma},\dec_{2,\gamma})$, 
for $\gamma\in\Gamma$, 
 and 
$
k=|\Gamma|\leq n^2 
$. 
Following (\ref{eq:Brpos}),  we have that for every $\eps_2>0$ and sufficiently large $\nu$, the code index $\gamma\in [1:k]$ can be sent over $\avbc$ using a $(2^{\nu\bR_1},2^{\nu\bR_2},\nu,\eps_2)$ deterministic code 
$ 
\code_{\text{i}}=(\tfnu,\gnu_1,\gnu_2)  
$, where $\bR_1>0$, $\bR_2>0$.
Since $k$ is at most polynomial, 
 the encoder can reliably convey $\gamma$ to the receiver with a negligible blocklength, \ie
$ 
\nu=o(n) 
$. 

Now, consider a code 
 formed by the concatenation of $\code_{\text{i}}$ as a prefix to a corresponding code in the code collection $\{\code_\gamma\}_{\gamma\in\Gamma}$. 
That is, the encoder sends both the index $\gamma$ and the message pair $(m_1,m_2)$ to the receivers, such that 
 the index $\gamma$ is transmitted first by $\tfnu(\gamma,s^\nu)$, and then the message pair $(m_1,m_2)$ is transmitted by the codeword $x^n=\enc_\gamma^n($ $m_1,m_2,$ $s_{\nu+1},\ldots,s_{\nu+n})$.
Subsequently, decoding is performed in two stages as well; decoder 1 estimates the index at first, with  
$\hgamma_1=$ $\gnu_1(y_{1,1},\ldots,$ $y_{1,\nu})$, and the message $m_1$ is then estimated by  
$\widehat{m}_1=$ $g_{1,\hgamma_1}(y_{1,\nu+1},$ $\ldots,y_{1,\nu+n})$.  
Similarly, decoder 2 estimates the index  with  
$\hgamma_2=$ $\gnu_2(y_{2,1},$ $\ldots,y_{2,\nu})$, and the message $m_2$ is then estimated by  
$\widehat{m}_2=$ $g_{2,\hgamma_2}(y_{2,\nu+1},\ldots,$ $y_{2,\nu+n})$.

 By the union of events bound, the probability of error 
 is then bounded by $\eps=\eps_1+\eps_2$, 
for every joint distribution 
in $\pSpace^{\nu+n}(\Sset^{\nu+n})$. 
That is, the concatenated code 
 is a $(2^{(\nu+n)\tR_{1,n}},2^{(\nu+n)\tR_{2,n}},\nu+n,\eps)$ code over the AVDBC $\avbc$ with causal SI, where $\nu=o(n)$. 
Hence, 
  the blocklength is $n+o(n)$, and the 
the rates   $\tR_{1,n}=\frac{n}{\nu+n}\cdot R_1$ and $\tR_{2,n}=\frac{n}{\nu+n}\cdot R_2$ approach $R_1$ and 
$R_2$, respectively, as $n\rightarrow \infty$. 
\end{proof}

\begin{proof}[Converse proof]
In general, the deterministic code capacity region is included within the random code capacity region. Namely,
$\BCavc\subseteq\BrCav$. 
\end{proof} 

\section{Proof of Corollary~\ref{coro:BmainDbound}}
\label{app:BmainDbound}
First, consider the inner and outer bounds in (\ref{eq:BmainInner}) and (\ref{eq:BmainOuter}).
The bounds are obtained as a  direct consequence of part 1 of Theorem~\ref{theo:Bmain} and  Theorem~\ref{theo:BcorrTOdetC}. Note that the outer bound (\ref{eq:BmainOuter}) holds regardless of any condition, since the deterministic code capacity region is always included within the random code capacity region, \ie $\BCavc \subseteq \BrCav\subseteq \BrICav$.

Now, suppose that the marginal $V_{Y_2|U_2,S}^{\encs'}$ is non-symmetrizable for some $\encs':
\Uset_2\times\Sset\rightarrow\Xset$, and the condition $\sCond$ holds. 
 Then, by part 2 of Theorem~\ref{theo:ALavcCstateC}, 
the capacity of the corresponding single-user AVC is positive, \ie $\opC(\avc_2)>0$.
Since the AVDBC $\avc$ is assumed to be degraded, we then have that $\opC(\avc_1)\geq\opC(\avc_2)>0$, which means that
 $\interior{\BCavc}\neq\emptyset$. 
%
 Hence, by Theorem~\ref{theo:BcorrTOdetC}, the deterministic code capacity region coincides with the random code capacity region, \ie $\BCavc=\BrCav$. Then, the proof follows from part 2 of Theorem~\ref{theo:Bmain}. \qed 


\section{Analysis of Example~\ref{example:AVBSBC}}
\label{app:AVBSBC}
We begin with the case of an arbitrarily varying BSBC $\Bavcig$ without SI. 
We claim that the single user marginal AVC $\avc_{1,0}$  without SI, corresponding to the stronger user, has zero capacity.
Denote $q\triangleq q(1)=1-q(0)$. 
Then, observe that the additive noise is distributed according to $Z_S\sim\text{Bernoulli}(\eps_q)$, with
$
\eps_q\triangleq (1-q)\cdot \theta_0+q\cdot \theta_1 
$, 
for $0\leq q\leq 1$. 
By Theorem~\ref{theo:avcC0R}, $\opC(\avc_{1,0})\leq\opC^{\rstarC}\hspace{-0.1cm}(\avc_{1,0})= \min_{0\leq q\leq 1} [1-h(\eps_q)]$.
Since $\theta_0< \frac{1}{2}\leq \theta_1$, there exists $0\leq q\leq 1$
such that $\eps_q=\frac{1}{2}$, thus $\opC(\avc_{1,0})=0$.
%
The capacity region of the AVDBC $\Bavcig$ without SI is then given by  
$\BCavcig=\{(0,0)\}$.

Now, consider the arbitrarily varying BSBC $\avbc$ with causal SI. 
 By Theorem~\ref{theo:Bmain}, 
 the random code capacity region is bounded by 
$
\BIRavc \subseteq \BrCav\subseteq \BrICav 
$. 
We show that the bounds coincide, and are thus tight. 
Let $\Brp$ denote the DBC $\bc$ with causal SI, governed by  an i.i.d. state sequence, distributed according to $S\sim\text{Bernoulli}(q)$.
By \cite{Steinberg:05p}, the corresponding capacity region 
 is given by
\begin{subequations}
\label{eq:Bex1ICrpT}
\begin{align}
\label{eq:Bex1ICrpa}
&\BICrp=
\bigcup_{0\leq \beta\leq 1}
\left\{
\begin{array}{lll}
(R_1,R_2) \,:\; & R_2 &\leq   1-h(\alpha*\beta*\delta_q) \;, \\
								& R_1 &\leq   h(\beta*\delta_q)-h(\delta_q)
\end{array}
\right\}
 \,,
\end{align}
where
 \bieee
\label{eq:ex1deltaq}
\delta_q\triangleq (1-q)\cdot \theta_0+q\cdot (1-\theta_1) \;,
\eieee
\end{subequations}
 for $0\leq q\leq 1$. For every given $0\leq q'\leq 1$, we have that 
$
\BrICav=\bigcap_{0\leq q\leq 1} \BICrp \subseteq 
\inC(\avbc^{q'}) 
$. 
 Thus, taking $q'=1$, we have that
\bieee
\label{eq:ex1out}
\BrICav
 \subseteq 
\bigcup_{0\leq \beta\leq \frac{1}{2}}
\left\{
\begin{array}{lll}
(R_1,R_2) \,:\; & R_2 &\leq   1-h(\alpha*\beta*\theta_1) \;, \\
								& R_1 &\leq   h(\beta*\theta_1)-h(\theta_1)
\end{array}
\right\}\;,
\eieee
where we have used the identity 
$h(\alpha*(1-\delta))=h(\alpha*\delta)$. 

Now, to show that the region above is achievable, we examine the inner bound,
\bieee
\BIRavc=
\bigcup_{p(u_1,u_2),\encs(u_1,u_2,s)}\, 
\left\{
\begin{array}{lll}
(R_1,R_2) \,:\; & R_2 &\leq \min_{0\leq q\leq 1}   I_q(U_2;Y_2) \;, \\
								& R_1 &\leq \min_{0\leq q\leq 1}  I_q(U_1;Y_1|U_2)  
\end{array}
\right\}\;.
\eieee
Consider the following choice of $p(u_1,u_2)$ and $\encs(u_1,u_2,s)$. Let $U_1$ and $U_2$ be independent random variables, 
\bieee
\label{eq:BSBCdistAchieve}
U_1\sim\text{Bernoulli}(\beta) \,,\;\text{and}\;\, U_2\sim\text{Bernoulli}\left(\frac{1}{2} \right)\;,
\eieee
for $0\leq \beta\leq\frac{1}{2}$, 
 and let 
\bieee
\label{eq:BSBCxiAchieve}
\encs(u_1,u_2,s)=u_1+u_2+s \mod 2\;.
\eieee
Then,
\begin{align}
 &H_q(Y_1|U_1,U_2)=H_q(S+Z_S)=h(\delta_q) \;,																\nonumber\\ 
 &H_q(Y_1|U_2)=H_q(U_1+S+Z_S)=h(\beta*\delta_q) \;,													\nonumber\\
 &H_q(Y_2|U_2)=H_q(U_1+S+Z_S+V)=h(\alpha*\beta*\delta_q) \;,								\nonumber\\ 
 &H_q(Y_2)=1 \;, 
\intertext{
where addition is modulo $2$, and $\delta_q$ is given by (\ref{eq:ex1deltaq}). 
Thus,
}
&I_q(U_2;Y_2)=1-h(\alpha*\beta*\delta_q) \;, \nonumber\\
&I_q(U_1;Y_1|U_2)=h(\beta*\delta_q)-h(\delta_q) \;,
\end{align}
hence 
\bieee
\label{ex1:innerR}
\BIRavc \supseteq
\bigcup_{0\leq \beta\leq \frac{1}{2}}\, 
\left\{
\begin{array}{lll}
(R_1,R_2) \,:\; & R_2 &\leq \min_{0\leq q\leq 1}  1-h(\alpha*\beta*\delta_q) \;, \\
								& R_1 &\leq \min_{0\leq q\leq 1}  h(\beta*\delta_q)-h(\delta_q)  
\end{array}
\right\}\;.
\eieee
Note that $\theta_0\leq \delta_q\leq 1-\theta_1 \leq\frac{1}{2}$. For $0\leq\delta\leq\frac{1}{2}$, the functions 
$g_1(\delta)=1-h(\alpha*\beta*\delta)$ and $g_2(\delta)= h(\beta*\delta)-h(\delta)$ are monotonic decreasing functions of $\delta$, hence the minima in (\ref{ex1:innerR}) are both achieved with $q=1$.
It follows that 
	\begin{align}
	\label{eq:ex1BrCav}
\BrCav=\BIRavc=\BrICav=
\bigcup_{0\leq \beta\leq 1}\, 
\left\{
\begin{array}{lll}
(R_1,R_2) \,:\; & R_2 &\leq   1-h(\alpha*\beta*\theta_1) \;, \\
								& R_1 &\leq h(\beta*\theta_1)-h(\theta_1)  
\end{array}
\right\}\;.
\end{align}

It can also be verified that the condition $\sCond$ holds (see Definition~\ref{def:sCond}), in agreement with part 2 of Theorem~\ref{theo:Bmain}. 
First, we specify a function $\encs(u_1,u_2,s)$ and a distributions set 
$\Dset^{\rstarC}$ that achieve $\BIRavc$ and $\BrICav$ (see Definition~\ref{eq:Bachieve}).
 Let $\encs(u_1,u_2,s)$ be as in (\ref{eq:BSBCxiAchieve}), and let $\Dset^{\rstarC}$ be the set of distributions $p(u_1,u_2)$ such that $U_1$ and $U_2$ are independent random variables, distributed according to  (\ref{eq:BSBCdistAchieve}). 
By the derivation above,  the requirement (\ref{eq:BIRachieve}) is satisfied. Now, by the derivation in 
\cite[Section IV]{Steinberg:05p}, we have that
\bieee
\BICrp=\bigcup_{p(u_1,u_2)\in\Dset^{\rstarC}} 
\left\{
\begin{array}{lll}
(R_1,R_2) \,:\; & R_2 &\leq    I_q(U_2;Y_2) \;, \\
								& R_1 &\leq   I_q(U_1;Y_1|U_2)  
\end{array}
\right\}\;.
\eieee
Then, the requirement (\ref{eq:BICachieve}) is satisfied as well, hence $\encs(u_1,u_2,s)$  and  $\Dset^{\rstarC}$
achieve $\BIRavc$ and $\BrICav$. It follows that condition $\sCond$ holds, as $q^*=1$ satisfies the desired property 
 with $\encs(u_1,u_2,s)$ and $\Dset^{\rstarC}$ as described above. 
 

We move to the deterministic code capacity region of the arbitrarily varying BSBC $\avbc$ with causal SI. 
If  $\theta_1=\frac{1}{2}$,  the capacity region is given by $\BCavc=\BrCav=\{(0,0)\}$, by (\ref{eq:ex1BrCav}). 
Otherwise, $\theta_0< \frac{1}{2}< \theta_1$, and we now  show  that the condition in Corollary~\ref{coro:BmainDbound} is met.
 Suppose that  $V^{\encs'}_{Y_2|U,S}$ is symmetrizable for all $\encs':\Uset_2\times\Sset\rightarrow\Xset$.
That is, for every $\encs'(u_2,s)$,  
there exists $\lambda_{u_2}=J(1|u_2)$ such that
\begin{multline}
(1-\lambda_{u_{b}})W_{Y_2|X,S}(y_2|\encs'(u_{a},0),0)+\lambda_{u_{b}}W_{Y_2|X,S}(y_2|\encs'(u_{a},1),1)    =\\
(1-\lambda_{u_{a}})W_{Y_2|X,S}(y_2|\encs'(u_{b},0),0)+\lambda_{u_{a}}W_{Y_2|X,S}(y_2|\encs'(u_{b},1),1)
\end{multline}
for all  $u_{a},u_{b}\in\Uset_2$, $y_2\in\{0,1\}$. 
If this is the case, then for $\encs'(u_2,s)=u_2+s \mod 2$, taking $u_{a}=0$, $u_{b}=1$, $y_2=1$, we have that
\begin{align}
\label{eq:BSCfair}
(1-\lambda_{1})\cdot(\alpha*\theta_0)+\lambda_{1}\cdot(1-\alpha*\theta_1)= (1-\lambda_{0})\cdot(1-\alpha*\theta_0)+\lambda_{0}\cdot(\alpha*\theta_1) \;.
\end{align}
This is a contradiction. 
 Since $f(\theta)=\alpha*\theta$ is a monotonic increasing function of $\theta$, and since $1-f(\theta)=f(1-\theta)$, we have that  the value of the LHS of (\ref{eq:BSCfair}) is in $[0,\frac{1}{2})$, while the value of the RHS of (\ref{eq:BSCfair}) is in $(\frac{1}{2},1]$.
 Thus, there exists 
$\encs':\Uset_2\times\Sset\rightarrow\Xset$ such that $V^{\encs'}_{Y_2|X,S}$ is non-symmetrizable for $\theta_0< \frac{1}{2}< \theta_1$. 
As the condition $\sCond$ holds, we have that $\BCavc=\BIRavc=\BrICav$, due to Corollary~\ref{coro:BmainDbound}. 
Hence, by (\ref{eq:ex1BrCav}), we have that the capacity region of the arbitrarily varying BSBC $\avbc$ with causal SI is given by  (\ref{eq:Bex1Cavc}). \qed

\end{appendices}

\vspace{-0.35cm}
\printbibliography
 
\end{document}